%% file: document.tex
\documentclass[a4paper,USenglish]{lipics-v2016} 
\usepackage[T1]{fontenc}
\usepackage{graphicx}
\usepackage[fleqn]{amsmath}
\usepackage{amssymb}
\usepackage{array}
\usepackage[fleqn]{mathtools}
\usepackage{listings}
\usepackage[usenames,dvipsnames]{xcolor}
\usepackage{tikz}
\usepackage{booktabs}
\usepackage{microtype}
\usepackage{multirow}
\usepackage{multicol}
\raggedcolumns 
\usepackage{flushend}
\usepackage{stmaryrd}
\usepackage[T1]{fontenc}
\usepackage{xspace}
\usepackage{xparse}
\usepackage{subdepth}
\usepackage{wrapfig}
\usepackage[inference]{semantic}
\usepackage{scalerel}
\usepackage{enumerate}
\usetikzlibrary{positioning,chains,shapes.arrows,shapes.geometric,fit,calc,arrows,decorations.pathmorphing}

\input{cocomacros}

\input{macros}
\input{lstconfig}

\input{tikz}

\renewcommand{\paragraph}[1]{\vspace{1ex}\noindent\textbf{#1}}

\begin{document}

\title{A Co-contextual Type Checker\newline for Featherweight Java (incl. Proofs)}
\titlerunning{A Co-contextual Type Checker for Featherweight Java (incl. Proofs)}

\author[1]{Edlira Kuci}
\author[2]{Sebastian Erdweg}
\author[1]{Oliver Bra\v{c}evac}
\author[1]{Andi Bejleri}
\author[1,3]{Mira Mezini}
\affil[1]{Technische Universit{\"a}t Darmstadt, Germany} 
\affil[2]{TU Delft, The Netherlands}
\affil[3]{Lancaster University, UK}

\authorrunning{E. Kuci, S. Erdweg, O. Bra\v{c}evac, A. Bejleri, and M. Mezini} 

\Copyright{Edlira Kuci, Sebastian Erdweg, Oliver Bra\v{c}evac, Andi Bejleri, and Mira Mezini}


\maketitle

\input{Abstract}
\input{Introduction}
\input{Section_2_background_motivation}
\input{Section_3_Co-Contx_FJava}
\input{Section_4_co_fj_type_rules}
\input{Section_5_theorem}
\input{Section_6_implementation}
\input{Section_7_evaluation}

\input{related_work}
\input{Conclusion}

\paragraph{Acknowledgments.}
This work has been supported by the European Research Council, grant No. 321217.

\bibliographystyle{plainurl} 
\bibliography{bib}

\clearpage
\appendix
\input{theorems+defs}

\end{document}

%% file: cocomacros.tex
\newcommand{\implementationurl}{\url{https://github.com/seba--/incremental}}

\newcommand{\OR}{\ensuremath{\;|\;}}
\newcommand{\WHERE}{\OR}
\newcommand{\ctxcolor}{\textblack}
\newcommand{\typecolor}{\textblack}

\newcommand{\typectxcolor}[1]{\textblack{\ensuremath{#1}}}

\newcommand{\ofType}{\ensuremath{\mathop{:}}}
\renewcommand{\to}{\mathop{\rightarrow}}
\newcommand{\ctxplus}[2]{\ensuremath{#1; #2}}

\newcommand{\dom}[1]{\func{dom}(#1)}

\newcommand{\reqs}[0]{\ensuremath{R}}

\newcommand{\ctxNode}[1]{\ensuremath{#1 \mathop{\vdash}\, }}
\newcommand{\typeNode}[1]{\ensuremath{\,\mathop{:} #1}}

\newcommand{\coFNode}[2]{\ensuremath{\mathrel{|} #1 \mathrel{|} #2}}

\newcommand{\judge}[3]{\ensuremath{#1 \mathrel{\vdash} #2 \mathrel{:} #3}}
\newcommand{\judgeM}[2]{\ensuremath{#1 \mathrel{\vdash} #2 }}
\newcommand{\judgeP}[1]{\ensuremath{#1 }}

\newcommand{\judgeOK}[3]{\ensuremath{#1 \mathrel{\vdash} #2 \mathrel{|} \typecolor{#3}}}
\newcommand{\cojudge}[5]{\ensuremath{#1 \mathrel{:} \typecolor{#2} \mathrel{|} \typecolor{#3} \mathrel{|} \ctxcolor{#4} \mathrel{|} \ctxcolor{#5}}}

\newcommand{\cojudgeOM}[4]{\ensuremath{#1\mathrel{|} \typecolor{#2} \mathrel{|} \ctxcolor{#3} \mathrel{|} \ctxcolor{#4}}}
\newcommand{\cojudgeOC}[3]{\ensuremath{#1 \mathrel{|} \typecolor{#2} \mathrel{|} \ctxcolor{#3}}}
\newcommand{\judgeOP}[2]{\ensuremath{#1 \mathrel{|} \typecolor{#2}}}

\newcommand{\mustEqual}[0]{\ensuremath{\overset{!}{=}}}
\newcommand{\rulename}[1]{\ensuremath{\textsc{\scriptsize #1}}\xspace}
\newcommand{\fresh}[1]{\ensuremath{#1\ \term{is\ fresh}}}
\newcommand{\afresh}[1]{\ensuremath{#1\ \term{are\ fresh}}}

\newcommand{\smap}{\ensuremath{\mathop{\!\mapsto\!}}}

\newcommand{\subtpe}{\ensuremath{\mathbin{\!\scalerel{{<}\!}{\colon}\!\!}}}

\newcommand{\domC}{\ensuremath{dom(CT)}}
\newcommand{\satisfy}[0]{\ensuremath{\mathit{satisfies}}}

\newcommand{\func}[1]{\ensuremath{\operatorname{#1}}}
\newcommand{\cname}[1]{\ensuremath{\mathsf{#1}}}
\newcommand{\seq}[1]{\ensuremath{\overline{#1}}}

\theoremstyle{prop}
\newtheorem{prop}[theorem]{Proposition}

\theoremstyle{assump}
\newtheorem{assump}[theorem]{Assumption}

\newcommand{\extendsCT}[2]{\ensuremath{#1 \mathrel{\keyw{extends}} #2}}
\newcommand{\constrCT}[2]{\ensuremath{#1.\keyw{init}(#2)}}

\newcommand{\cond}[0]{\ensuremath{\mathit{cond}}}
\newcommand{\extendsCR}[2]{\ensuremath{#1 \mathrel{.} \mathrel{\keyw{extends}} \mathrel{:} #2}}

\newcommand{\projExt}[1]{\func{projExt}(#1)}
\newcommand{\term}[1]{\ensuremath{\mathrm{#1}}}

\newcommand{\keyw}[1]{\ensuremath{\mathbf{#1}}}


%% file: macros.tex
\newcolumntype{C}[1]{>{\centering\arraybackslash}m{#1}} 
\newcolumntype{P}[1]{>{\arraybackslash}m{#1}} 

\makeatletter
\newcommand*{\mycleardoublepage}{\clearpage\if@twoside
\ifodd\c@page
  \hbox{}
  \clearpage
\fi
\hbox{}
\thispagestyle{empty}
\clearpage
\fi}
\makeatother

\newcounter{MOPReqCounter}

\colorlet{myblue}{blue}
\colorlet{myred}{red}
\colorlet{mypurple}{violet!150!black}
\colorlet{mygreen}{green!40!black}
\newcommand{\textblack}{\textcolor{black}}

\newenvironment{indented}[1][]%
{\noindent
 \begin{itemize}
 \item[#1]}%
{\end{itemize}}

\newcounter{goal}{}
\makeatletter
\def\nvphantom{\v@true\h@false\nph@nt}
\def\nhphantom{\v@false\h@true\nph@nt}
\def\nphantom{\v@true\h@true\nph@nt}
\def\nph@nt{\ifmmode\def\next{\mathpalette\nmathph@nt}%
  \else\let\next\nmakeph@nt\fi\next}
\def\nmakeph@nt#1{\setbox\z@\hbox{#1}\nfinph@nt}
\def\nmathph@nt#1#2{\setbox\z@\hbox{$\m@th#1{#2}$}\nfinph@nt}
\def\nfinph@nt{\setbox\tw@\null
  \ifv@ \ht\tw@\ht\z@ \dp\tw@\dp\z@\fi
  \ifh@ \wd\tw@-\wd\z@\fi \box\tw@}
\makeatother


%% file: lstconfig.tex
\definecolor{keyword}{RGB}{96,0,53}
\definecolor{darkgreen}{RGB}{0,128,0}
\definecolor{listingsFrame}{gray}{0.6}

\lstdefinestyle{normal}{%
  xleftmargin=\parindent,
}

\lstdefinestyle{nospace}{%
  aboveskip=0em,
  belowskip=0em,
  xleftmargin=0em,
  xrightmargin=0em,
}

\lstdefinestyle{nohspace}{%
  xleftmargin=0em,
  xrightmargin=0em,
}

\lstdefinestyle{scriptsize}{%
  basicstyle=\scriptsize\sf,
}

\lstdefinestyle{footnote}{%
  basicstyle=\footnotesize\sf,
}

\lstdefinestyle{figureframe}{%
  xleftmargin=0pt,
  frame=single,
  rulecolor=\color{listingsFrame},
  numbersep=6pt,
}
\lstdefinestyle{inlineframe}{%
  xleftmargin=2pt,
  xrightmargin=2pt,
  frame=single,
  rulecolor=\color{listingsFrame},
  numbersep=6pt,
}

\newcommand{\mytilde}{%
  \texttt{\resizebox{.48em}{1ex}{\hbox{$\sim$}}}%
}

\lstdefinelanguage{LaTeX}[]{Tex}{%
  language=Tex,
  morekeywords={emph, label, ref, begin, end, part, chapter, section,
                subsection, subsubsection, paragraph, subparagraph, cite},
}

\lstdefinelanguage{FJ}[]{Java}{%
  moredelim=[is][\sffamily]{!}{!},
  moredelim=[is][\rmfamily\itshape]{'}{'}
}

\lstdefinelanguage{SugarJ}[]{Java}{%
  language=Java,
  morekeywords={sugar, extension, context, free, syntax, desugarings, sorts, signature,
    constructors, rules, strategies, assert, as, editor, services, colorer, folding,
    outliner, checks, recursive, errors, warnings, where,
    css, outlining, rec, color, completions, completion, template,
    layout, analyses, then},
  mathescape=true,
  deletestring=[b]',
  morecomment=[l]{//},
}

\lstdefinelanguage{SugarJXML}[]{SugarJ}{%
  deletecomment=[l]{//},
  morecomment=[s]{<!--}{-->}
}

\lstdefinelanguage{SugarMDD}[]{SugarJ}{%
  morekeywords={model, transformation},
}

\lstdefinelanguage{SugarJATM}[]{SugarMDD}{%
  morekeywords={statemachine, initial, state, events},
}

\lstdefinelanguage{SugarJEntity}[]{SugarMDD}{%
  morekeywords={entity},
}
\lstdefinelanguage{SugarJATMEntity}[]{SugarJATM}{%
  morekeywords={data, entity},
}

\lstdefinelanguage{SugarJTemplate}[]{SugarMDD}{%
  morekeywords={in, \$for, template},
  mathescape=false
}

\lstdefinelanguage{SugarJFeature}[]{SugarMDD}{%
  morekeywords={featuremodel, features, constraint, config, enable, disable, variable},
  otherkeywords={\#ifdef},
}

\lstdefinelanguage{scala}{
  morekeywords={abstract,case,catch,class,def,%
    do,else,extends,false,final,finally,%
    for,if,implicit,import,match,mixin,%
    new,null,object,override,package,%
    private,protected,requires,return,sealed,%
    super,this,throw,trait,true,try,%
    type,val,var,while,with,yield},
  sensitive=true,
  morecomment=[l]{//},
  morecomment=[n]{/*}{*/},
  morestring=[b]",
  morestring=[b]',
  morestring=[b]"""
}

\lstdefinelanguage{SDF}{%
  morekeywords={context, free, syntax, sorts, signature,
    constructors, rules, strategies, module, imports, exports},
  escapechar=_,
  mathescape=true,
  comment=[l]{\%\%},
  morestring=[b]",
}

\lstdefinelanguage{MyPython}[]{Python}{%
  escapechar=\%,
  mathescape=true,
}

\lstdefinelanguage{MyHaskell}[]{Haskell}{%
  escapechar=\%,
  mathescape=true,
  deletekeywords={Nothing,Just,False,True,putStrLn,fail,fromJust,lookup,Num,exp,free,snd,String,
  return,error,otherwise,not,show,read,Eval,Read,readsPrec,print},
}

\lstdefinelanguage{SugarHaskell}[]{MyHaskell}{%
  morekeywords={context, free, syntax, desugarings, sorts, signature,
    constructors, rules, strategies, lexical, reject},
  mathescape=false,
}

\lstdefinelanguage{SugarHaskellArrows}[]{SugarHaskell}{%
  morekeywords={proc},
}

\lstdefinelanguage{EBNF}{
  morestring=[b]"
}

\lstdefinelanguage{Constraint}{
}

\lstdefinelanguage{Plain}{}

\lstdefinelanguage{Pseudo}{
  keywords={foreach,match,case,in,return,if,else},
  mathescape=true,
}

\lstdefinelanguage{Questionnaire}[]{Java}{%
  morekeywords = {questionnaire, question, value, Boolean, String, Integer, group,
    if, else, define, ask}
}

\lstdefinelanguage{SugarFomega}{
  keywords = {module, val, type, mu, if, then, else, case, of,
    fold, unfold, true, false, as, public, import, syntax, desugaring, typing, context, free, let, in, end, forall, do},
  mathescape = true,
  morestring=[b]",
}


%% file: tikz.tex
\usetikzlibrary{positioning,chains,fit,calc,arrows,decorations.pathmorphing}
\usetikzlibrary{shapes.arrows,shapes.geometric,shapes.symbols}

\tikzstyle{invisible} = []

\tikzstyle{model}
  = [ shape=rectangle
    , draw
    ]
\tikzstyle{transformation}
  = [ model
    , shape=single arrow
    , draw
    ]

\tikzstyle{meta} = [ double ]
\tikzstyle{generated} = [ dashed ]

\tikzstyle{mopdependency}
  = [ -stealth'
    , semithick
    ]
\tikzstyle{instance}
  = [ mopdependency
    , -open triangle 60
    ]

\tikzstyle{document}
  = [ shape=rectangle
    , draw
    , minimum height=1.5em
    , text width=5em
    , text centered
    ]

\tikzstyle{component}
  = [ font=\scriptsize\it
    ]

\tikzstyle{code}
  = [ shape=rectangle
    , draw
    ]

\tikzstyle{process}
  = [ shape=single arrow
    , single arrow head extend=.75em
    , single arrow head indent=.25em
    , minimum width=3em
    , draw
    ]

\tikzstyle{point}
  = [ coordinate
    , minimum width=1em
    ]

\tikzstyle{flow diagram}
  = [ start chain
    , node distance=1em
    , every node/.style={on chain}
    ]

\tikzstyle{ast}
  = [ node distance=1em and 0.38em
    , every node/.style=ast node
    ]

\tikzstyle{ast node}
  = [ shape=circle
    , minimum size=0.5em
    , inner sep=0
    , fill
    ]

\tikzstyle{dependency}
  = [ dashed
    , -stealth'
    ]

\tikzstyle{red}
  = [ shape=rectangle
    , color=red
    ]

\tikzstyle{blue}
  = [ shape=circle
    , color=blue
    ]

\tikzstyle{green}
  = [ shape=diamond
    , color=green!80!black
    , minimum size=0.7em
    ]

\tikzstyle{note}
  = [ font=\scriptsize\it
    ]
    
\tikzstyle{zoomed arrow}
  = [ solid
    , decorate
    , decoration=snake
    , -stealth'
    ]

\tikzstyle{zoom}
  = [ dashed
    ]

\tikzstyle{uml class}
  = [ shape=rectangle
    , draw
    , font=\sf
    ]

\tikzstyle{uml package}
  = [ uml class
    , inner sep=1em
    ]

\tikzstyle{uml dependency}
  = [ dependency, 
    , dashed
    ]

\tikzstyle{double arrow}
  = [ shape=double arrow
    , double arrow head extend=.75em
    , double arrow head indent=.25em
    , minimum width=3em
    , draw
    , font=\sf
    ]

\newdimen\nodedistance
\tikzstyle{name graph}
  = [ node distance=\nodedistance
    , every node/.style={namenode}
    , every path/.style={ref}
    ]

\tikzstyle{namenode}
  = [ circle
    , thick
    , anchor=center
    , draw
    , minimum size=2em
    , inner sep=2pt
    ]

\tikzstyle{synthesized}
  = [ namenode
    , fill=gray!45
    ]

\tikzstyle{ref}
  = [ -stealth'
    , semithick
    , draw
    ]

\tikzstyle{badref}
  = [ ref
    , dashed
    ]

\newdimen\nodedistance
\tikzstyle{exp tree}
  = [ node distance=\nodedistance
    , every node/.style={exp node}
    , every path/.style={link}
    , execute at begin node={\strut}
    ]

\tikzstyle{exp node}
  = [ anchor=center
    ]
\tikzstyle{link}
  = [ semithick
    , draw
    ]

\tikzstyle{ctx node}
  = [ node distance = -.6em,
    ]
\tikzstyle{ctx edge}
  = [ color = myred
    , -stealth'
    ]

\tikzstyle{type node}
  = [ node distance = -.5em,
    ]
\tikzstyle{type edge}
  = [ color = myblue
    , -stealth'
    ]

\newcommand{\distanceTop}{7.15pt}

\newcommand{\distanceBottom}{-2.55pt}

\newcommand{\distanceLeft}{-0.5pt}

\newcommand{\distanceRight}{0.5pt}

\newcommand{\drawrect}[4][]{\begin{tikzpicture}[remember picture, overlay]
\draw[layout box, #1]
  ($(#2) +(\distanceLeft, \distanceTop) + (-#4, #4)$) rectangle
  ($(#3) + (\distanceRight, \distanceBottom) + (#4, -#4)$);
\end{tikzpicture}}

\newcommand{\minirect}{\hbox to 9pt{\drawrect[thin,scale=0.3,black]{0pt, 15pt}{26pt, -5pt}{0pt}}}

\tikzstyle{layout box}
  = [ blue
    , semithick
    , draw
    ]

\tikzstyle{mini layout box}
  = [ black
    , very thin
    , draw
    ]

\tikzstyle{annotation}
  = [ font=\small\it
    ]

\tikzstyle{code annotation}
  = [ font=\small\it
    ]

\newcommand{\selectKeywordFont}{\bfseries}

\newcommand{\selectStringLitFont}{\tt}

\newcommand{\selectIdentifierFont}{\relax}

\newcommand{\selectOperatorFont}{\tt}

\newcommand{\tokenHeight}{2ex}

\newcommand{\tokenDepth}{.5ex}

\tikzstyle{pretty print}
  = [ start chain=going base right
    , text height=\tokenHeight
    , text depth=\tokenDepth
    , inner sep=0
    , node distance=0em
    ]

\tikzstyle{new line/helper}
  = [ on chain
    , anchor=north west,
      at={($(#1.west |- \tikzchainprevious.base) + (0, -1.1\baselineskip)$)}
    ]

\tikzstyle{new line tight/helper}
  = [ on chain
    , anchor=north west,
      at={($(#1.west |- \tikzchainprevious.base) + (0, -0.9\baselineskip)$)}
    ]

\tikzstyle{new line}[\tikzchainprevious]
  = [ on chain=placed {new line/helper=#1}
    ]

\tikzstyle{new line tight}[\tikzchainprevious]
  = [ on chain=placed {new line tight/helper=#1}
    ]

\tikzstyle{tab forward/helper}
  = [ on chain
    , anchor=north west,
      at={(#1 |- \tikzchainprevious.base)}
    ]

\tikzstyle{tab forward}
  = [ on chain=placed {tab forward/helper=#1}
    ]

\tikzstyle{indent}[1em]
  = [ on chain=placed {base right=#1 of \tikzchainprevious}
    ]

\tikzstyle{identifier}
  = [ on chain
    , font=\selectIdentifierFont
    ]

\tikzstyle{operator}
  = [ on chain
    , font=\selectOperatorFont
    ]

\tikzstyle{keyword}
  = [ on chain
    , font=\selectKeywordFont
    , text=keyword
    ]

\tikzstyle{stringlit}
  = [ on chain
    , font=\selectStringLitFont
    , text=blue
    ]

\newcommand{\Empty}{}
\newcommand{\ignore}[1]{\relax}

\makeatletter
\newcommand{\ppBox}[2][]{{
  \newcommand{\KW}[1]{
    \node[keyword] {##1};
    \HandleToken{\tikzchaincurrent}}

  \newcommand{\ID}[1]{
    \node[identifier] {##1};
    \HandleToken{\tikzchaincurrent}}

  \newcommand{\OP}[1]{
    \node[operator] {##1};
    \HandleToken{\tikzchaincurrent}}

  \newcommand{\SP}{
    \node[on chain] { };
    \HandleInsensitiveToken{\tikzchaincurrent}}

  \newcommand{\LB}[1]{
    \coordinate[on chain] (##1);}

  \newcommand{\TB}[1]{
    \coordinate[tab forward=##1];
  }

  \newcommand{\STR}[1]{
    \node[stringlit] {##1};
    \HandleToken{\tikzchaincurrent}}

  \renewcommand{\\}{
    \coordinate[new line=\NewlineNode];
    \let\HandleToken=\HandleTokenSubsequentLine}

  \newcommand{\newlineTight}{
    \coordinate[new line tight=\NewlineNode];
    \let\HandleToken=\HandleTokenSubsequentLine}

  \let\FirstToken=\Empty
  \newcommand{\AdjustFirst}[1]{
    \ifx\FirstToken\Empty
    \edef\FirstToken{##1}
    \fi
  }

  \let\LastToken=\Empty
  \newcommand{\AdjustLast}[1]{
    \edef\LastToken{##1}
  }

  \let\LeftToken=\Empty
  \newcommand{\AdjustLeft}[1]{
    \ifx\LeftToken\Empty
      \edef\LeftToken{##1}
    \else
      \pgf@process{\pgfpointanchor{##1}{west}}
      \setlength{\pgf@xa}{\pgf@x}
      \pgf@process{\pgfpointanchor{\LeftToken}{west}}
      \setlength{\pgf@xb}{\pgf@x}
      \ifdim\pgf@xa<\pgf@xb
        \edef\LeftToken{##1}
      \fi
    \fi
  }

  \let\RightToken=\Empty
  \newcommand{\AdjustRight}[1]{
    \ifx\RightToken\Empty
      \edef\RightToken{##1}
    \else
      \pgf@process{\pgfpointanchor{##1}{east}}
      \setlength{\pgf@xa}{\pgf@x}
      \pgf@process{\pgfpointanchor{\RightToken}{east}}
      \setlength{\pgf@xb}{\pgf@x}
      \ifdim\pgf@xa>\pgf@xb
        \edef\RightToken{##1}
      \fi
    \fi
  }

  \newcommand{\HandleTokenFirstLine}[1]{
    \AdjustFirst{##1}
    \AdjustRight{##1}
    \AdjustLast{##1}
  }

  \newcommand{\HandleTokenSubsequentLine}[1]{
    \AdjustLeft{##1}
    \AdjustRight{##1}
    \AdjustLast{##1}}

  \newcommand{\HandleInsensitiveToken}[1]{
  }

  \newcommand{\HandleToken}[1]{}

  \newcommand{\ppSubBox}[2][]{
    \coordinate[on chain];
    {
      \let\NewlineNode\tikzchaincurrent

      \let\FirstToken=\Empty
      \let\LastToken=\Empty
      \let\LeftToken=\Empty
      \let\RightToken=\Empty

      \let\HandleToken=\HandleTokenFirstLine

      ##2

      \ifx\LeftToken\Empty
        \path [##1]
          (\FirstToken.north west) rectangle
          (\LastToken.south east);
      \else
        \path [##1]
          (\RightToken.east |- \FirstToken.north) --
          (\FirstToken.north west) --
          (\FirstToken.south west) --
          (\LeftToken.west |- \FirstToken.south) --
          (\LeftToken.west |- \LastToken.south) --
          (\LastToken.south east) --
          (\LastToken.north east) --
          (\RightToken.east |- \LastToken.north) --
          cycle;
      \fi

      \global\let\InnerFirstToken=\FirstToken
      \global\let\InnerLastToken=\LastToken
      \global\let\InnerLeftToken=\LeftToken
      \global\let\InnerRightToken=\RightToken

      \global\let\InnerHandleToken=\HandleToken
    }
    \HandleToken{\InnerFirstToken}
    \ifx\HandleToken\HandleTokenFirstLine
    \let\HandleToken=\InnerHandleToken
    \fi
    \ifx\InnerLeftToken\Empty
    \else
    \HandleToken{\InnerLeftToken}
    \fi
    \HandleToken{\InnerRightToken}
    \HandleToken{\InnerLastToken}
  }

  \newcommand{\IN}[1]{\coordinate[indent]; ##1}
  \newcommand{\DE}[1]{\coordinate[indent=-1em]; ##1}

  \let\ppBox=\ppSubBox
  \ppBox[#1]{#2}}}
\makeatother


%% file: Abstract.tex
\begin{abstract}
  This paper addresses compositional and incremental type checking for object-oriented programming
  languages. Recent work achieved incremental type checking for structurally typed
  functional languages through \emph{co-contextual typing rules}, a constraint-based formulation
  that removes any context dependency for expression typings. However,
  that work does not cover key features of object-oriented languages: Subtype polymorphism, nominal
  typing, and implementation inheritance. Type checkers encode these features in the form of class
  tables, an additional form of typing context inhibiting incrementalization.

  In the present work, we demonstrate that an appropriate co-contextual notion to class tables exists,
  paving the way to efficient incremental type checkers for object-oriented languages. 
  This yields a novel formulation of Igarashi et al.'s Featherweight Java (FJ) type system,
  where we replace class tables by the dual concept of class table requirements and
  class table operations by dual operations on class table requirements.
  We prove the equivalence of FJ's type system and our co-contextual formulation.
  Based on our formulation, we implemented an incremental FJ type checker
  and compared its performance against javac on a number of realistic example programs.
 
\end{abstract}


%% file: Introduction.tex

\section{Introduction}
\label{sec:introduction}

Previous work~\cite{Erdweg15} presented a \emph{co-contextual formulation} of the PCF type
system with records, parametric polymorphism, and subtyping
by duality of the traditional contextual formulation.  The contextual formulation is based on a typing context and operations
for looking up, splitting, and extending the context. The co-contextual formulation replaces the
typing context and its operations with the dual concepts of context requirements and operations for
generating, merging, and satisfying requirements. This enables bottom-up type checking that
starts at the leaves of an expression tree.  Whenever a traditional type checker would look up
variable types in the typing context, the bottom-up co-contextual type checker generates fresh type
variables and generates context requirements stating that these type variables need to be bound to
actual types; it merges and satisfies these requirements as it visits the syntax tree upwards to the
root.  The co-contextual type formulation of PCF enables incremental type checking giving rise to
order-of-magnitude speedups~\cite{Erdweg15}.

These results motivated us to investigate co-contextual formulation of the type systems for
statically typed object-oriented (OO) languages, the state-of-the-art programming technology for
large-scale systems. We use Featherweight Java~\cite{Igarashi01}~(FJ) as a
representative calculus for these languages. Specifically, we consider two research questions: (a) Can we
formulate an equivalent co-contextual type system for FJ by duality to the traditional formulation,
and (b) if yes, how to define an incremental type checker based on it with significant speedups?
Addressing these questions is an important step towards a general theory of incremental type
checkers for statically typed OO languages, such as Java, C$\sharp$, or Eiffel.

We observe that the general principle of replacing the typing context and its operations with
co-contextual duals carries over to the \emph{class table}.
The latter is propagated top-down and completely specifies the available classes in the program,
e.g., member signatures and super classes. Dually, a co-contextual type checker
propagates \emph{class table requirements} bottom-up. This data structure
specifies requirements on classes and members and accompanying operations for
generating, merging, and removing these requirements.

However, defining appropriate merge and remove operations on co-contextual class table requirements
poses significant challenges, as they substantially differ from the equivalent operations on context
requirements.
Unlike the global namespace and structural typing of PCF,
FJ features context dependent member signatures (subtype polymorphism), a declared type hierarchy
(nominal typing), and inherited definitions (implementation inheritance).  

\begin{figure*}
\begin{tikzpicture}[exp tree, align=left]
    \node (add) at (0,2.3) {$+$};
    \node (addL) at  (-2.2, 2.3) {$\keyw{new}\ \cname{List}().\mathit{add}(1).\mathit{size}()$}; 
    \node (addR) at (2.8, 2.3) {$\keyw{new}\ \cname{LinkedList}().\mathit{add}(2).\mathit{size}();$};
    \node (R1L) at (-1.9,1.5) {$(R_1)\;\; \cname{List}.init()$\hfill $\hspace{1.5cm} $}; 
    \node (R1R) at (3,1.5) {$(R_4)\;\; \cname{LinkedList}.init()$\hfill $\hspace{1.5cm} $}; 
    \node (R2L) at (-2,1) {$(R_2)\;\; \cname{List}.\mathit{add}: \cname{Int}\rightarrow U_1$}; 
    \node (R2R) at (2.9,1) {$(R_5)\;\; \cname{LinkedList}.\mathit{add}: \cname{Int}\rightarrow U_2$}; 
    \node (R3L) at (-1.6, 0.5) {$(R_3)\;\; U_1.\mathit{size}: () \rightarrow U_3$\hfill $\hspace{1cm}$};
    \node (R3R) at (2.8, 0.5) {$(R_6)\;\; U_2.\mathit{size}: () \rightarrow U_4$\hfill $\hspace{1cm}$};
\end{tikzpicture}
\caption{Requirements generated from co-contextually type checking the $+$ expression.}
\label{fig:intro}
\end{figure*}

For an intuition of class table requirements and the specific challenges concerning their
operations, consider the example in Figure~\ref{fig:intro}. Type checking the operands of $+$
yields the class table requirements $R_1$ to $R_6$.
Here and throughout the paper we use metavariable $U$
to denote unification variables as placeholders for actual types. For example, the invocation of
method $\mathit{add}$
on \keyw{new} \cname{List}() yields a class table requirement $R_2$.
The goal of co-contextual type checking is to avoid using any context information, hence we cannot
look up the signature of $\cname{List}.\mathit{add}$
in the class table. Instead, we use a placeholder $U_1$
until we discover the definition of $\cname{List}.\mathit{add}$
later on. As consequence, we lack knowledge about the receiver type of any subsequent method call,
such as $\mathit{size}$
in our example. This leads to requirement $R_3$,
which states that (yet unknown) class $U_1$
should exist that has a method $\mathit{size}$
with no arguments and (yet unknown) return type $U_3$.
Assuming $+$ operates on integers, type checking the $+$~operator later unifies $U_3$ and $U_4$ with \cname{Int}, thus refining the class table requirements.

To illustrate issues with merging requirements, consider the requirements $R_3$
and $R_6$
regarding $\mathit{size}$.
Due to nominal typing, the signature of this method depends on $U_1$
and $U_2$,
where it is yet unknown how these classes are related to each other. It might be that $U_1$
and $U_2$
refer to the same class, which implies that these two requirements overlap and the corresponding types of $\mathit{size}$ in $R_3$ and $R_6$ are unified. Alternatively,
it might be the case that $U_1$
and $U_2$
are distinct classes, individually declaring a method $\mathit{size}$. Unifying
the types of $size$ from $R_3$ and $R_6$ would be wrong. Therefore, it is locally indeterminate
whether a merge should unify or keep the requirements separate.

To illustrate issues with removing class requirements, consider the requirement $R_5$. Suppose that we encounter a declaration of $\mathit{add}$ in \cname{LinkedList}. Just removing $R_5$ is not sufficient because we do not know whether \cname{LinkedList} overrides $\mathit{add}$ of a yet unknown superclasss $U$, or not.
Again, the situation is locally indeterminate. In case of overriding, 
FJ requires that the signatures of overriding and overridden methods be identical.
Hence, it would necessary add constraints equating the two signatures.
However, it is equally possible that $\cname{LinkedList}.\mathit{add}$ overrides nothing,
so that no additional constraints are necessary.
If, however, $\cname{LinkedList}$ inherits $\mathit{add}$ from $\cname{List}$ without overriding it, we need to record the inheritance relation between these two classes, in order to be able to replace $U_2$ with the actual return type of $\mathit{size}$.

The example illustrates that a co-contextual formulation for nominal typing with subtype 
polymorphism and implementation inheritance poses new research questions that the 
work on co-contextual PCF did not address. A key contribution of the work presented in 
this paper is to answer these questions. The other key contribution is an 
incremental type checker for FJ based on the co-contextual FJ formulation. 
We evaluate the initial and incremental performance
of the co-contextual FJ type checker on synthesized FJ programs and realistic 
java programs by comparison
to javac and a context-based implementation of FJ.

\noindent
To summarize, the paper makes the following contributions:
\begin{itemize}
\item We present a co-contextual formulation of FJ's type system by duality to the traditional type system formulation by Igarashi et al.~\cite{Igarashi01}. Our formulation replaces the class table by its dual concept of class table requirements and it replaces field/method lookup, class table duplication, and class table extension by the dual operations of requirement generation, merging, and removing. In particular, defining the semantics of merging and removing class table requirements in the presence of nominal types, OO subtype polymorphism, and implementation inheritance constitute a key contribution of this work.
\item We present a method to derive co-contextual typing rules for FJ from traditional ones and provide a proof of equivalence between contextual and co-contextual FJ. 
\item We provide a description of type checker optimizations for co-contextual FJ with incrementalization and a performance evaluation. 
\end{itemize}


%% file: Section_2_background_motivation.tex

\section{Background and Motivation}
\label{sec:backgr-motiv}
In this section, we present the FJ typing rules from~\cite{Igarashi01} and give an example to
illustrate how contextual and co-contextual FJ type checkers work.

\subsection{Featherweight Java: Syntax and Typing Rules}
\label{sec:feath-java-synt}
Featherweight Java~\cite{Igarashi01} is a minimal core language for modeling Java's type system.
Figure~\ref{fig:fj-syntax} shows the syntax of classes, constructors, methods, expressions, and
typing contexts. Metavariables $C$, $D$, and $E$
denote class names and types; $f$
denotes fields; $m$
denotes method names; $\keyw{this}$
denotes the reference to the current object. As is customary, an overline denotes a sequence in the
metalanguage. $\Gamma$ is a set of bindings from variables and \keyw{this} to types.

\input{fig-fj-syntax}
\input{fjava-type-rules}

The type system (Figure~\ref{fig:fjava-rules}) ensures that variables, field access, method
invocation, constructor calls, casting, and method and class declarations are well-typed. The typing
judgment for expressions has the form \judge {\ctxplus{\Gamma}{CT}}{e}{C}, where $\Gamma$
denotes the typing context, $CT$
the class table, $e$
the expression under analysis, and $C$
the type of $e$.
The typing judgment for methods has the form \judgeM{\ctxplus{C}{CT}} {M \ \term{OK}}
and for classes \judgeM{CT}{L \ \term{OK}}.

In contrast to the FJ paper~\cite{Igarashi01}, we added some cosmetic changes to the
presentation. For example, the class table $CT$
is an implicit global definition in FJ. Our presentation explicitly propagates $CT$
top-down along with the typing context. Another difference to Igarashi et al.~is in the rule
\rulename{T-New}: Looking up all fields of a class returns a constructor signature, i.e.,
$\func{fields}(C, CT)=\constrCT{C}{\seq D}$
instead of returning a list of fields with their corresponding types.  We made this subtle change
because it clearer communicates the intention of checking the constructor arguments against the
declared parameter types.  Later on, these changes pay off, because they enable a systematic
translation of typing rules to co-contextual FJ (Sections~\ref{sec:co-cont-FJ} and
\ref{sec:co-cont-rules}) and give a strong and rigorous equivalence result for the two type systems
(Section~\ref{sec:theorems-equivalence}).

Furthermore, we explicitly include a typing rule \rulename{T-Program} for programs, which is implicit in Igarashi et al.'s presentation.
The typing judgment for programs has the form \judgeP{\overline{L}\ \term{OK}}: A program is
well-typed if all class declarations are well-typed. The auxiliary functions \func{addExt},
\func{addCtor}, \func{addFs}, and \func{addMs} extract the supertype, constructor, field
and method declarations from a class declaration into entries for the class table.
Initially, the class table is empty, then it is gradually extended
 with information from every class declaration by using the above-mentioned auxiliary functions.
This is to emphasize that we view the class table as an additional
form of typing context, having its own set of extension operations.
We describe the class table extension operations and their co-contextual duals formally in
Section~\ref{sec:co-cont-FJ}.

\subsection{Contextual and Co-Contextual Featherweight Java by Example}
\label{sec:coco-fj-example}

\begin{wrapfigure}[4]{r}{.35\textwidth}
\vspace{-1em}
\begin{lstlisting}[language=FJ, escapeinside={*}{*}]
class !List! extends !Object! {
*$\;\;$*!Int! *$size()$* {*$\ldots$*}
*$\;\;$*!List! *$add(\cname{Int}\ a)$*{*$\ldots$*}
}
class !LinkedList! extends !List! { }
\end{lstlisting}
\vspace{-1em}
\end{wrapfigure} 
We revisit the example from the introduction to illustrate that, in absence of context information, maintaining requirements on class members is non-trivial:
\begin{lstlisting}[language=FJ,backgroundcolor=\color{white},escapeinside={*}{*}]
*\hspace{-15pt}*new !List!().'add'(1).'size'() + new !LinkedList!().'add'(2).'size'().
\end{lstlisting}
Here we assume the class declarations on the right-hand side: \cname{List} with methods $\mathit{add}()$ and $\mathit{size}()$ and \cname{LinkedList} inheriting from \cname{List}.
As before, we assume there are typing rules for numeric \cname{Int} literals and the $+$ operator over \cname{Int} values.
We use \cname{LList} instead of \cname{LinkedList} in Figure~\ref{fig:ast} for space reasons. 

\indent
Figure~\ref{fig:ast}~(a) depicts standard type checking with typing contexts in FJ.
The type checker in FJ visits the syntax tree ``down-up'', starting at the root. Its inputs
(propagated downwards) are the context $\Gamma$, class table $CT$, and the current subexpression $e$. Its output
(propagated upwards) is the type $C$ of the current subexpression. The output is
computed according to the currently applicable typing rule, which is 
determined by the shape of the current subexpression.
The class table used by the standard type checker contains classes \cname{List} and \cname{LinkedList} shown above.
The type checker retrieves the signatures for the method invocations of $\mathit{add}$ and $\mathit{size}$ from the class table $CT$.

To recap, while type checking constructor calls, method invocations, and field accesses the context and the class table flow top-down; types of fields/methods are looked up in the class table. 
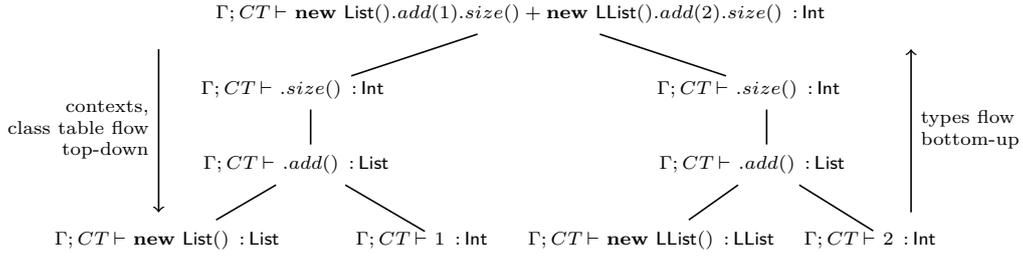
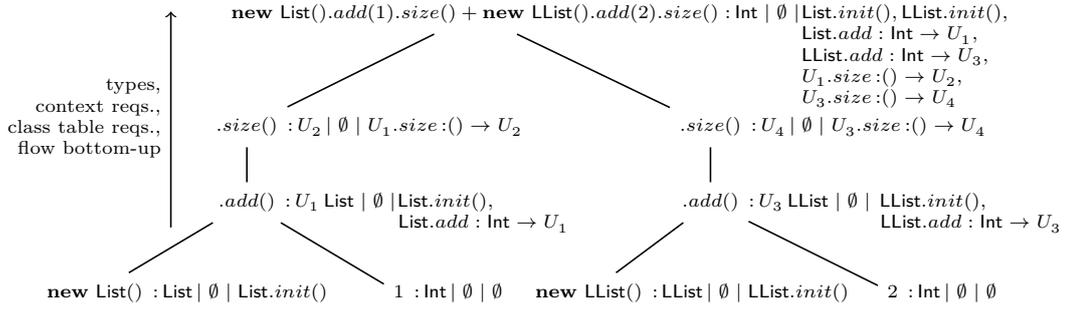
\begin{figure*}[h!]
\centering
\begin{subfigure}[b]{\textwidth}
      \input{check-contextual.tikz}
      \subcaption{Contextual type checking propagates contexts and class tables top-down.}
      \label{fig:check-contextual}
  \end{subfigure}
   \qquad
    \begin{subfigure}[b]{\textwidth}
     \input{check-cocontextual.tikz}
      \subcaption{Co-contextual type checking propagates context and class table requirements bottom-up.}
      \label{fig:check-cocontextual}
      \end{subfigure}
\caption{Contextual and co-contextual type checking.}
\label{fig:ast}
\end{figure*}
Figure~\ref{fig:ast}~(b) depicts type checking of the same expression in co-contextual FJ. Here, the
type checker starts at the leaves of the tree with no information about the context or the class
table. The expression type $T$, the context requirements $R$,
 and class table requirements $CR$ all are
outputs and only the current expression $e$ is input to the type checker, making the type checker context-independent.
At the leaves, we do not know the signature of the
constructors of \cname{List} and \cname{LinkedList}. Therefore, we generate requirements for the constructor calls
$\cname{List}.init()$ and $\cname{LinkedList}.init()$ and propagate them as class table requirements. For each
method invocation of $\mathit{add}$ and $\mathit{size}$ in the tree, we generate requirements on the receiver type and propagate them together with the requirements of the subexpressions.

In addition to generating requirements and propagating them upwards as shown 
in Figure~\ref{fig:ast}~(b), a co-contextual type checker also \emph{merges requirements} when they have compatible receiver types.
In our example, we have two requirements for method $\mathit{add}$ and two requirements for method $\mathit{size}$.
The requirements for method $\mathit{add}$ have incompatible ground receiver types and therefore cannot be merged.
The requirements for method $\mathit{size}$ both have placeholder receivers and therefore cannot be merged just yet.
However, for the $\mathit{size}$ requirements, we can already extract a conditional constraint that must hold if the requirements become mergeable, namely $(U_2 = U_4\ \mathit{if} \ U_1 = U_3)$.
This constraint ensures the signatures of both $\mathit{size}$ invocations are equal in case their receiver types $U_1$ and $U_3$ are equal.
This way, we enable early error detection and incremental solving of constraints.
Constraints can be solved continuously as soon as they have been generated in order to not wait for the whole program to be
 type checked. We discuss incremental type checking in more detail in Section~\ref{sec:implementation}.

After type checking the $+$ operator, the type checker encounters the class declarations of \cname{List} and \cname{LinkedList}.
When type checking the class header \cname{LinkedList} \keyw{extends} \cname{List}, we have to record
the inheritance relation between the two classes because methods can be invoked by \cname{LinkedList}, but
declared in \cname{List}. For example, if \cname{List} is not known to be a superclass of \cname{LinkedList} and given the declaration \cname{List}.$\mathit{add}$, then we cannot just yet satisfy the requirement \cname{LinkedList}.$\mathit{add}: Num\rightarrow U_3$. Therefore, we duplicate the requirement regarding $\mathit{add}$ having as receiver \cname{List}, i.e., $\cname{List}.\mathit{add}: Num \rightarrow U_3$. By doing so, we can deduce the actual type of $U_3$ for the given declaration of $\mathit{add}$ in \cname{List}. Also, requirements regarding $\mathit{size}$ are duplicated.

In the next step, the method declaration of $\mathit{size}$ in \cname{List} is type checked. Hence, we consider all requirements regarding $\mathit{size}$, i.e, $U_1.\mathit{size} : () \rightarrow U_2$ and $U_3.\mathit{size}:() \rightarrow U_4$. The receivers of $mathit{size}$ in both requirements are unknown. We cannot yet satisfy these requirements because we do not know whether $U_1$ and $U_3$ are equal to \cname{List}, or not. To solve this, we introduce conditions as part of the requirements, to keep track of the relations between the unknown required classes and the declared ones. By doing so, we can deduce the actual types of $U_2$ and $U_4$, and satisfy the requirements later, when we have more information about $U_1$ and $U_3$.
 
Next, we encounter the method declaration $\mathit{add}$ and satisfy the corresponding requirements. After satisfying the requirements regarding $\mathit{add}$, the type checker can infer the actual types of $U_1$ and $U_3$. Therefore, we can also satisfy the requirements regarding $\mathit{size}$. 

To summarize, during the co-contextual type checking of constructor calls, method invocations, and field accesses, the requirements flow
bottom-up. Instead of looking up types of fields/methods in the class
table, we introduce new class table requirements. These requirements are satisfied
when the actual types of fields/methods become available.


%% file: fig-fj-syntax.tex
\begin{figure*}[t]
  \centering
  $
\newcommand{\comment}[1]{\hskip2em\text{#1}}
\begin{array}{c@{\hskip.5em}c@{\hskip.5em}ll}
L & :: = & \keyw{class}\ C\ \keyw{extends}\ D\ \{ \overline{C}\ \overline{f};\ K\ \overline{M}\} & \comment{class\ declaration} \\
K & :: = & C(\overline{C}\ \overline{f})\{ \keyw{super}( \overline{f});\ \keyw{this}.\overline{f} =\overline{f}\} & \comment{constructor} \\
M & :: = & C\ m(\overline{C}\ \overline{x})\{\ \keyw{return}\ e;\} & \comment{method declaration}\\
  e & ::= & x \OR \keyw{this} \OR e.f \OR e.m(\overline{e}) \OR \keyw{new}\ C(\overline{e}) \OR (C)e & \comment{expression}\\[2mm]
  \Gamma & ::= & \emptyset \OR \ctxplus{\Gamma}{x : C} \OR \ctxplus{\Gamma}{\keyw{this} : C} & \comment{typing contexts} 
\end{array}
$
\caption{Featherweight Java syntax and typing context.}
\label{fig:fj-syntax}
\end{figure*}

%% file: fjava-type-rules.tex
\begin{figure*}[t]
  \raggedright
  {
  \begin{gather*}
      \inference[\rulename{T-Var}]
        {{\Gamma(x)} = {C}}
        {\judge{\ctxplus{\Gamma}{CT}} {x} {C}}
\hskip2em
      \inference[\rulename{T-Field}]
        {\judge {\ctxplus{\Gamma}{CT}} {e} {C_e} & \func{field}(f_i, C_e ,CT) = C_i}
        {\judge {\ctxplus{\Gamma}{CT}} {e.f_i}{C_i}}
\\[2ex]
      \inference[\rulename{T-Invk}]
        {\judge {\ctxplus{\Gamma}{CT}} {e} {C_e} & \judge{\ctxplus{\Gamma}{CT}} {\seq{e}} {\seq{C}} &  \func{mtype}(m, C_e, CT) = \seq{D}\rightarrow C \quad \seq{C} <: \seq{D}}
        {\judge{\ctxplus{\Gamma}{CT}}{e.m(\seq{e})} {C}}       
\\[2ex]
      \inference[\rulename{T-New}]
      {\judge {\ctxplus{\Gamma}{CT}}{\seq{e}} {\seq{C}} &  \func{fields}(C, CT) = \constrCT{C}{\seq{D}} &  \seq{C} <: \seq{D}}
        {\judge{\ctxplus{\Gamma}{CT}}{\keyw{new}\ C(\seq{e})} {C}}  
\\[2ex]
      \inference[\rulename{T-UCast}]
      {\judge {\ctxplus{\Gamma}{CT}}{e } {D} \quad  D <: C}  
       {\judge {\ctxplus{\Gamma}{CT}}{(C)e} {C}}
\hskip1em
      \inference[\rulename{T-DCast}]
      {\judge{\ctxplus{\Gamma}{CT}}{e} {D} \quad  C <: D \quad C \neq D }  
       {\judge {\ctxplus{\Gamma}{CT}}{(C)e} {C}}  
\\[2ex]
      \inference[\rulename{T-SCast}]
      {\judge {\ctxplus{\Gamma}{CT}}{e} {D} &  C \nless: D & D \nless: C}   
       {\judge {\ctxplus{\Gamma}{CT}}{(C)e} {C}}  
\\[2ex]
      \inference[\rulename{T-Method}]
        {\judge {\ctxplus{\seq{x}: \seq{C}; \keyw{this}: C}{CT}} {e} {E_0}  & E_0 <: C_0 \\
          \func{extends}(C, CT)= D\\
       \term{if}\ \func{mtype}(m, D, CT) = \seq{D}\rightarrow D_0, \term{then}\ \seq{C} = \seq{D};\ C _0= D_0}
        {\judgeM{\ctxplus{C}{CT}} {C \ m(\seq{C}\  \seq{x}) \{\keyw{return}\ e\}\ \term{OK}}}
\\[2ex]
	 \inference[\rulename{T-Class}]
	 {K=C(\seq{D}'\ \seq{g}, \seq{C}'\ \seq{f})\{\keyw{super}(\seq{g}); \keyw{this}.\seq{f} = \seq{f}\} & \func{fields}(D, CT)=\constrCT{D}{\seq D'} \\\judgeM{\ctxplus{C}{CT}}{\seq{M}\ \term{OK}}}
	 {\judgeM{CT}{\keyw{class}\ C\ \keyw{extends}\ D\ \{\seq{C}\ \seq{f}; K\ \seq{M}\}\ \term{OK}}}
\\[2ex]
	\inference[\rulename{T-Program}]
	{ CT = \bigcup_{L'\in\seq{L}}  (\func{addExt}(L')\cup \func{addCtor}(L')\cup \func{addFs}(L') \cup \func{addMs}(L'))  \\
	 (\judgeM{CT}{L'\ \term{OK}})_{L'\in \seq{L}}}
	{\judgeP {\seq{L}\ \term{OK}}}
  \end{gather*}
  }
  \caption{Typing rules of Featherweight Java.}
  \label{fig:fjava-rules}
\end{figure*}


%% file: check-contextual.tikz
\begin{tikzpicture}[exp tree]
\tikzstyle{every node}=[font=\scriptsize]
   \node (add) at (0,3) {$\keyw{new}\ \cname{List}().add(1).size() + \keyw{new}\ \cname{LList}().add(2).size()$};
   \node (lsize) at  (-3, 2) {$.size()$}; 
   \node (rsize) at (3, 2) {$.size()$};
   \node (ladd) at (-3,1) {$.add()$}; 
   \node (radd) at (3,1) {$.add()$}; 
   \node (list) at (-4.7, 0) {$\keyw{new}\ \cname{List}()$};
   \node (one) at (-1.3, 0) {$1$};
   \node (llist) at (1.6, 0) {$\keyw{new}\ \cname{LList}()$};
   \node (two) at (4.6, 0) {$2$};
     
   \node[ctx node, left=of add] (add-ctx) {\ctxNode{\ctxplus{\Gamma}{CT}}};
   \node[type node, right=of add] (add-type){\typeNode{\cname{Int}}}; 
   \node[ctx node, left=of lsize] (lsize-ctx) {\ctxNode{\ctxplus{\Gamma}{CT}}};
   \node[type node, right=of lsize] (lsize-type){\typeNode{\cname{Int}}}; 
   \node[ctx node, left=of rsize] (rsize-ctx) {\ctxNode{\ctxplus{\Gamma}{CT}}}; 
   \node[type node, right=of rsize] (rsize-type){\typeNode{\cname{Int}}}; 
   \node[ctx node, left=of ladd] (ladd-ctx) {\ctxNode{\ctxplus{\Gamma}{CT}}}; 
   \node[type node, right=of ladd] (ladd-type){\typeNode{\cname{List}}};
   \node[ctx node, left=of radd] (radd-ctx) {\ctxNode{\ctxplus{\Gamma}{CT}}}; 
   \node[type node, right=of radd] (radd-type){\typeNode{\cname{List}}};  
   \node[ctx node, left=of list] (list-ctx) {\ctxNode{\ctxplus{\Gamma}{CT}}}; 
   \node[type node, right=of list] (list-type){\typeNode{\cname{List}}};  
   \node[ctx node, left=of llist] (llist-ctx) {\ctxNode{\ctxplus{\Gamma}{CT}}}; 
   \node[type node, right=of llist] (llist-type){\typeNode{\cname{LList}}}; 
   \node[ctx node, left=of one] (one-ctx) {\ctxNode{\ctxplus{\Gamma}{CT}}}; 
   \node[type node, right=of one] (one-type){\typeNode{\cname{Int}}}; 
   \node[ctx node, left=of two] (two-ctx) {\ctxNode{\ctxplus{\Gamma}{CT}}}; 
   \node[type node, right=of two] (two-type){\typeNode{\cname{Int}}};  

   \path (add) edge (lsize); \path (add) edge (rsize); 
   \path (lsize) edge (ladd); \path (rsize) edge (radd);
   \path (ladd) edge (list); \path (ladd) edge (one);
   \path (radd) edge (llist); \path (radd) edge (two);

   \node (ctx-flow-top) at (-5,2.8) {}; \node (ctx-flow-bot) at (-5,0.1) {};
   \draw [->] (ctx-flow-top) -- (ctx-flow-bot) node[left,pos=.5,align=right] {contexts, \\class table flow \\top-down};
   \node (type-flow-bot) at (4.9,2.8) {}; \node (type-flow-top) at (4.9,0.1) {};
   \draw [->] (type-flow-top) -- (type-flow-bot) node[right,pos=.5,align=left] {types flow\\ bottom-up};
\end{tikzpicture}

%% file: check-cocontextual.tikz
\begin{tikzpicture}[exp tree]
\tikzstyle{every node}=[font=\scriptsize]
  \node (add) at (-1,3.7) {$\keyw{new}\ \cname{List}().add(1).size() + \keyw{new}\ \cname{LList}().add(2).size()$};
  \node (lsize) at  (-4, 2.2) {$.size()$}; 
  \node (rsize) at (2.1, 2.2) {$.size()$};
  \node (ladd) at (-4,1.2) {$.add()$}; 
  \node (radd) at (2.1,1.2) {$.add()$}; 
  \node (list) at (-6, 0) {$\keyw{new}\ \cname{List()}$};
  \node (one) at (-2, 0) {$1$};
  \node (llist) at (0.5,0) {$\keyw{new}\ \cname{LList()}$};
  \node (two) at (4.5, 0) {$2$};
  \node[type node, right=of add] (add-type){$:$}; 
  \node[ctx node, right=of add-type, align=left] (add-ctx) {\begin{tabular}{ll}
    & \\ 
    &\\
    &\\
    &\\
  \hspace{-0.2cm}{\scriptsize$\cname{Int}\mathrel{|}\emptyset\mathrel{|}$} & \hspace{-0.4cm}{\scriptsize$\cname{List}.init(), \cname{LList}.init(),$} \\
    & \hspace{-0.4cm}{\scriptsize$ \cname{List}.add:\cname{Int}\rightarrow U_1,$}\\ 
    & \hspace{-0.4cm}{\scriptsize$\cname{LList}.add:\cname{Int}\rightarrow U_3,$} \\
    & \hspace{-0.4cm}{\scriptsize$ U_1.size \ofType () \rightarrow U_2,$} \\
    & \hspace{-0.4cm}{\scriptsize$U_3.size \ofType () \rightarrow U_4$} \\
  \end{tabular}};
  \node[type node, right=of lsize] (ls-type){\typeNode{U_2}}; 
  \node[ctx node, right=of ls-type] (ls-ctx) {\coFNode{\emptyset}{U_1.size \ofType () \rightarrow U_2}};
  \node[type node, right=of rsize] (rs-type){\typeNode{U_4}}; 
  \node[ctx node, right=of rs-type] (rs-ctx) {\coFNode{\emptyset}{U_3.size \ofType () \rightarrow U_4}};
  \node[type node, right=of ladd] (la-type){\typeNode{U_1}}; 
  \node[ctx node, right=of la-type] (la-ctx) {\begin{tabular}{ll}
    & \\ 
  \hspace{-0.2cm}{\scriptsize$\cname{List}\mathrel{|}\emptyset\mathrel{|}$} &  \hspace{-0.4cm}{\scriptsize$\cname{List}.init(), $} \\
    &  \hspace{-0.5cm} {\scriptsize$ \cname{List}.add:\cname{Int}\rightarrow U_1$}\\ 
  \end{tabular}};
  \node[type node, right=of radd] (ra-type){\typeNode{U_3}}; 
  \node[ctx node, right=of ra-type] (ra-ctx) {\begin{tabular}{ll}
    & \\ 
  \hspace{-0.2cm}{\scriptsize$\cname{LList}\mathrel{|}\emptyset\mathrel{|}$} &  \hspace{-0.3cm}{\scriptsize$\cname{LList}.init(), $} \\
    &  \hspace{-0.4cm} {\scriptsize$ \cname{LList}.add:\cname{Int}\rightarrow U_3$}\\ 
  \end{tabular}};
  \node[type node, right=of list] (l-type){\typeNode{\cname{List}}}; 
  \node[ctx node, right=of l-type] (l-ctx) {\coFNode{\emptyset}{\cname{List}.init()}};
  \node[type node, right=of llist] (ll-type){\typeNode{\cname{LList}}}; 
  \node[ctx node, right=of ll-type] (ll-ctx) {\coFNode{\emptyset}{\cname{LList}.init()}};
  \node[type node, right=of one] (o-type){\typeNode{\cname{Int}}}; 
  \node[ctx node, right=of o-type] (o-ctx) {\coFNode{\emptyset}{\emptyset}};
  \node[type node, right=of two] (t-type){\typeNode{\cname{Int}}}; 
  \node[ctx node, right=of t-type] (t-ctx) {\coFNode{\emptyset}{\emptyset}};

  \path (add) edge (lsize); \path (add) edge (rsize); 
  \path (lsize) edge (ladd); \path (rsize) edge (radd);
  \path (ladd) edge (list); \path (ladd) edge (one);
  \path (radd) edge (llist); \path (radd) edge (two);

  \node (type-flow-bot) at (-5,4) {};
  \node (type-flow-top) at (-5,0.6) {};
  \draw[->] (type-flow-top) -- (type-flow-bot) node[left,pos=.5, align=right] {types,\\ context reqs.,\\  class table reqs.,\\ flow bottom-up};

\end{tikzpicture}
 

%% file: Section_3_Co-Contx_FJava.tex

\section{Co-Contextual Structures for Featherweight Java}
\label{sec:co-cont-FJ}

In this section, we present the dual structures and operations for the co-contextual formulation of
FJ's type system. Specifically, we introduce bottom-up propagated \emph{context and class table
  requirements}, replacing top-down propagated typing contexts and class tables. 

\subsection{Class Variables and Constraints}
\label{sec:class-vars}

For co-contextual FJ, we reuse the syntax of FJ in Figure~\ref{fig:fj-syntax}, but extend the
type language to \emph{class types}:
\\[1ex]
$\newcommand{\comment}[1]{\hskip2em\text{#1}}
\begin{array}{c@{\hskip.5em}c@{\hskip.5em}ll}
  \multicolumn{3}{l}{U,V,\ldots} & \comment{Class Variable} \\
  T & ::= & C \OR U & \comment{Class Type} 
\end{array}
$\\[.7ex]
We use constraints for refining class types, i.e., co-contextual FJ is a constraint-based type
system.
That is, next to class names, the type system may assign \emph{class variables}, designating
unknowns in constraints. We further assume that there are countably many class variables, equality
of class variables is decidable and that class variables and class names are disjoint.

During bottom-up checking, we propagate sets $S$ of constraints: 
\\[1ex]
$
\newcommand{\comment}[1]{\hskip2em\text{#1}}
\begin{array}{c@{\hskip.5em}c@{\hskip.5em}ll}
  s  & ::= & T = T \OR T\neq T \OR T <: T \OR T\nless T \OR T = T\ \mathrm{if}\; \cond & \comment{constraint} \\ 
  S  & ::= & \emptyset \OR S;s & \comment{constraint set}
\end{array}
$\\[.7ex]
A constraint $s$
either states that two class types must be equal, non-equal, in a subtype relation, non-subtype, or
equal if some condition holds, which we leave underspecified for the moment.

\subsection{Context Requirements}
\label{sec:context-requirements}
A typing context is a set of bindings from variables to types, while a context requirement is a set
of bindings from variables to class variables $U$.
Below we show the operations on typing contexts and their co-contextual
correspondences, reproduced from~\cite{Erdweg15}. 
Operations on typing context are lookup, extension, and duplication; their respective requirement
context duals are: generating, removing, and
merging.
Co-contextual FJ adopts context requirements and operations for method parameters and \keyw{this} unchanged.

  \begin{tabular}[t]{l@{\hskip0.5em}l}
    \toprule
    Contextual & Co-contextual \\
    \midrule
    Context syntax $\Gamma ::= \emptyset \OR \ctxplus{\Gamma}{x \ofType T}$ & Requirements $\reqs \subset x \times T$ map variables to their types \\
    Context lookup $\Gamma(x) = T$ & Requirement introduction $\reqs = \{x \ofType U\}$ with \\
                                                           & fresh unification variable $U$ \\
    Context extension $\ctxplus{\Gamma}{x \ofType T}$ & Requirement satisfaction $\reqs - x$\, if $(\reqs(x) = T)$ holds \\
    Context duplication $\Gamma \to\, (\Gamma, \Gamma)$ & Requirement merging $merge_R(R_1, R_2) = \reqs|_S$\\
                                                          & if all constraints $(T_1 = T_2) \in S$ hold \\
    Context is empty $\Gamma = \emptyset$ & No unsatisfied requirements $\reqs \mustEqual \emptyset$ \\
    \bottomrule
  \end{tabular}

\subsection{Structure of Class Tables and Class Table Requirements}
\label{sec:struct-class-tabl}

\input{fig-ct-ctreqs-syntax}

In the following, we describe the dual notion of a class table, called \emph{class table
  requirements} and their operations.  We first recapitulate the structure of FJ class
tables~\cite{Igarashi01}, then stipulate the structure of class table
requirements. Figure~\ref{fig:ct-ctreqs-syntax} shows the syntax of both.
A class table is a collection of class definition clauses $\mathit{CTcls}$ defining the available
classes.\footnote{To make the correspondence to class table requirements more obvious, we show a
decomposed form of class tables. The original FJ formulation~\cite{Igarashi01} groups clauses by
the containing class declaration.}  A clause is a class name $C$ followed by either the
superclass, the signature of the constructor, a field type, or a method signature of $C$'s
definition.

As Figure~\ref{fig:ct-ctreqs-syntax} suggests, class tables and definition clauses in FJ have a
counterpart in co-contextual FJ. Class tables become \emph{class table requirements} $\mathit{CR}$,
which are collections of pairs $(\mathit{CReq}, \mathit{cond})$, where $\mathit{CReq}$ is a
\emph{class requirement} and $\mathit{cond}$ is its \emph{condition}.
Each class definition clause has a corresponding class requirement $\mathit{CReq}$, which is
one of the following:
\begin{itemize}
\item A \emph{inheritance requirement} \extendsCR{T}{T'}, i.e., class type $T$ must inherit from $T'$.
  
\item A \emph{constructor requirement} \constrCT{T}{\seq{T}'}, i.e., class type $T$'s
  constructor signature must match $\seq{T}'$.
\item A \emph{field requirement} $T.f: T'$, i.e., class $T$ (or one of its \emph{supertypes}) must declare
  field $f$ with class type $T'$.
  
\item A \emph{method requirement} $T.m: \seq{T}' \to T''$, i.e., class $T$ (or one of its
  \emph{supertypes}) must declare method $m$ matching signature $\seq{T}'\to T''$.

\item An \emph{optional method requirement} $(T.m: \seq{T}' \rightarrow T'')_{opt}$, i.e., if the
  class type $T$ declares the method $m$, then its signature must match $\seq{T}' \rightarrow T''$. While type checking method declarations, this requirement is used to ensure that method overrides in subclasses are well-defined. 
 An optional method requirement is used as a counterpart of the conditional method lookup in rule \rulename{T-Method} of standard FJ, 
 i.e., $if\ \func{mtype}(m, D, CT) = \bar{D}\rightarrow D_0, then\ \bar{C} = \bar{D};\ C _0= D_0$, 
 where $D$ is the superclass of the class $C$, in which the method declaration $m$ under scrutiny is type checked, 
 and $\bar C,\ C_0$ are the parameter and returned types of $m$ as part of $C$.
\end{itemize}

\noindent A condition $\mathit{cond}$ is a conjunction of equality and nonequality constraints on class
types. Intuitively, $(\mathit{CReq}, \mathit{cond})$ states that if the condition $\mathit{cond}$ is
satisfied, then the requirement $\mathit{CReq}$ must be satisfied, too. Otherwise, we have unsolvable constraints, indicating a typing error. 
With conditional requirements and constraints, we address the feature of nominal typing and inheritance for co-contextual FJ.
In the following, we will describe their usage.

\subsection{Operations on Class Tables and Requirements}

\input{fig-correspondenceCT}

In this section, we describe the co-contextual dual to FJ's class table operations as outlined in Figure~\ref{fig:correspondenceCT}.
We first consider FJ's lookup operations on class tables, which appear in premises of typing rules
shown in Figure~\ref{fig:fjava-rules} to look up (1) fields, (2) field lists, (3) methods and (4)
superclass lookup. The dual operation is to introduce a
corresponding class requirement for the field, list of fields, method, or superclass. 

Let us consider
closely field lookup, i.e., $\func{field}(f_i, C, CT) = C_i$, meaning that class $C$ in the class
table $CT$ has as member a field $f_i$ of type $C_i$. We translate it to the dual operation of introducing a
new class requirement $(C.f_i: U, \emptyset)$. Since we do not have any
information about the type of the field, we choose a \emph{fresh}
class variable $U$ as type of field $f_i$. At the time of introducing a new requirement, its
condition is empty.

Consider the next operation $\func{fields}(C, CT)$, 
which looks up all field members of a class. This lookup is used in the constructor call rule \rulename{T-New}; 
the intention is to retrieve the \emph{constructor signature} in order to type check the subtyping relation between 
this signature and the types of expressions as parameters of the \emph{constructor call}, i.e., $\bar C <: \bar D$ (rule \rulename{T-New}). 
As we can observe, the field names are not needed in this rule, only their types. 
Hence, in contrast to the original FJ rule~\cite{Igarashi01}, we deduce the constructor signature 
from fields lookup, rather than field names and their corresponding types, 
i.e., $\func{fields}(C, CT)=\constrCT{C}{\bar D}$.
The dual
operation on class requirements is to add a new class requirement for the 
constructor, i.e., $(\constrCT{C}{\bar{U}}, \emptyset)$. Analogously, the class table operations for method signature
lookup and super class lookup map to corresponding class table requirements.

Finally, standard FJ uses class table duplication to forward the class table to all parts of an FJ program, thus ensuring all parts are checked against the same context.
The dual co-contextual operation, $merge_{CR}$,
merges class table requirements originating from different parts of the program.
Importantly, requirements merging needs to assure all parts of the program require compatible
inheritance, constructors, fields, and methods for any given class.
To merge two sets of requirements, we first identify the field and
method names used in both sets and then compare the classes they belong to. The result of merging two
sets of class requirements $CR_1$ and $CR_2$ is a new set $CR$ of class requirements and a set of
constraints, which ensure compatibility between the two original sets of overlapping requirements.
Non-overlapping requirements get propagated unchanged to $CR$ whereas potentially overlapping requirements receive special treatment depending on the requirement kind.

 \begin{figure*}[t]
   	    \begin{flalign*}
    	& CR_m =\{ (T_1.m: \seq{T_1}\rightarrow  T'_1, cond_1 \cup (T_1 \neq T_2))  \\
    	& \hskip4em \cup (T_2.m :\seq{T_2}\rightarrow T'_2, cond_2 \cup (T_1 \neq T_2)) \\
        & \hskip4em \cup (T_1.m: \seq{T_1}\rightarrow T'_1, cond_1 \cup cond_2 \cup (T_1 = T_2))  \\
        & \hskip4em \WHERE (T_1.m :\seq{T_1}\rightarrow T'_1, cond_1) \in CR_1  \wedge 
        (T_2.m : \seq{T _2}\rightarrow T'_2, cond_2) \in CR_2  \}\\
\\
        & S_m =\{(T'_1 = T'_2\ if\ T_1 = T_2) \cup (\seq{T_1 = T_2}\ if\ T_1 = T_2) \\
        & \hskip4em \WHERE (T_1.m: \seq{T_1}\rightarrow T'_1, cond_1) \in CR_1
          \wedge (T_2.m: \seq{T_2} \rightarrow T'_2, cond_2) \in CR_2\} 
     \end{flalign*}
     \caption{Merge operation of method requirements $CR_1$ and $CR_2$.}
     \label{fig:mergeM}
\end{figure*} 

The full merge definition appears in Appendix \ref{sec:appendixA}; 
Figure~\ref{fig:mergeM} shows the merge operation for overlapping method requirements, 
which results in a new set of requirements $CR_m$ and constraints $S_m$.
To compute $CR_m$, we identify method requirements on the equally-named methods $m$ in both sets and distinguish two cases.
First, if the receivers are different $T_1 \neq T_2$, then the requirements are not actually overlapping. We retain the two requirements unchanged, except that we remember the failed condition for future reference.
Second, if the receivers are equal $T_1 = T_2$, then the requirements are actually overlapping. We merge them into a single requirement and produce corresponding constraints in $S_m$.
One of the key benefits of keeping track of conditions in class table requirements is that often these conditions allow us to discharge requirements early on when their conditions are unsatisfiable.
In particular, in Section~\ref{sec:implementation} we describe a compact representation of conditional requirements that facilitates early pruning and is paramount for good performance.
However, the main reason for conditional class table requirements is their removal, which we discuss subsequently.

\subsection{Class Table Construction and Requirements Removal}

Our formulation of the contextual FJ type system differs in the presentation of
the class table compared to the original paper~\cite{Igarashi01}.
Whereas Igarashi et al. assume that the class table is a pre-defined static structure,
we explicitly consider its formation through a sequence of operations.
The class table is initially empty and gradually extended with class table clauses $CTcls$ for each 
class declaration $L$ of a program. Dual to that, class table requirements are initially unsatisfied and gradually removed. We define an operation for \emph{adding} clauses to the class table
and a corresponding co-contextual dual operation on class table requirements for \emph{removing} 
requirements.
Figure~\ref{fig:ctcorrespondence} shows a collection 
of adding and removing operations for every possible kind of class table clause $CTcls$.

In general, clauses are added to the class table starting from superclass to subclass declarations. 
For a given class, the class header with \keyw{extends} is added before the other clauses. 
Dually, we start removing requirements that correspond to clauses of a subclass, followed by
those corresponding to clauses of superclass declarations. 
For a given class, we first remove requirements corresponding 
to method, fields, or constructor clauses, then those corresponding to the 
class header \keyw{extends} clause. 
Note that our sequencing still allows for mutual class dependencies. For example, the following is a valid sequence of clauses where \cname{A} depends on \cname{B} and vice versa:
\begin{lstlisting}[language=FJ,backgroundcolor=\color{white}, mathescape]
class !A! extends !Object!; class !B! extends !Object!; !A!.'m': $() \rightarrow$ !B!; !B!.'m': $() \rightarrow$ !A!.
\end{lstlisting}

\noindent
The full definition of the addition and removal operations for all cases of clause definition appears in Appendix \ref{sec:appendixA}; 
Figure~\ref{fig:removeM} presents the definitions of adding and removing method and \keyw{extends} clauses.  

\input{fig-ct-correspondence-table}

\begin{figure*}[t]
\begin{flalign*}
& \func{addMs}(C, \overline{M}, CT) = \overline{C.m: \seq{C} \rightarrow C'} \cup CT \\
&\func{removeM}( C , C'\ m(\seq{C}\ \seq{x}) \ \{\keyw{return} \ e\}, CR) = CR'|_{S} \\
&where \; CR' = \{(T.m :\seq{T}\rightarrow T',cond \cup (T \neq C)) \WHERE (T.m :\seq{T}\rightarrow T', \cond) \in CR \}\\
& \hskip6em \cup (CR\setminus \seq{(T.m :\seq{T}\rightarrow T', \cond)} )\\
& \hskip2em S = \{(T' =  C' \ if\ T= C) \cup (\seq{T = C}\ if\ T= C) \WHERE (T.m :\seq{T}\rightarrow T', \cond) \in CR \} \\
&\func{removeMs}(C, \seq{M}, CR ) = CR'|_{S}  \\
& where\ CR' = \{ CR_m \WHERE (C'\ m(\seq{C}\ \seq{x}) \ \{\keyw{return} \ e\}) \in\overline{M} \\
& \hskip6em \wedge \func{removeM}(C, C'\ m(\seq{C}\ \seq{e}) \ \{\keyw{return} \ e\}, CR) = CR_m|_{S_m} \}   \\
& \hskip4em S = \{S_m \WHERE (C'\ m(\seq{C}\ \seq{x}) \ \{\keyw{return} \ e\}) \in\overline{M}  \\
& \hskip6em \wedge \func{removeM}(C, C'\ m(\seq{C}\ \seq{x}) \ \{\keyw{return} \ e\}, CR) = CR_m|_{S_m} \} 
\end{flalign*}
\vskip-1em
\begin{flalign*}
&\func{addExt}(\keyw{class}\ C\ \keyw{extends}\ D, CT)= ( C\ \keyw{extends}\ D)\cup CT \\
&\func{removeExt}(\keyw{class}\ C\ \keyw{extends}\ D, CR) = CR'|_S \\
&where\ CR'= 
\!\begin{aligned}[t]
&\{(T.\keyw{extends}: T',cond \cup (T \neq C))  \WHERE (T.\keyw{extends}: T', cond) \in CR \} \\
&\cup \{(T.m: \seq{T}\rightarrow T', cond\cup(T\neq C)) \\
& \hskip0.7em \cup(D.m:\seq{T} \rightarrow T', cond \cup(T = C))\WHERE  (T.m: \seq{T}\rightarrow T', cond) \in CR \}\\
&\cup \{(T.m : \seq{T}\rightarrow T', cond \cup (T\neq C))_{opt} \\
& \hskip0.7em \cup (D.m: \seq{T} \rightarrow T', cond \cup (T = C))_{opt} \\
& \hskip0.7em \WHERE  (T.m: \seq{T}\rightarrow T', cond)_{opt} \in CR \} \\
&\cup \{(T.f:  T', cond \cup (T\neq C)) \cup (D.f: T', cond \cup (T = C))\\
& \hskip0.7em \WHERE  (T.f: T', cond) \in CR \} \\
& \hskip-2em S = \{(T' = D\ if\ T= C)  \WHERE (T.\keyw{extends} :T', \cond) \in CR \} 
\end{aligned}
\end{flalign*}
\caption{Add and remove operations of method and extends clauses.}
\label{fig:removeM}
\end{figure*}

\paragraph{Remove operations for method clauses.}
The function \func{removeMs} removes a list of methods by applying the function \func{removeM} to each of them. 
\func{removeM} removes a single method declaration defined in class $C$.
To this end, \func{removeM} identifies requirements on the same method name $m$ and refines their receiver to be different from the removed declaration's defining class.
That is, the refined requirement $(T.m : \ldots, cond \cup (T \neq C))$ only requires method $m$ if the receiver $T$ is different from the defining class $C$.
If the receiver $T$ is, in fact, equal to $C$, then the condition of the refined requirement is unsatisfiable and can be discharged.
To ensure the required type also matches the declared type, \func{removeM} also generates conditional constraints in case $T = C$.
Note that whether $T = C$ can often not be determined immediately because $T$ may be a placeholder type $U$.

We illustrate the removal of methods using the class declaration of \cname{List} shown in Section~\ref{sec:coco-fj-example}.
Consider the class requirement set $CR =$ $(U_1.\mathit{size}()\rightarrow U_2,\emptyset )$.
Encountering the declaration of method $\mathit{add}$ has no effect on this set because there is no requirement on $\mathit{add}$.
However, when encountering the declaration of method $\mathit{size}$, we refine the set as follows:
$$removeM(\cname{List}, \cname{Int}\ \mathit{size}()\ \{\ldots\}, CR) = \{(U_1.\mathit{size}:()\rightarrow U_2, \cname{U_1}\neq \cname{List} )\}|_S$$
with a new constraint $S = \{U_2 =\cname{Int}\ \mathit{if}\ \cname{U_1} = \cname{List}\}$.
Thus, we have satisfied the requirement in $CR$ for $U_1 = \cname{List}$, only leaving the requirement in case $U_1$ represents another type.
In particular, if we learn at some point that $U_1$ indeed represents \cname{List}, we can discharge the requirement because its condition is unsatisfiable.
This is important because a program is only closed and well-typed if its requirement set is empty.

\paragraph{Remove operations for extends clauses.}
The function \func{removeExt} removes the \keyw{extends} clauses ($C.\ \keyw{extends}\ D$). This function, in addition to identifying the requirements regarding \keyw{extends} and following the same steps as above for \func{removeM}, duplicates all requirements for fields and methods. The duplicate introduces a requirement the same as the existing one, but with a different receiver, which is the superclass $D$ that potentially declares the required fields and methods. The conditions also change. We add to the existing requirements an inequality condition ($T\neq C$), to encounter the case when the receiver $T$ is actually replaced by $C$, but it is required to have a certain field or method, which is declared in $D$, the superclass of $T$. This requirement should be discharged because we know the actual type of the required field or method, which is inherited from the given declaration in $D$. Also, we add an equality condition to the duplicate requirement  $T = C$, because this requirement will be discharged when we encounter the actual declarations of fields or methods in the superclass.

We illustrate the removal of \keyw{extends} using the class declaration 
\cname{LinkedList} \keyw{extends} \cname{List}. 
Consider the requirement set $CR = (U_3.\mathit{size}: () \rightarrow U_4, \emptyset )$. We encounter the declaration for
\cname{LinkedList} and the requirement set changes as follows: 
\begin{align*}
  &\func{removeExt}(\keyw{class}\ \cname{LinkedList}\ \keyw{extends}\ \cname{List}, CR) =\\
   &\quad\{(U_3. \mathit{size}:() \rightarrow U_4, U_3\neq \cname{LinkedList}), (\cname{List}.\mathit{size}:() \rightarrow U_4, U_3 = \cname{LinkedList})\}|_S,
\end{align*}
where 
$S = \emptyset$. $S$ is empty, because there are no requirements on \keyw{extends}. If we learn at some point that $U_3 = \cname{LinkedList}$, then the requirement $(U_3. \mathit{size}:() \rightarrow U_4, U_3\neq \cname{LinkedList})$ is discharged because its condition is unsatisfiable. Also, if we learn that $\mathit{size}$ is declared in \cname{List}, then $(\cname{List}.\mathit{size}:() \rightarrow U_4, U_3 = \cname{LinkedList})$ is discharged applying \func{removeM}, as shown above, and $U_4$ can be replaced by its actual type. 

\paragraph{Usage and necessity of conditions.} As shown throughout this section, conditions play an
important role to enable merging and removal of requirements over nominal receiver types and to support inheritance. Because of nominal typing, field and method lookup depends on the name of the defining context and we do not know the actual type of the
receiver class when encountering a field or method reference. Thus, it is impossible to deduce their types until more information is
known. Moreover, if a class is required to have fields/methods, which are actually declared in a superclass
of the required class, then we need to deduce their actual type/signature and meanwhile fulfill the
respective requirements.
For example, considering the requirement $\cname{U_3}.\mathit{size}:()\rightarrow U_4$,
if $U_3 = \cname{LinkedList}$, $\cname{LinkedList}\;\keyw{extends}\; \cname{List}$, and $\mathit{size}$ is declared in $\cname{List}$, then we have to deduce the actual type of $U_4$
and satisfy this requirement. To overcome these obstacles we need additional structure to
maintain the relations between the required classes and the declared ones,
and also to reason about the partial fulfillment of requirements. 
Conditions come to place as the needed structure to maintain these relations and indicate the fulfillment
of requirements.


%% file: fig-ct-ctreqs-syntax.tex
\begin{figure}[t]
  \newcommand{\comment}[1]{\hskip2em\text{#1}}
  \setlength{\columnsep}{-4.5em}
  \begin{multicols}{2}
    \textbf{Contextual}\\[1em]
$\begin{array}{l@{\hspace{0.3em}}l@{\hspace{2pt}}l@{\hspace{-15pt}}l}
\mathit{CT}    & ::= & \emptyset                          & \comment{class table} \\
               & \OR & \mathit{CTcls}\cup \mathit{CT}        & \comment{} \\
\mathit{CTcls} & ::= &                                    & \comment{def. clause} \\
                           
               & \OR &          \extendsCT{C}{D}          & \comment{extends clause} \\
               & \OR & \constrCT{C}{\seq{C}}              & \comment{ctor clause} \\
               & \OR & C.f : C'                            & \comment{field clause}\\
               & \OR & C.m : \seq{C} \rightarrow C'        & \comment{method clause} \\
 \end{array}$\\

\columnbreak
\textbf{Co-Contextual}\\[1em]
$\begin{array}{l@{\hspace{0.3em}}l@{\hspace{2pt}}l@{\hspace{-15pt}}l}
\mathit{CR}   & ::= & \emptyset                            & \comment{class table req.} \\
              & \OR & (\mathit{CReq, cond})\cup \mathit{CR}   & \comment{} \\
\mathit{CReq} & ::= &                                      & \comment{class req.}\\
              & \OR & \extendsCR{T}{T'}                    & \comment{inheritance req.} \\
              & \OR & \constrCT{T}{\seq{T}}                & \comment{ctor req.} \\
              & \OR & T.f : T'                              & \comment{field req.}\\
              & \OR & T.m : \seq{T} \rightarrow T'          & \comment{method req.} \\
              & \OR & (T.m: \seq{T} \rightarrow T')_{opt}   & \comment{optional method req.} \\
\mathit{cond} & ::= & \emptyset \OR T = T' ;  \mathit{cond} & \comment{condition} \\
              & \OR & T \neq T' ;  \mathit{cond}            &                
\end{array}$
\end{multicols}
\caption{Class Table and Class Table Requirements Syntax.}
\label{fig:ct-ctreqs-syntax}
\end{figure}


%% file: fig-correspondenceCT.tex
\begin{figure*}[t]
  \centering
  \begin{tabular}[t]{l@{\hskip0.5em}l}
    \toprule
    \textbf{Contextual}                                    &\textbf{ Co-contextual} \\
    \midrule
    Field name lookup $\func{field}(f_i,C, CT) = C_i$             & Class requirement for field \\
                                                           &  \quad $(C.f_i : U,\emptyset )$\\
    Fields lookup $\func{fields}(C, CT) = \constrCT{C}{\seq{C}}$  & Class requirement for constructor \\
                                                           & \quad$(\constrCT{C}{\seq{U}},\emptyset )$\\
    Method lookup $\func{mtype}(m, C, CT) = \seq{C}\rightarrow C$ & Class requirement for method \\
                                                           & \quad $(C.m: \seq{U}\rightarrow U,\emptyset )$\\
    Conditional method override                            & Optional class requirement for method \\                         
    \quad $if\ \func{mtype}(m, C, CT) = \seq{C}\rightarrow C$     & \quad $(C.m: \seq{U}\rightarrow U,\emptyset )_{opt}$\\       
    Super class lookup $\func{extends}(C, CT) = D$                & Class requirement for super class \\
                                                           & \quad $(\extendsCR{C}{U},\emptyset )$ \\
    Class table duplication $CT \to\, (CT, CT)$            & Class requirement merging \\
                                                           & \quad $ merge_{CR}(CR_1, CR_2)= CR|_S$\\
                                                           & \quad if all constraints in $S$ hold\\
  \bottomrule
  \end{tabular}
  \caption{Operations on class table and their co-contextual correspondence.}
  \label{fig:correspondenceCT}
 \vspace{-1em}
\end{figure*}

%% file: fig-ct-correspondence-table.tex
\begin{figure*}[t]
  \centering
  \begin{tabular}[t]{l@{\hskip0.5em}l}
    \toprule
    Contextual & Co-contextual \\
    \midrule
    Class table empty $CT = \emptyset$ & Unsatisfied class requirements $CR$ \\
     Adding extend $\func{addExt}(L, CT)$ & Remove extend $\func{removeExt}(L, CR)$  \\
       Adding constructor $\func{addCtor}(L, CT)$ & Remove constructor $\func{removeCtor}(L, CR)$  \\
     Adding fields $\func{addFs}(L, CT)$ & Remove fields $\func{removeFs(L, CR)}$  \\
     Adding methods $\func{addMs}(L, CT)$ & Remove methods $\func{removeMs(L, CR)}$  \\
    \bottomrule
  \end{tabular}
  \caption{Constructing class table and their co-contextual correspondence.}
  \label{fig:ctcorrespondence}
\end{figure*}

%% file: Section_4_co_fj_type_rules.tex

\section{Co-Contextual Featherweight Java Typing Rules}
\label{sec:co-cont-rules}
In this section we
derive co-contextual FJ's typing rules systematically from FJ's typing rules. The main idea is to
transform the rules into a form that eliminates any context dependencies that require top-down
propagation of information.

Concretely, context and class table requirements
(Section~\ref{sec:co-cont-FJ}) in output positions to the right replace typing contexts and class
tables in input positions to the left. Additionally, co-contextual FJ propagates constraint sets $S$
in output positions. Note that the program typing judgment does not change, because programs are
closed, requiring neither typing context nor class table inputs. Correspondingly, neither context
nor class table requirements need to be propagated as outputs.

\input{co-fjava}
Figure~\ref{fig:co-fjava} shows the co-contextual FJ typing rules (the reader may
want to compare against contextual FJ in Figure \ref{fig:fjava-rules}). In what follows, we will
discuss the rules for each kind of judgment. 

\subsection{Expression Typing}
\label{sec:expression-typing}

Typing rule \rulename{TC-Var} is dual to the standard variable lookup rule \rulename{T-Var}. It
marks a distinct \emph{occurrence} of $x$ (or the self reference $\keyw{this}$) by assigning a fresh
class variable $U$. Furthermore, it introduces a new context requirement $\{x : U\}$, as the dual
operation of context lookup for variables $x$ ($\Gamma(x) = C$) in \rulename{T-Var}. Since the
latter does not access the class table, dually, \rulename{TC-Var} outputs empty class table
requirements.

Typing rule \rulename{TC-Field} is dual to \rulename{T-Field} for field accesses. The latter
requires a field name lookup (\func{field}), which, dually, translates to a new class requirement for the
field $f_i$, i.e., $(T_e.f_i : U, \emptyset)$ (cf. Section~\ref{sec:co-cont-FJ}).  Here, $T_e$ is
the class type of the receiver $e$. $U$ is a fresh class variable, marking a distinct occurrence of
field name $f_i$, which is the class type of the entire expression. Furthermore, we merge the new
field requirement with the class table requirements $CR_e$ propagated from $e$. The result of merging
is a new set of requirements $CR$ and a new set of constraints $S_{cr}$. Just as the context
$\Gamma$ is passed into the subexpression $e$ in \rulename{T-Field}, we propagate the context
requirements for $e$ for the entire expression. Finally, we propagate both the constraints $S_e$
for $e$ and the merge constraints $S_f$ as the resulting output constraints.

Typing rule \rulename{TC-Invk} is dual to \rulename{T-Invk} for method invocations.
Similarly to field access, the dual of method lookup is introducing a requirement for the
method $m$ and merge it with the requirements from the
premises. 
Again, we choose fresh class variables for the method signature $\seq{U}\rightarrow U'$, marking a
distinct occurrence of $m$. We type check the list $\seq{e}$ of parameters, adding a subtype
constraint $\seq{T} <: \seq{U}$, corresponding to the subtype check in \rulename{T-Invk}.
Finally, we merge all context and class table requirements propagated from the receiver $e$ and the
parameters $\seq{e}$, and all the constraints.

Typing rule \rulename{TC-New} is dual to \rulename{T-New} for object creation. We add a new class
requirement $\constrCT{C}{\seq{U}}$ for the constructor of class $C$, corresponding to the $fields$
operation in FJ. We cannot look up the fields of $C$ in the class table, therefore we assign fresh
class variables $\seq{U}$ for the constructor signature. We add the subtyping constraint
$\seq{T} <: \seq{U}$ for the parameters, analogous to the subtype check in \rulename{T-New}. As in
the other rules, we propagate a collective merge of the propagated requirement
structures/constraints from the subexpressions with the newly created requirements/constraints.

Typing rules for casts, i.e., \rulename{TC-UCast}, \rulename{TC-DCast} and \rulename{TC-SCast} are
straightforward adaptions of their contextual counterparts following the same principles. 
These three type rules do overlap. We do not distinguish them in the formalization, 
but to have an algorithmic formulation, we implement different node names for each of them. 
That is, typing rules for casts are syntactically distinguished.

\subsection{Method Typing}
\label{sec:method-typing}

The typing rule \rulename{TC-Method} is dual to \rulename{T-Method} for checking method
declarations. For checking the method body, the contextual version extends the empty typing context
with entries for the method parameters $\seq{x}$ and the self-reference $\keyw{this}$, which is
implicitly in scope. Dually, we remove the requirements on the parameters and self-reference in
$R_e$ propagated from the method body. Corresponding to extending an empty context, the removal
should leave no remaining requirements on the method body. Furthermore, the equality constraints
$\seq{S_x}$ ensure that the annotated class types for the parameters agree with the class types in
$R_e$.\footnote{Note that a parameter $x$ occurs in the method body if and only if there is a
  requirement for $x$ in $R_e$ (i.e., $x\in\dom{R_e}$), which is due to the bottom-up
  propagation. The same holds for the self-reference $\keyw{this}$.} This corresponds to binding the
parameters to the annotated classes in a typing context. Analogously, the constraints $S_c$ deal
with the self-reference. For the latter, we need to know the associated class type, which in the
absence of the class table is at this point unknown. Hence, we assign a fresh class variable $U_c$
for the yet to be determined class containing the method declaration. The contextual rule
\rulename{T-Method} further checks if the method declaration correctly overrides another method
declaration in the superclass, that is, if it exists in the superclass must have the same type. 
We choose another fresh class variable $U_d$ for the
yet to be determined superclass of $U_c$ and add appropriate supertype and optional method
override requirements. We assign to the optional method requirement $U_d.m$ the type of $m$ 
declared in $U_c$. If later is known that there exists a declaration for $m$ in the actual 
type of $U_d$, the optional requirement is considered and equality constraints are generated. 
These constraints ensure that the required type of $m$ in the optional requirement is the same as 
the provided type of $m$ in the actual superclass of $U_c$. 
Otherwise this optional method requirement is invalidated and not considered. 
By doing so, we enable the feature of subtype polymorphism for co-contextual FJ. 
Finally, we add the subtype constraint ensuring that the method body's type 
is conforming to the annotated return type.  

\subsection{Class Typing}
\label{sec:class-typing}
Typing rule \rulename{TC-Class} is used for checking class
declarations. A declaration of a given class $C$ provides definite information on the
identity of its superclass $D$, constructor, fields, and methods signatures. Dual to the fields
lookup for superclass $D$ in \rulename{T-Class}, we add the constructor requirement
$\constrCT{D}{\seq{D}'}$. We merge this requirement with all requirements generated from type
checking $C$'s method declarations $\seq{M}$. Recall that typing of method $m$ yields a
distinct class variable $U$ for the enclosing class type, because we type check each method 
declaration independently. Therefore, we add the constraints $\seq{\{U= C\}}$, effectively completing
the method declarations with their enclosing class $C$. 

\subsection{Program Typing}
\label{sec:program-typing}
Type rule \rulename{TC-Program} checks a list of class declarations $\seq{L}$. Class declarations 
of all classes provide a definite information on the identity of their super classes, constructor, 
fields, methods signatures. Dual to adding clauses in the class table by constructing it, 
we remove requirements with respect to the 
provided information from the declarations. Hence, dually to class table being fully extended 
with clauses from all class declarations, requirements are empty. The result of removing different 
clauses is a new set of requirement and a set of constraints. Hence, we use notation $\uplus$ to 
express the union of the returned tuples (requirements and constraints), 
i.e., $CR|_S \uplus CR'|_{S'}= CR\cup CR'|_{S\cup S'}$
 After applying remove to the set 
of requirements, the set should be empty at this point. A class requirement 
is discharged from the set, either when performing remove operation
(Section~\ref{sec:co-cont-FJ}), or when it is satisfied (all conditions hold). 

As shown, we can systematically derive co-contextual typing rules for Featherweight Java through duality.


%% file: co-fjava.tex
\begin{figure}[h!]
  {
  \begin{gather*}
      \inference[\rulename{TC-Var}]
        {\fresh{U}}
        {\cojudge {x} {U} {\emptyset} {x:U} {\emptyset}}
\\[2ex]
   \inference[\rulename{TC-Field}]
     {\cojudge {e} {T_e} {S_e} {R_e} {CR_e} & CR|_{S_f} = \func{merge}_{CR}(CR_e, (T_e.f_i : U, \emptyset ))  \\
       \fresh{U}}
      {\cojudge {e.f_i} {U} {S_e \cup  S_{f}} {R_e} {CR }} 
\\[2ex]
   \inference[\rulename{TC-Invk}]
   {\cojudge {e} {T_e}  {S_e} {R_e} {CR_e} &  \seq{\cojudge {{e}} {{T}} {{S}} {{R}} {{CR}}} \\ 
   CR_m = (T_e.m: \seq{U} \rightarrow U', \emptyset  ) & \seq{S_s = \{{T} <: {U}\}} & {\afresh{U',\ \seq{U}}} \\
  R'|_{\typectxcolor{S_r}} =  \func{merge}_R(R_e, \seq{R})  &
     CR'|_{S_{cr}} = \func{merge}_{CR}(CR_e, CR_m , \seq{CR}) }
  {\cojudge {e.m(\seq{e})} {U'} {\seq{S} \cup S_e \cup  \seq{S_s}  \cup S_r \cup S_{cr}} {R'} {CR'}}
  \\[2ex]
   \inference[\rulename{TC-New}]
   { \seq{\cojudge {{e}} {{T}} {{S}} {{R}} {{CR}}} & CR_f=(\constrCT{C}{\seq{U}},\emptyset ) & 
        \seq{S_s= \{{T} <: {U}}\}\\
  \fresh{\seq{U}} &  R'|_{S_r} = \func{merge}_R(\seq{R})  &
  CR'|_{S_{cr}} = \func{merge}_{CR}(CR_f, \seq{CR}) }
 {\cojudge {\keyw{new}\ C(\seq{e})} {C} { \seq{S} \cup \seq{S_s}\cup \typectxcolor{S_r} \cup \typectxcolor{S_{cr}}} { R'} {CR'}}
\\[2ex] 
  \inference[\rulename{TC-UCast}]
  {\cojudge {e} {T_e} {S_e} {R_e} {CR_e}  &  \typecolor{S_s} = \{T_e <: C\}}
  {\cojudge {(C) e} {C} {S_e \cup S_s} {R_e} {CR_e}}
\\[2ex]
\inference[\rulename{TC-DCast}]
      {\cojudge {e} {T_e} {S_e} {R_e} {CR_e} &  S_s = \{ C <: T_e\} & S_n = \{ C \neq T_e\}}  
       {\cojudge {(C) e} {C} {S_e \cup S_s \cup S_n} {R_e} {CR_e}}  
\\[2ex]
\inference[\rulename{TC-SCast}]
   {\cojudge {e} {T_e} {S_e} {R_e} {CR_e} &  \typecolor{S_s} = \{ C \nless: T_e\} & \typecolor{S'_s} = \{ T_e \nless: C\}}   
   {\cojudge {(C) e} {C} {S_e \cup S_s \cup S'_s } {R_e} {CR_e}}  
\\[2ex]
\inference[\rulename{TC-Method}]
	{\cojudge {e} {T_e} {S_e} {R_e} {CR_e} &  \seq{\typecolor{S_x} = \{ C = R_e(x)\ \OR\ x \in \dom {R_e} \}}  \\   
	 \typecolor{S_c}=\{ U_c = R_e(\keyw{this})\ \OR\ \keyw{this} \in \dom{R_e}\} &  \typecolor{S_s} = \{T_e <: C_0\}\\
	   \ctxcolor{R_e} - \ctxcolor{\keyw{this}} - \ctxcolor{\seq{x}} = \emptyset  & \afresh{U_c,\ U_d}\\
	    	CR|_{\typectxcolor{S_{cr}}} = \func{merge}_{CR}(CR_e, (\extendsCR{U_c}{U_d}, \emptyset ) , (U_d.m:\seq{C} \rightarrow C_0, \emptyset )_{opt}) }
	 {\cojudgeOM {C_0\ m(\seq{C}\ \seq{x})\ \{ \keyw{return}\ e\} \ \term{OK}}  {S_e \cup S_s \cup S_c \cup \typecolor{S_{cr}} \cup \typecolor{\seq{S}_x}} {U_c}{ CR}}
  \\[2ex]
\inference[\rulename{TC-Class}] 
 { K=C(\seq{D}'\ \seq{g}, \seq{C}'\ \seq{f})\{\keyw{super}(\seq{g}); \keyw{this}.\seq{f} = \seq{f}\} & \seq{\cojudgeOM {{M} \ \term{OK}}  {{S}} {{U}} {{CR}}} \\
    CR'|_{ \typectxcolor{S_{cr}}} =  \func{merge}_{CR}((\constrCT{D}{\seq{D}'}, \emptyset ), \seq{CR}) \quad \seq{\typecolor{S_{eq}} = \{U = C\}}} 
 {\cojudgeOC {\keyw{class}\ C\ \keyw{extends}\ D \{\seq{C}\ \seq{f};\ K\ \seq{M}\}\ \term{OK}} {\seq{S} \cup \seq{S}_{eq}  \cup  \typectxcolor{S_{cr}}} {CR' }}
 \\[2ex]
 \inference[\rulename{TC-Program}]
	{\seq{\cojudgeOC{{L}\ \term{OK}}{{S}}{{CR}}} & \func{merge}_{CR}(\seq{CR})= CR'|_{S'} \\
	\biguplus_{ L' \in \seq{L}}( \func{removeMs}(CR', L')\uplus \func{removeFs}(CR', L')\uplus \func{removeCtor}(CR', L') \\
	   \uplus \func{removeExt}(CR', L') ) = \emptyset|_{S} }
	{\judgeOP {\seq{L}\ \term{OK}}{\seq{S} \cup S'\cup S}}
 \end{gather*}
 }
 \caption{A co-contextual formulation of the type system of Featherweight Java.}
  \label{fig:co-fjava}
\end{figure}


%% file: Section_5_theorem.tex
\section{Typing Equivalence}
\label{sec:theorems-equivalence}
In this section, we prove the typing equivalence of expressions, methods, classes, and programs between FJ and co-contextual FJ. That is, 
(1) we want to convey that an expression, method, class and program is type checked in FJ if and only if it is type checked in co-contextual FJ, and (2) that there is a correspondence relation 
between typing constructs for each typing judgment. 

We use $\sigma$ to represent substitution, which is a set of bindings from class variables to class types $(\{U\smap C\})$. \projExt{CT} is a function that given a class table $CT$ returns the immediate subclass relation $\Sigma$ of classes in $CT$. That is, $\Sigma := \{(\cname{C_1}, \cname{C_2})\WHERE (\cname{C_1}\ \keyw{extends}\ \cname{C_2}) \in CT \}$. Given a set of constraints $S$ and a relation between class types $\Sigma$, where $\projExt{CT} = \Sigma$, then the solution to that set of constraints is a substitution, i.e., $\func{solve}(S,\Sigma)=\sigma$. Also we assume that every element of the $class\ table$, i.e., super types, constructors, fields and methods types are class type, namely ground types. 
Ground types are types that cannot be substituted. 

Initially, we prove equivalence for expressions.
Let us first delineate the \emph{correspondence relation}. 
Correspondence states that $a)$ the types of expressions are the same in both formulations, $b)$ provided variables in context are more than required ones in context requirements and $c)$ provided class members are more than required ones. Intuitively, an expression to be well-typed in co-contextual FJ should have all requirements satisfied. Context requirements are satisfied when for all required variables, we find the corresponding bindings in context. Class table requirements are satisfied, 
when for all valid requirements, i.e., all conditions of a requirement hold, we can find a corresponding declaration in a class of the same type as the required one, or in its superclasses.
The relation between class table and class requirements is formally defined in the Appendix~\ref{sec:appendixB}.
 \begin{definition}[Correspondence relation for expressions]
Given judgments \judge{\ctxplus{\Gamma}{CT}}{e}{C}, \cojudge{e}{T}{S}{R}{CR}, and $\func{solve}(\Sigma,S) = \sigma$, where 
$\projExt{CT}=\Sigma$.  
The correspondence relation between $\Gamma$ and $R$, $CT$ and $CR$, written $(C, \Gamma, CT)\vartriangleright \sigma(T, R, CR)$, is defined as: 
\begin{enumerate}[a)]
	\item $C = \sigma(T)$
	\item $\Gamma \supseteq \sigma(R)$
	\item $CT\ satisfies\ \sigma(CR)$ 
\end{enumerate}	
\end{definition}

We stipulate two different theorems to state both directions of equivalence for expressions.
 \begin{theorem}[Equivalence of expressions: $\Rightarrow$]
 Given $e,\ C, \ \Gamma ,\ CT, $ if $\judge{\ctxplus{\Gamma}{CT}}{e}{C}$, 
then there exists $T, \ S,\ R,\ CR,\ \Sigma,\ \sigma$, where $\projExt{CT}=\Sigma$ and $\func{solve}(\Sigma,S)=\sigma$, such that $\cojudge{e}{T}{S}{R}{CR}$ holds, $\sigma$ is a ground solution and $(C, \Gamma, CT)\vartriangleright \sigma(T, R, CR)$ holds. 
\label{theo: ExprER}
 \end{theorem}
\begin{theorem}[Equivalence of expressions: $\Leftarrow$]
	Given $e,\ T,\ S,\ R ,\ CR,\ \Sigma,$ if $\cojudge{e}{T}{S}{R}{CR}$, $\func{solve}(\Sigma ,S) = \sigma$, and $\sigma$ is a ground solution, then there exists $C,\ \ctxcolor{ \Gamma,\ CT}$, such that \\
$\judge{\ctxplus{\Gamma}{CT}}{e}{C}$, $(C, \Gamma, CT)\vartriangleright\sigma(T, R, CR)$ and  $\projExt{CT}=\Sigma$.
\label{theo: ExprEL}
 \end{theorem}

Theorems \ref{theo: ExprER} and \ref{theo: ExprEL} are proved by induction on the typing judgment of expressions. 
The most challenging aspect consists in proving the relation between the class table and class table requirements. 
In Theorem \ref{theo: ExprER}, the class table is given and the requirements are a collective merge of the propagated requirement from the subexpressions with the newly created requirements.
In Theorem \ref{theo: ExprEL}, the class table is not given,therefore we construct it through the information 
retrieved from \emph{ground class requirements}.
We ensure class table correctness and completeness with respect to the given requirements. First, we ensure that the class table we construct is correct, i.e., types of \keyw{extends}, fields, and methods clauses we add in the class table are equal to the types of the same \keyw{extends}, fields, and methods if they already exist in the class table. 
Second, we ensure that the class table we construct is complete, i.e., the constructed class table satisfies all given requirements. 

Next, we present the theorem of equivalence for methods. The difference from expressions is that there is no context, therefore no relation between context and context requirements is required. Instead, the fresh class variable introduced in co-contextual FJ as a placeholder for the actual class, where the method under scrutiny is type checked in, after substitution should be the same as the class where the method is type checked in FJ.

 \begin{theorem}[Equivalence of methods: $\Rightarrow$]
Given $m,\ C,\ CT, $ if $\judgeM{\ctxplus{C}{CT}}{C_0 \ m(\seq{C}\  \seq{x}) \{return\ e\}\\ OK}$, 
then there exists $ S,\ T,\ CR,\ \Sigma,\ \sigma$, where $\projExt{CT}=\Sigma$ and $\func{solve}(\Sigma, S)= \sigma$, such that\\
$\cojudgeOM {C_0\ m(\seq{C}\ \seq{x})\ \{ return\ e_0\} \ OK}  {S} {T}{CR}$ holds, $\sigma$ is a ground solution and \\
$(C, CT)\vartriangleright_m \sigma(T, CR)$ holds. 
 \end{theorem}
 
 \begin{theorem}[Equivalence of methods: $\Leftarrow$]
Given $m,\ T,\ S,\ CR,\ \Sigma, $ if $\cojudgeOM {C_0\ m(\seq{C}\ \seq{x})\ \{ return\ e_0\} \\ OK}  {S} {T}{CR}$, $\func{solve}(\Sigma, S) = \sigma$, and $\sigma$ is a ground solution, then there exists $ C,\ CT$, such that
$\judgeM{\ctxplus{C}{CT}}{C_0 \ m(\seq{C}\  \seq{x}) \{return\ e\}\ OK}$ holds, $(C, CT)\vartriangleright_m \sigma(T, CR)$ and $\projExt{CT}=\Sigma$. 
 \end{theorem}
 
Theorems 5 and 6 are proved by induction on the typing judgment.  
The difficulty increases in proving equivalence for methods because we have to consider the optional requirement, as introduced in the previous sections. It requires a different strategy to prove the relation between the class table and optional requirements; we accomplish the proof by using case distinction. We have a detailed proof for method declaration, and also how this affects class table construction, and we prove a correct and complete construction of it.
   
Lastly, we present the theorem of equivalence for classes and programs. 

 \begin{theorem}[Equivalence of classes: $\Rightarrow$]
Given $C,\ CT, $ if \judgeM{CT}{class\ C\ extends\ D\ \{\seq{C}\ \seq{f}; K\\ \seq{M}\}\ OK},
then there exists $ S,\ CR,\ \Sigma,\ \sigma$, where $\projExt{CT}=\Sigma$ and $\func{solve}(\Sigma, S)=\sigma$, such that
\cojudgeOC{class\ C\ extends\ D \{\seq{C}\ \seq{f};\ K\ \seq{M}\}\ OK} {S} {CR} holds, $\sigma$ is a ground solution and 
$(CT)\vartriangleright_c \sigma(CR)$ holds. 
 \end{theorem}
\begin{theorem}[Equivalence of classes: $\Leftarrow$]
Given $C,\ CR,\ \Sigma,$ if \cojudgeOC {class\ C\ extends\ D \{\seq{C}\ \seq{f};\ K\ \seq{M}\} \\
 OK} {S} {CR}, 
$\func{solve}(\Sigma,S)=\sigma$, and 
$\sigma$ is a ground solution, then there exists $CT$, such that \\
 \judgeM{CT}{class\ C\ extends\ D\ \{\seq{C}\ \seq{f}; K\ \seq{M}\}\ OK} holds, $(CT)\vartriangleright_c \sigma(CR)$ holds and $\projExt{CT}=\Sigma$. 
 \end{theorem} 
Theorems 8 and 9 are proved by induction on the typing judgment. Class declaration requires to prove only the relation between the class table and class table requirements since there is no context. 

Typing rule for programs does not have as inputs context and class table, therefore there is no relation between context, class table and requirements. The equivalence theorem describes that a program in FJ and co-contextual FJ is well-typed. 

 \begin{theorem}[Equivalence for programs: $\Rightarrow$]  \label{thm:equivP}
Given $\seq L $, if \judgeP{\seq{L}\ OK}, then there exists S, $\Sigma$, $\sigma$, where $\projExt{\seq L} = \Sigma$ and $\func{solve}(\Sigma, S)=\sigma$, such that 
\judgeOP{\seq{L}\ OK}{S} holds and $\sigma$ ground solution. 
\end{theorem}
\begin{theorem}[Equivalence for programs: $\Leftarrow$]  \label{thm:equivPl}
Given $\seq L $, if \judgeOP{\seq{L}\ OK}{S}, $\func{solve}(\Sigma, S)=\sigma $, where $\projExt{\seq L} =\Sigma$, and $\sigma$ is a ground solution, then  \judgeP{ \seq{L}\ OK} holds.
\end{theorem}
Theorems 10 and 11 are proved by induction on the typing judgment. In here, we prove that a class table containing all clauses provided from the given class declarations is dual to empty class table requirements in the inductive step.  \\
Omitted definitions, lemmas and proofs can be found at the Appendix~\ref{sec:appendixB}.

%% file: Section_6_implementation.tex
\section{Efficient Incremental FJ Type Checking}
\label{sec:implementation}
\lstset{language=FJ}

The co-contextual FJ model from Section~\ref{sec:co-cont-FJ} and \ref{sec:co-cont-rules} was
designed such that it closely resembles the formulation of the original FJ type system, where all
differences are motivated by dually replacing contextual operations with co-contextual ones.  As
such, this model served as a good basis for the equivalence proof from the previous section.
However, to obtain a type checker implementation for co-contextual FJ that is amenable to efficient
incrementalization, we have to employ a number of behavior-preserving optimizations.  In the present
section, we describe these optimization and the resulting \emph{incremental} type checker
implementation for co-contextual FJ.  The source code is available online at \implementationurl{}.

\paragraph{Condition normalization.}
In our formal model from Section~\ref{sec:co-cont-FJ} and \ref{sec:co-cont-rules}, we represent
context requirements as a set of conditional class requirements
$CR \subset \mathit{Creq} \times \mathit{cond}$.
Throughout type checking, we add new class requirements using function \func{merge}, but we only
discharge class requirements in rule \rulename{TC-Program} at the very end of type checking.  Since
\func{merge} generates $3*m*n$
conditional requirements for inputs with $m$
and $n$ requirements respectively, requirements quickly become intractable even for small programs.

The first optimization we conduct is to eagerly normalize conditions of class requirements.  Instead
of representing conditions as a list of type equations and inequations, we map receiver types to the
following condition representation (shown as Scala code):
\begin{lstlisting}[language=FJ,mathescape,backgroundcolor=\color{white}]
case class !Condition(notGround: Set[CName], notVar: Set[UCName], 
                            $\,$sameVar: Set[UCName], sameGroundAlternatives: Set[CName])!.
\end{lstlisting}
A condition is true if the receiver type is different from all ground types (\lstinline|!CName!|) and
unification variables (\lstinline|!UCName!|) in \lstinline|!notGround!| and \lstinline|!notVar!|, if the
receiver type is equal to all unification variables in \lstinline|!sameVar!|, and if
\lstinline|!sameGroundAlternatives!| is either empty or the receiver type occurs in it.  That is, if
\lstinline|!sameGroundAlternatives!| is non-empty, then it stores a set of alternative ground types,
one of which the receiver type must be equal to.

When adding an equation or inequation to the condition over a receiver type, we check whether the
condition becomes unsatisfiable. For example, when equating the receiver type to the ground type
\lstinline|!C!| and \lstinline|!notGround.contains(C)!|, we mark the resulting condition to be
unsatisfiable.  Recognizing unsatisfiable conditions has the immediate benefit of allowing us to
discard the corresponding class requirements right away.  Unsatisfiable conditions occur quite
frequently because \func{merge} generates both equations and inequations for all receiver types that
occur in the two merged requirement sets.

If a condition is not unsatisfiable, we normalize it such that the following assertions are
satisfied: (i) the receiver type does not occur in any of the sets, (ii)
\lstinline{!sameGroundAlternatives.isEmpty!} \lstinline{ || !notGround.isEmpty!}, and (iii)
\lstinline|!notVar.intersect(sameVar).isEmpty!|.  Since normalized conditions are more compact, this
optimization saves memory and time required for memory management.  Moreover, it makes it easy to
identify irrefutable conditions, which is the case exactly when all four sets are empty, meaning
that there are no further requisites on the receiver type.  Such knowledge is useful when
\func{merge} generates conditional constraints, because irrefutable conditions can be ignored.
Finally, condition normalization is a prerequisite for the subsequent optimization.

\paragraph{In-depth merging of conditional class requirements.}
In the work on co-contextual PCF~\cite{Erdweg15}, the number of requirements of an expression was
bound by the number of free variables that occur in that expression.  To this end, the \func{merge}
operation used for co-contextual PCF identifies subexpression requirements on the same free variable
and merges them into a single requirement.  For example, the expression $x + x$
has only one requirement $\{x: U_1\}|_{\{U_1 = U_2\}}$,
even though the two subexpressions propagate two requirements $\{x: U_1\}$
and $\{x: U_2\}$, respectively.

Unfortunately, the \func{merge} operation of co-contextual FJ given in
Section~\ref{sec:context-requirements} does not enjoy this property.  Instead of merging
requirements, it merely collects them and updates their conditions.  A more in-depth merge of
requirements is possible whenever two code fragments require the same member from the same receiver
type.  For example, the expression $\keyw{this}.x + \keyw{this}.x$
needs only one requirement $\{U_1.x() : U_2\}|_{\{U_1=U_3,U_2=U_4\}}$,
even though the two subexpressions propagate two requirements $\{U_1.x() : U_2\}$
and $\{U_3.x() : U_4\}$,
respectively.  Note that $U_1=U_3$
because of the use of $\keyw{this}$
in both subexpressions, but $U_2=U_4$ because of the in-depth merge.

However, conditions complicate the in-depth merging of class requirements: We may only merge two
requirements if we can also merge their conditions.  That is, for conditional requirements
$(\mathit{creq}_1, \mathit{cond}_1)$
and $(\mathit{creq}_2, \mathit{cond}_2)$
with the same receiver type, the merged requirement must have the condition
$\mathit{cond}_1 \vee \mathit{cond}_2$.
In general, we cannot express $\mathit{cond}_1 \vee \mathit{cond}_2$
using our \lstinline|!Condition!| representation from above because all fields except
\lstinline|!sameGroundAlternatives!| represent conjunctive prerequisites, whereas
\lstinline|!sameGroundAlternatives!| represents disjunctive prerequisites.  Therefore, we only
support in-depth merging when the conditions are identical up to
\lstinline|!sameGroundAlternatives!| and we use the union operator to combine their
\lstinline|!sameGroundAlternatives!| fields.

This optimization may seem a bit overly specific to certain use cases, but it turns out it is
generally applicable.  The reason is that function \func{removeExt} creates requirements of the form
$(D.f: T', cond \cup (T = C_i))$
transitively for all subclasses $C_i$
of $D$
where no class between $C_i$
and $D$
defines field $f$.
Our optimization combines these requirements into a single one, roughly of the form
$(D.f: T', cond \cup (T = \bigvee_i C_i))$.
Basically, this requirement concisely states that $D$
must provide a field $f$
of type $T'$ if the original receiver type $T$ corresponds to any of the subclasses $C_i$ of $D$.

\paragraph{Incrementalization and continuous constraint solving.}
We adopt the general incrementalization strategy from co-contextual PCF~\cite{Erdweg15}: Initially,
type check the full program bottom-up and memoize the typing output for each node (including class
requirements and constraint system).  Then, upon a change to the program, recheck each node from the
change to the root of the program, reusing the memoized results from unchanged subtrees.  This way,
incremental type checking asymptotically requires only $\log n$ steps for a program with $n$ nodes.

In our formal model of co-contextual FJ, we collect constraints during type checking and solve them
at the end to yield a substitution for the unification variables.  As was discussed by
Erdweg~et~al.\ for co-contextual PCF~\cite{Erdweg15}, this strategy is inadequate for incremental
type checking, because we would memoize unsolved constraints and thus only obtain an incremental
constraint generator, but even a small change would entail that all constraints had to be solved
from scratch.
In our implementation, we follow Erdweg~et~al.'s strategy of continuously solving constraints as
soon as they are generated, memoizing the resulting partial constraint solutions.  In particular,
equality constraints that result from \func{merge} and \func{remove} operations can be solved
immediately to yield a substitution, while subtype constraints often have to be deferred until more
information about the inheritance hierarchy is available.  In the context of FJ with its nominal
types, continuous constraint solving has the added benefit of enabling additional requirement
merging, for example, because two method requirements share the same receiver type after
substitution.

\paragraph{Tree balancing.}
Even with continuous constraint solving, co-contextual FJ as defined in
Section~\ref{sec:co-cont-rules} still does not yield satisfactory incremental performance.  The
reason is that the syntax tree is deformed due to the root node, which consists of a sequence of
\emph{all} class declarations in the program.  Thus, the root node has a branching factor only bound
by the number of classes in the program, whereas the rest of the tree has a relative small branching
factor bound by the number of arguments to a method.  Since incremental type checking recomputes
each step from the changed node to the root node, the type checker would have to repeat discharging
class requirements at the root node after every code change, which would seriously impair
incremental performance.

To counter this effect, we apply tree balancing as our final optimization.  Specifically, instead of
storing the class declarations as a sequence in the root node, we allow sequences of class
declarations to occur as inner nodes of the syntax tree:
\[L \ \  :: = \ \  \overline{L} \ \OR\  \keyw{class}\ C\ \keyw{extends}\ D\ \{ \overline{C}\ \overline{f};\ K\ \overline{M}\} \]
This allows us to layout a program's class declarations structurally as in\linebreak[4]
$((((C_1\ C_2)\ C_3)\ (C_4\ C_5))\ (C_6\ C_7))$,
thus reducing the costs for rechecking any path from a changed node to the root node.  As part of
this optimization, to satisfy requirements of classes that occur in different tree nodes such as
$C_1$
and $C_6$,
we also neeed to propagate \emph{class facts} such as actual method signatures upwards.  As
consequence, we can now link classes in any order without changing the type checking result.

We have implemented an incremental co-contextual FJ type checker in Scala using the optimizations
described here.  In the following section, we present our run-time performance evaluation.


%% file: Section_7_evaluation.tex
\section{Performance Evaluation}
\label{sec:evaluation}

\lstset{language=FJ}

We have benchmarked the initial and incremental run-time performance of co-contextual FJ
implementation.  However, this evaluation makes no claim to be complete, but rather is intended to
confirm the feasibility and potential of co-contextual FJ for incremental type checking.

\subsection{Evaluation on synthesized FJ programs}
\paragraph{Input data.}
We synthesized FJ programs with 40~root classes that inherit from \lstinline|!Object!|. Each root
class starts a binary tree in the inheritance hierarchy of height~$5$.
Thus, each root-class hierarchy contains $31$
FJ class declarations. In total, our synthesized programs have $31*40 + 3 = 1243$
class declarations, since we always require classes for natural numbers \lstinline|!Nat!|,
\lstinline|!Zero!|, and \lstinline|!Succ!|.

Each class has at least a field of type \lstinline|!Nat!| and each class has a single method that
takes no arguments and returns a \lstinline|!Nat!|.  We generated the method body according to one
of three schemes:
\begin{itemize}
\item \emph{AccumSuper}: The method adds the field's value of this class to the result of calling the method of the super class.
\item \emph{AccumPrev}: Each class in root hierarchy $k > 1$ has an additional field that has the type of the class at the same position in the previous root hierarchy $k-1$. The method adds the field's value of this class to the result of calling the method of the class at the same position in the previous root hierarchy $k-1$, using the additional field as receiver object.
\item \emph{AccumPrevSuper}: Combines the other two schemes; the method adds all three numbers.
\end{itemize}

\noindent
We also varied the names used for the generated fields and methods:
\begin{itemize}
\item \emph{Unique}: Every name is unique.
\item \emph{Mirrored}: Root hierarchies use the same names in the same classes, but names within a single root hierarchy are unique.
\item \emph{Override}: Root hierarchies use different names, but all classes within a single root hierarchy use the same names for the same members.
\item \emph{Mir+Over}: Combines the previous two schemes, that is, all classes in all root hierarchies use the same names for the same members.
\end{itemize}

For evaluating the incremental performance, we invalidate the memoized results for the three
\lstinline|!Nat!| classes.  This is a critical case because all other classes depend on the
\lstinline|!Nat!| classes and a change is traditionally hard to incrementalize.

\paragraph{Experimental setup.}
First, we measured the wall-clock time for the initial check of each program using our co-contextual
FJ implementation.  Second, we measured the wall-clock time for the incremental reanalysis after
invalidating the memoized results of the three \lstinline|!Nat!| classes.  Third, we measured the
wall-clock time of checking the synthesized programs on javac and on a straightforward
implementation of contextual FJ for comparison. Contextual FJ is the standard FJ described in
Section~\ref{sec:backgr-motiv}, that uses contexts and class tables during type checking. Our
implementation of contextual FJ is up to 2-times slower than javac, because it is not production
quality.  We used ScalaMeter\footnote{\url{https://scalameter.github.io/}} to take care of JIT
warm-up, garbage-collection noise, etc.  All measurements were conducted on a 3.1GHz duo-core
MacBook Pro with 16GB memory running the Java HotSpot 64-Bit Server VM build 25.111-b14 with 4GB
maximum heap space.  We confirmed that confidence intervals were small.

\paragraph{Results.}
We show the measurement results in table~\ref{tab:evaluation}. All numbers are in milliseconds. We
also show the speedups of initial and incremental run of co-contextual type checking relative to
both javac and contextual type checking.

\begin{table}[h]
  \centering
\begin{tabular}[t]{l@{}c@{\ \ }|l@{\ \ }|l@{\ \ }}
\toprule
  \textbf{Super} & \multicolumn{1}{l|}{javac / contextual} & \multicolumn{1}{c|}{co-contextual init} &  \multicolumn{1}{c}{co-contextual inc} \\
unique &           70.00 /  93.99  & 3117.73 (0.02x / 0.03x) & 23.44 (2.9x / 4x) \\
mirrored &         68.03 / 88.73 & 1860.18 (0.04x / 0.05x) & 15.17 (4.5x / 6x)\\
override &         73.18 / 107.83 & 513.44 $\;\;$(0.14x / 0.21x) & 16.92 (4.3x / 6x)\\
mir+over &         72.64 / 132.09 & 481.07 $\;\;$(0.15x / 0.27x) & 16.60 (4.4x / 8x)\\ 
\midrule
  \textbf{Prev} & \multicolumn{1}{l|}{javac / contextual} & \multicolumn{1}{c|}{co-contextual init} &  \multicolumn{1}{c}{co-contextual inc} \\
unique &           82.16 / 87.66 & 3402.28 (0.02x / 0.02x) & 23.43 (3.5x / 4x) \\
mirrored &         81.19 / 84.94 & 2136.42 (0.04x / 0.04x) & 15.46 (5.3x / 5x) \\
override &         81.51 / 120.60 & 840.14  $\;\;$(0.09x / 0.14x)  & 17.37 (4.7x / 7x) \\
mir+over &         79.71 / 120.46 & 816.16  $\;\;$(0.09x / 0.15x) & 16.61 (4.8x / 7x)\\ 
\midrule
  \textbf{PrevSuper} & \multicolumn{1}{l|}{javac / contextual} & \multicolumn{1}{c|}{co-contextual init} &  \multicolumn{1}{c}{co-contextual inc} \\
unique &           93.12 / 104.03 & 6318.69 (0.01x / 0.02x) & 26.26 (3.5x / 4x) \\
mirrored &         95.41 / 100.00 & 5014.12 (0.02x / 0.02x) & 15.71 (6.1x / 6x)\\
override &         92.88 / 130.01 & 3601.44 (0.03x / 0.04x) & 17.35 (5.4x / 7x) \\
mir+over &         93.37 / 126.57 & 3579.90 (0.03x / 0.04x) & 16.61 (5.6x / 8x) \\ 
\bottomrule
\end{tabular}
\caption{Performance measurement results with $k = 40$ root classes in \textbf{Milliseconds}. Numbers in parentheses indicate speedup relative to (javac/contextual) base lines.}\label{tab:evaluation}
\end{table}

As this data shows, the initial performance of co-contextual FJ is subpar: The initial type check
takes up to 68-times and 61-times longer than using javac and a standard contextual checker
respectively.

However, co-contextual FJ consistently yields high speedups for incremental checks.  In fact, it
only takes between $3$
and $21$
code changes until co-contextual type checking is faster overall.  In an interactive code editing
session where every keystroke or word could be considered a code change, incremental co-contextual
type checking will quickly break even and outperform a contextual type checker or javac.

The reason that the initial run of co-contextual FJ induces such high slowdowns is because the
occurrence of class requirements is far removed from the occurrence of the corresponding class
facts.  This is true for the \lstinline|!Nat!| classes that we merge with the synthesized code at
the top-most node as well as for dependencies from one root hierarchy to another one.  Therefore,
the type checker has to propagate and merge class requirements for a long time until finally
discovering class facts that discharge them.  We conducted an informal exploratory experiment that
revealed that the performance of the initial run can be greatly reduced by bringing requirements and
corresponding class facts closer together.  On the other hand, incremental performance is best when
the changed code occurs as close to the root node as possible, such that a change entails fewer
rechecks.  In future work, when scaling our approach to full Java, we will explore different layouts
for class declarations (e.g., following the inheritance hierarchy or the package structure) and for
reshuffling the layout of class declarations during incremental type checking in order to keep
frequently changing classes as close to the root as possible.

\subsection{Evaluation on real Java program}
\paragraph{Input data.}
We conduct an evaluation for our co-contextual type checking on realistic FJ programs.  We wrote
about $500$
SLOCs in Java, implementing purely functional data structures for binary search trees and red black
trees.  In the Java code, we only used features supported by FJ and mechanically translated the Java
code to FJ.  For evaluating the incremental performance, we invalidate the memoized results for the
three \lstinline|!Nat!| classes as in the experiment above.

\paragraph{Experimental setup.} Same as above.

\paragraph{Results.}
We show the measurements in milliseconds for the 500 lines of Java code. 

\begin{tabular}[t]{c@{\ \ }|l@{\ \ }|l@{\ \ }}
\toprule
  \multicolumn{1}{l|}{javac / contextual} & \multicolumn{1}{c|}{co-contextual init} &  \multicolumn{1}{c}{co-contextual inc} \\
    14.88 /  3.74  & 48.07 (0.31x / 0.08x) & 9.41 (1.6x / 0.39x) \\
\bottomrule
\end{tabular}
\ \\

Our own non-incremental contextual type checker is surprisingly fast compared to javac, and not even
our incremental co-contextual checker gets close to that performance.  When comparing javac and the
co-contextual type checker, we observe that the initial performance of the co-contextual type
checker improved compared to the previous experiment, whereas the incremental performance degraded.
While the exact cause of this effect is unclear, one explanation might be that the small input size
in this experiment reduces the relative performance loss of the initial co-contextual check, but
also reduces the relative performance gain of the incremental co-contextual check.


%% file: related_work.tex

\section{Related work}
\label{sec:related-work}

The work presented in this paper on co-contextual 
type checking for OO programming languages, specifically for Featherweight Java,
is inspired by the work on co-contextual type checking for 
PCF~\cite{Erdweg15}.
OO languages and FJ impose features like nominal typing, subtype polymorphism, and inheritance that are not covered in the work for co-contextual PCF~\cite{Erdweg15}.
In particular, here we developed a solution for merging and removing requirements in presence of nominal typing.

Introducing type variables as placeholders for the actual types of variables, classes, fields, methods is a known technique in type inference~\cite{pierce2002types,PierceT00}. The difference is that we introduce a fresh class variable for 
 each occurrence of a method $m$ or fields in different
branches of the typing derivation. Since fresh class variables are generated independently, no
coordination is required while moving up the derivation tree, ensuring context and class table independence. 
Type inference uses the context to coordinate type checking of $m$ in different branches, by using the same type variable.  In contrast to type inference where context and class table are available, we remove them (no actual context and class table). Hence, in type inference inheritance relation between classes and members of the classes are given, whereas in co-contextual FJ we establish these relations through requirements. That is, classes are required to have certain members with unknown types and unknown inheritance relation, dictated from the surrounding program. \\
 Also, in contrast to bidirectional type checking~\cite{Christiansen13,Dunfield13} that uses two sets of typing rules one for inference and one for checking, we use one set of co-contextual type rules, and the direction of type checking is all oriented bottom-up; types and requirements flow up. As in type inference, bidirectional type checking uses context to look up variables. Whereas co-contextual FJ has no context or class table, it uses requirements as a structure to record the required information on fields, methods, such that it enables resolving class variables of the required fields, methods to their actual types.

Co-contextual formulation of type rules for FJ is related to the work on
principal typing~\cite{Jim96,Wells02}, and especially to principal typing 
for Java-like languages~\cite{AnconaZ04}.
A principal typing~\cite{AnconaZ04} of each fragment (e.g., modules, packages, or classes) is associated 
with a set of type constraints on classes, which represents all other possible typings and
can be used to check compatibility in all possible contexts. That is, principle typing finds the strongest type 
of a source fragment in the weakest context. This is used for type inference and separate compilation in FJ. 
They can deduce exact type constraints using a type inference algorithm. 
We generalize this and do not only infer requirements on classes but also on method parameters and the current class.
Moreover, we developed a duality relation between the class table and class requirements that enables the systematic development of co-contextual type systems for OO languages beyond FJ.

Related to our co-contextual FJ is the formulations 
used in the context of compositional compilation~\cite{AnconaDDZ05} (continuation of the work on principal 
typing~\cite{AnconaZ04}) and the compositional explanation of type errors~\cite{Chitil01}. 
This type system~\cite{AnconaDDZ05} partially eliminate the class table, namely only inside a fragment, and 
 does not eliminate the context. Hence, type checking of parameters and \keyw{this} is coordinated and subexpressions are 
 coupled through dependencies on the usage of context. 
 In our work, we eliminate both class table (not only partially) and context, therefore all dependencies 
are removed. By doing so we can enable compositional compilation for
individual methods. 
To resolve the type constraints on classes, compositional compilation~\cite{AnconaZ04} needs a linker
in addition to an inference algorithm (to deduce exact type constraints), whereas, we use a 
constraint system and requirements.
We use duality to derive a co-contextual type system for FJ and we also ensure that both 
formulations are equivalent (\ref{sec:theorems-equivalence}). That is, we ensure that an expression, method, 
class, or program is well-typed in FJ if and only if it is well-typed in co-contextual FJ, and that all requirements are fulfilled. In contrast, 
compositional compilation rules do not check whether the inferred collection of constraints on classes 
is satisfiable; they actually allow to derive judgments for any fragment, even for those that are not statically correct. 

Refactoring for generalization using type constraints~\cite{Tip03,Tip11}
is a technique Tip~et~al.~used to manipulate types and class hierarchies to enable refactoring. 
That work uses variable type constraints
 as placeholders for changeable declarations.
They use the constraints to restrict when a refactoring can be performed. 
Tip~et~al.~are interested to find a way to represent the actual class hierarchy 
and to use constraints to have a safe refactoring and a well-typed program after refactoring. 
The constraint system used by Tip~et~al.~is specialized to refactoring, 
because different variable constraints and solving techniques are needed 
In contrast, in our work, we use class variables
 as placeholders for the actual type of required extends, constructors, fields, and methods of a class, in the lack of the class table. 
We want to gradually learn the information about the class hierarchy. 
We are interested in the type checking technique and how to co-contextualize it 
and use constraints for type refinement. 

 Adapton~\cite{Hammer14} is a programming language where the runtime system traces memory reads and writes and selectively replays dependent computations when a memory change occurs. In principle, this can be used to define an “incremental” contextual type checker. However, due to the top-down threading of the context, most of the computation will be sensitive to context changes and will have to be replayed, thus yielding unsatisfactory incremental performance. Given a co-contextual formulation as developed in our paper, it might be possible to define an efficient implementation in Adapton. 

 The works on smart/est recompilation~\cite{Shao93,Tichy86} has a different purpose from ours, namely to achieve separate compilation they need algorithms for the inference and also the linking phase which are specific to SML. In contrast, we use duality as a guiding principle to enable the translation from FJ to co-contextual FJ. This technique allows us to do perform a systematic (but yet not mechanical) translation from a given type system to the co-contextual one. Our type system facilitates incremental type checking because we decouple the dependencies between subexpressions and the smallest unit of compilation is any node in the syntax tree. Moreover, we have investigated optimizations for facilitating the early solving of requirements and constraints.


%% file: Conclusion.tex
\section{Conclusion and Future Work}
\label{sec:concl-future-work}

In this paper, we presented a co-contextual type system for FJ by  
transforming the typing rules in the traditional formulation into a form that replaces 
top-down propagated contexts and class tables with bottom-up 
propagated \emph{context and class table requirements}. We used duality as a technique to derive 
co-contextual FJ's typing rules from FJ's typing rules. To make the correspondence between class 
table and requirements, we presented class tables that are gradually extended with information from the class declarations,
and how to map operations on contexts and class tables to their dual operations on context and class table requirements.
To cover the OO features of nominal typing, subtype polymorphism, and implementation inheritance,
co-contextual FJ uses conditional requirements, inequality conditions, and conditional constraints. 
Also, it changes the set of requirements by adding requirements 
with the different receiver from the ones defined 
by the surrounding program, 
in the process of merging and removing requirements as the type checker moves upwards and 
discovers class declarations. 
We proved the typing equivalence of expressions, 
methods, classes, and programs between FJ and co-contextual FJ.  

The co-contextual formulation of FJ
typing rules enables incremental type checking 
because it removes dependencies between subexpressions.
We implemented an incremental co-contextual FJ type checker. Also, we evaluated its performance on synthesized programs up to $1243$~FJ classes and 500 SLOCs of java programs.

There are several interesting directions for future work.
In short term, we want to explore parallel co-contextual type checking for FJ. 
A next step would be to develop a co-contextual
type system for full Java. 
Another interesting direction is to investigate co-contextual formulation for gradual type systems.


%% file: theorems+defs.tex
\section{Auxiliary definitions; merge, add, remove}
\label{sec:appendixA}

%
%
We give the definition of $merge_{CR}$ for all cases of the clause definition
\footnote{Merge operation for optional methods is the same as merge for methods.}.
\begin{figure}[htbp]
    \begin{flalign*}
    	& {merge_{CR}(CR_1, CR_2) = CR|_S} \\
        & where\; CR = \{((CR_1 \setminus (\seq{T_1.extends: T_2, cond_1} \cup \seq{T_1.init(\seq T_1), cond_1} \cup \seq{T_1.f: T_2, cond_1} \\
        & \hskip6em  \cup \seq{T_1.m: \seq T_1\rightarrow T_2, cond_1} ) \cup ((CR_2 \setminus (\seq{T_2.extends: T_3, cond_2} \\
        &\hskip6em \cup \seq{T_2.init(\seq T_2), cond_2} \cup \seq{T_2.f: T_3, cond_2} \cup 
        \seq{T_2.m: \seq T_2\rightarrow T_3, cond_2} )\\
        & \hskip6em  \cup CR_e \cup CR_k \cup CR_f \cup CR_m \}\\
        & \hskip4em S = S_e \cup S_k \cup S_f \cup S_m
    \end{flalign*}
    \vskip-2.5em
  	\begin{flalign*}
		& where \hskip1em CR_e =\{ (T_1.extends : T'_1, cond_1 ,\ (T_1 \neq T_2)) ,\ (T_2.extends : T'_2, cond_2 \cup \\       
		& \hskip7em (T_1 \neq T_2)), (T_1.extends : T'_1, (cond_1 \cup cond_2  \cup (T_1 = T_2)) \\
        & \hskip7em  \WHERE (T_1.extends : T'_1, cond_1) \in CR_1  \wedge (T_2.extends : T'_2, cond_2) \in CR_2\}\\
        & \hskip5em S_e =\{(T'_1 = T'_2\ if\ T_1 = T_2)  \WHERE (T_1.extends : T'_1, cond_1) \in CR_1 \\
        & \hskip7em \wedge (T_2.extends : T'_2, cond_2) \in CR_2\}
	\end{flalign*}
	\vskip-2.5em
	\begin{flalign*}
		& where \hskip1em CR_k =\{ (T_1.init(\seq{T_1}), cond_1 \cup (T_1\neq T_2)) \cup(T_2.init(\seq{T_2}), cond_2 \cup (T_1 \neq T_2 )) \\
		& \hskip7em (T_1.init(\seq{T_1}), cond_1 \cup cond_2 \cup (T_1 = T_2)) \WHERE (T_1.init(\seq{T_1}), cond_1) \in CR_1 \\
        & \hskip7em \wedge (T_2.init(\seq{T_2}), cond_2) \in CR_2 \}\\
        & \hskip5em S_k =\{(\seq{T_1} = \seq{T_2}\ if\ T_1 = T_2)  \WHERE (T_1.init(\seq{T_1}), cond_1) \in CR_1\\
        & \hskip7em \wedge (T_2.init(\seq{T_2}), cond_2) \in CR_2 \}
    \end{flalign*}
    \vskip-2.5em
    \begin{flalign*}
    	& where \hskip1em CR_f =\{ (T_1.f : T'_1, cond _1\cup (T_1 \neq T_2)) \cup (T_2.f : T'_2, cond_2 \cup (T_1 \neq T_2)) \\
        & \hskip7em \cup (T_1.f : T'_1, cond_1 \cup cond_2 \cup (T_1 = T_2))  \WHERE (T_1.f : T'_1, cond_1) \in CR_1 \\
        & \hskip7em \wedge (T_2.f : T'_2, cond_2) \in CR_2 \}\\
        & \hskip5em S_f =\{(T'_1 = T'_2\ if\ T_1 = T_2)  \WHERE (T_1.f : T'_1, cond_1) \in CR_1 \\
        & \hskip7em \wedge (T_2.f : T'_2, cond_2) \in CR_2 \}
    \end{flalign*}
    \vskip-2.5em
    \begin{flalign*}
    	& where \hskip1em CR_m =\{ (T_1.m: \seq{T_1}\rightarrow  T'_1, cond_1 \cup (T_1 \neq T_2)) \cup (T_2.m:\seq{T_2}\rightarrow T'_2, cond_2  \\
        & \hskip7em  (T_1 \neq T_2)) \cup (T_1.m: \seq{T_1}\rightarrow T'_1, cond_1 \cup cond_2 \cup (T_1 = T_2))  \WHERE (T_1.\\
        & \hskip7em m :  \seq{T_1}\rightarrow T'_1, cond_1) \in CR_1\wedge (T_2.m : \seq{T _2}\rightarrow T'_2, cond_2) \in CR_2  \}\\
        & \hskip5em S_m =\{(T'_1 = T'_2\ if\ T_1 = T_2) \cup (\seq{T_1} = \seq{T_2}\ if\ T_1 = T_2) \WHERE (T_1.m: \seq{T_1}\rightarrow T'_1, \\
        & \hskip7em   cond_1) \in CR_1 \wedge (T_2.m : \seq{T_2} \rightarrow T'_2, cond_2) \in CR_2\} 
     \end{flalign*}
     \vspace{-1cm}
\end{figure}
\newpage

Next we define add and remove operations for all cases of the clause definition. 
\begin{figure}[htbp]
\begin{flalign*}
&addExt(class\ C\ \keyw{extends}\ D, CT)= ( C extends D)\cup CT \\
&removeExt(class\ C\ \keyw{extends}\ D, CR) = CR' |_S \\
&where\ CR'= 
\!\begin{aligned}[t]
&\{(T.extends: T',cond \cup (T \neq C)) \WHERE (T.extends: T', cond) \in CR \} \\
&\cup \{(T.m: \seq{T}\rightarrow T', cond\cup(T\neq C)) \\
& \hskip0.7em \cup(D.m:\seq{T} \rightarrow T', cond \cup(T = C))\WHERE  (T.m: \seq{T}\rightarrow T', cond) \in CR \}\\
&\cup \{(T.m: \seq{T}\rightarrow T', cond \cup (T\neq C))_{opt} \\
& \hskip0.7em \cup (D.m: \seq{T} \rightarrow T', cond \cup (T = C))_{opt} \\
& \hskip0.7em \WHERE  (T.m: \seq{T}\rightarrow T', cond)_{opt} \in CR \} \\
&\cup \{(T.f:  T', cond \cup (T\neq C)) \cup (D.f: T', cond \cup (T = C))\\
& \hskip0.7em \WHERE  (T.f: T', cond) \in CR \} \\
& \hskip-2em S = \{(T' = D\ if\ T= C)  \WHERE (T.extends :T', \cond) \in CR \} 
\end{aligned}
\end{flalign*}
\vskip-1em
\begin{flalign*}
&addCtor(C,  (\seq{D}\ \seq{g}, \seq{C}\ \seq{f}), CT) = (\constrCT{C}{\seq{D}; \seq{C}})\cup CT \\
&removeCtor(C,(\seq{D}\ \seq{g}, \seq{C}\ \seq{f}), CR) = CR'|_{S} \\
&where \hskip1em CR' = \{ (\constrCT{T}{\seq{T}}),cond \cup (T \neq C)) \WHERE (\constrCT{T}{\seq{T}}), \cond) \in CR  \} \\
& \hskip7em \cup (CR\setminus \seq{(\constrCT{T}{\seq{T}}), \cond)} )\\
& \hskip5em S = \{(\seq{T} = \seq{D}\ \seq{C}\ if\ T = C)  \WHERE (\constrCT{T}{\seq{T}}), \cond) \in CR \} 
\end{flalign*}
\vskip-2em
\begin{flalign*}
&addFs(C, \seq{C_f}\ \seq{f}, CT) = \overline{C.f : C_f} \cup CT \\
&removeF(C, C_f\  f, CR) = CR'|_{S} \\
&where \hskip1em CR' = \{(T.f : T',cond \cup (T \neq C)) \WHERE (T.f : T', \cond) \in CR \} \\
&\hskip7em \cup (CR\setminus \seq{(T.f : T', \cond)})\\
& \hskip5em S = \{(T' = \ C_f\ if\ T= C)  \WHERE (T.f : T', \cond) \in CR \}\\
&removeFs(C,\overline{C_f\  f}, CR) = CR'|_{S} \\
&where \hskip1em CR' = \{ CR_f \WHERE (C_f\ f) \in \overline{C_f\ f}\wedge \mathit{removeF}(CR,C, C_f\  f) = CR_f|_{S_f} \} \\
& \hskip5em S = \{S_f \WHERE  (C_f\ f) \in \overline{C_f\ f}  \wedge\mathit{removeF}(CR,C, C_f\  f) = CR_f|_{S_f}\}
\end{flalign*}
\vskip-2em
\begin{flalign*}
& \func{addMs}(C, \overline{M}, CT) = \overline{C.m: \seq{C} \rightarrow C'} \cup CT \\
&\func{removeM}( C , C'\ m(\seq{C}\ \seq{x}) \ \{\keyw{return} \ e\}, CR) = CR'|_{S} \\
&where \; CR' = \{(T.m:\seq{T}\rightarrow T',cond \cup (T \neq C)) \WHERE (T.m :\seq{T}\rightarrow T', \cond) \in CR \}\\ 
& \hskip6em \cup (CR\setminus (T.m :\seq{T}\rightarrow T', \cond))\\
& \hskip2em S = \{(T' =  C' \ if\ T= C) \cup (\seq{T = C}\ if\ T= C) \WHERE (T.m :\seq{T}\rightarrow T', \cond) \in CR \} \\
\end{flalign*}
\vspace{-1cm}
\end{figure}
\begin{figure}[htbp]
\begin{flalign*}
&\func{removeMs}(C, \seq{M}, CR ) = CR'|_{S}  \\
& where\ CR' = \{ CR_m \WHERE (C'\ m(\seq{C}\ \seq{x}) \ \{return \ e\}) \in\overline{M} \\
& \hskip6em \wedge \func{removeM}(C, C'\ m(\seq{C}\ \seq{e}) \ \{\keyw{return} \ e\}, CR) = CR_m|_{S_m} \}   \\
& \hskip4em S = \{S_m \WHERE (C'\ m(\seq{C}\ \seq{x}) \ \{return \ e\}) \in\overline{M}  \\
& \hskip6em \wedge \func{removeM}(C, C'\ m(\seq{C}\ \seq{x}) \ \{\keyw{return} \ e\}, CR) = CR_m|_{S_m} \} \\
& removeOptM(C, C'\ m(\seq{C}\ \seq{x}) \ \{return \ e\}, CR) = CR'|_{S} \\
& where \hskip1em CR' = \{(T.m :\seq{T}\rightarrow T',cond \cup (T \neq C))_{opt} \WHERE (T.m :\seq{T}\rightarrow T', \cond)_{opt} \in CR \} \\
& \hskip7em \cup (CR\setminus (T.m :\seq{T}\rightarrow T', \cond_{opt}))\\
& \hskip2em S = \{(T' =  C' \ if\ T= C) \cup (\seq{T} = \seq{C}\ if\ T= C) \WHERE (T.m :\seq{T}\rightarrow T', \cond)_{opt} \in CR \} \\
& removeOptMs(C, \overline{M}, CR) = (CR' \cup (CR\setminus CR'))|_{S} \\
& where \hskip1em CR' = \{ CR_m \WHERE (C'\ m(\seq{C}\ \seq{x}) \ \{return \ e\}) \in\overline{M} \\
& \hskip7em \wedge \mathit{removeOptM}(CR,C, C'\ m(\seq{C}\ \seq{e}) \ \{return \ e\}) = CR_m|_{S_m} \} \\
& \hskip5em S = \{S_m \WHERE (C'\ m(\seq{C}\ \seq{x}) \ \{return \ e\}) \in\overline{M} \\
& \hskip7em \wedge removeOptM(CR,C, C'\ m(\seq{C}\ \seq{x}) \ \{return \ e\}) = CR_m|_{S_m} \}
\end{flalign*}
\vspace{-1cm}
\end{figure}
 
\newpage
  
\section{Equivalence of Contextual and Co-Contextual FJ}  
  \label{sec:appendixB}
 In here we describe a detailed proof of typing equivalence between FJ and co-contextual FJ. 
Co-contextual FJ is constraint based type system. We present the formal definitions for substitution, and Figures \ref{fig:proj}, \ref{fig:projL} give formal definition how to retrieve the immediate subclass relation $\Sigma$ from rep. class table, and a list of class declaration. That is, a projection from class table/list of declarations to a set of tuples, which represent the relation between two classes in an extends clause. 
\input{projExt}

 \begin{definition}[Subtyping relative to $\Sigma$]
   Let $\Sigma$ be a binary relation on class names, $C$, $D$ class names. Then $C$ is a subtype of
   $D$ \emph{relative to $\Sigma$} ($C \subtpe_\Sigma D$), if and only if $(C,D)\in\Sigma^*$, where
   $\Sigma^*$ is the reflexive, transitive closure of $\Sigma$.
 \end{definition}
 
\begin{definition}[Substitution $\sigma$]
  Given sets of context and class requirements $R,\ CR$, $\sigma$ is a set of mappings from class variables to class types, i.e., $\sigma = \{ U\smap C \WHERE U \in freshU(R) \cup freshU(CR) \}$.
 \end{definition}

 \begin{definition}[Constraint Satisfaction]
   Let $s$ be a constraint on class types, $\sigma$ a substitution from class variables to class types, $\Sigma$ a binary relation
   on class names. The pair $(\Sigma, \sigma)$ \emph{satisfies $s$} ($\func{sat}(\Sigma, \sigma, s)$) if and only if one
   of the following holds:
    \begin{enumerate}
   \item If $s = (T \subtpe T')$, then $T\sigma \subtpe_\Sigma T'\sigma$.
    \item If $s = (T = T')$, then $T\sigma = T'\sigma$.
    \item If $s = (T = T'\ \term{if}\ \mathit{cond})$ and for all $s'\in \mathit{cond}$, $\func{sat}(\Sigma, \sigma, s')$ then $T\sigma = T'\sigma$.
    \item If $s = (T \neq T')$, then $T\sigma \neq T'\sigma$.
   \item If $s = (T \not{\subtpe} T')$, then $T\sigma \not{\subtpe_\Sigma} T'\sigma$.
   \end{enumerate}
   
 \end{definition}

\begin{assump}[Properties of $\func{solve}$]
   Let $\Sigma$ be a binary relation on class names, $S$ a set of constraints on class types:
   \begin{enumerate}    
   \item $\func{solve}(\Sigma, S)$ terminates.
   \item If $\func{solve}(\Sigma, S) = \sigma$.  Then for all $s \in S$,$\func{sat}(\Sigma, \sigma, s)$.
     
   \item If $\func{solve}(\Sigma, S) = \bot$. Then there exists $s \in S$, where $\func{sat}(\Sigma, \sigma, s)$ does not hold.
\end{enumerate}
\end{assump}

\begin{definition}[Ground context requirement]
 $\sigma(R)$ is ground, if for all $(x: T) \in R$ then $\sigma(T)$ is ground.
 \end{definition}
 
 \begin{definition}[Ground class table requirements]
 $\sigma(CR)$ is ground, if for all $(CReq, cond) \in CR$ then $\sigma(CReq)$ is ground and $\sigma(cond)$ in ground.
  \end{definition}
  
   \begin{definition}[Ground class requirement]
   	\begin{equation}
 \sigma(CReq)\ ground =  \begin{cases}
 	\sigma(T.extends : T')\ ground & \mbox{if } (CReq)= (T.extends T') \wedge\\
 	        &    \quad  \sigma(T),\ \sigma(T')\ ground  \\
 	\sigma(T.f: T')\ ground & \mbox{if } (CReq) = (T.f: T') \wedge \\
 	  &    \quad  \sigma(T),\ \sigma(T')\ ground  \\
    \sigma(T.m: \seq{T}\rightarrow T' )\ ground & \mbox{if } (CReq) = (T.m : \seq{T}\rightarrow T' )\wedge \\
     &    \quad  \sigma(T),\ \sigma(\seq{T})\ ground \\
      &    \quad  \wedge \sigma(T')\ ground  \\
   \sigma(T.init:\seq{T})\ ground  & \mbox{if } (CReq) = (T.init(\seq{T})) \wedge\\
     &    \quad  \sigma(T),\ \sigma(\seq{T})\ ground
 \end{cases}
\end{equation}
\end{definition}

\begin{definition}[Ground conditions]
$\sigma(cond)$ is ground, if for all $(T = T'), (T'' \neq T^*) \in cond$ then $\sigma(T)$, $\sigma(T')$, $\sigma(T'')$, $\sigma(T^*)$ are ground.
\end{definition}
 
 \begin{definition}[Ground Solution $\sigma$]
  For a given type $T$, a set of constraints $S$, where $\sigma = solve(S)$, we lift substitution $\sigma$ to sets of context requirements $R$, class requirements $CR$ and $\sigma$ is a ground solution if:
   \begin{enumerate}[1)]
 	\item $\sigma(T)$ is ground 
 	\item $\sigma(R)$ is ground 
 	\item $\sigma(CR)$ is ground 
 \end{enumerate} 	
 \end{definition}  
  
The two first rules of Figure.\ref{fig:satisfy} define the field lookup and extends lookup. The other rules formally define the relation between the class table and class table requirements. We assume that class table requirements are ground.

  \input{satisfy-judg}

\begin{lemma}
\label{lem:merge-cons}
  Let $merge_R(R_1, R_2)=R|_S$, $\Gamma \supseteq \sigma_1(R_1)$, $\Gamma
  \supseteq \sigma_2(R_2)$, and $\sigma_1(R_1)$,  $\sigma_2(R_2)$  are ground.
  Then $\sigma_1\circ\sigma_2$ solves $S$.
\end{lemma}
\begin{proof}
  By the definition of $merge_R$, $S = \{R_1(x) = R_2(x) \WHERE x \in \dom{\reqs_1} \cap \dom{\reqs_2}\}$.
  Since $\Gamma \supseteq \sigma_i(R_i)$, we know $\Gamma(x) = \sigma_i(R_i(x))$ for all $x \in \dom{R_i}$.
  In particular, $\Gamma(x)=\sigma_1(R_1(x))=\sigma_2(R_2(x))$ for all $x \in \dom{R_1}\cap\dom{R_2}$.
  Thus, $\sigma_1\circ\sigma_2$ solves $C$ because
  $(\sigma_1\circ\sigma_2)(R_1(x))=\sigma_1(R_1(x))=\sigma_2(R_2(x))=(\sigma_1\circ\sigma_2)(R_2(x))$ for all $x \in \dom{R_1}\cap\dom{R_2}$, because 
  $\sigma_1(R_1)$ and $\sigma_2(R_2)$ are ground.
\end{proof}

\begin{lemma}
\label{lem:mergeCR-cons}
  Let $merge_{CR}(CR_1, CR_2)=CR|_S$, $\sigma_1(CR_1)$, $\sigma_2(CR_2)$ are ground, and $CT \satisfy\ \sigma_1(CR_1)$, $CT
  \satisfy\ \sigma_2(CR_2)$.
  Then $\sigma_1\circ\sigma_2$ solves $S$.
\end{lemma}
\begin{proof}
  By the definition of $merge_{CR}$, $S = S_c \cup S_e \cup S_k \cup S_f \cup S_m$, where $ S_f = \{(T'_1 = T'_2\ if\ T_1 = T_2) \WHERE (T_1.f: T'_1, cond_1) \in CR_1 \wedge (T_2.f : T'_2, cond_2) \in CR_2\}$. 
    
  Since $CT\ satisfy\ \sigma_i(CR_i)$, we know $\sigma_i(T_i.f : T'_i, cond_i) \in CT$, for all $f \in \dom{CR_i}$, where $\sigma_i(cond_i)\ hold$.
  In particular, for all $f \in \dom{CR_1}\cap\dom{CR_2}$, where $(T_1.f : T'_1, cond_1) \in CR_1$, $(T_2.f : T'_2, cond_2) \in CR_2$,
   $\sigma_1(T'_1) = \sigma_2(T'_2)\ if\ \sigma_1(T_1) = \sigma_2(T_2)$.
  Thus, $\sigma_1\circ\sigma_2$ solves $S$ because
  $(\sigma_1\circ\sigma_2)(T'_1)=\sigma_1(T'_1)=\sigma_2(T'_2)=(\sigma_1\circ\sigma_2)(T'_2)$, if $\sigma_1(T_1) = \sigma_2(T_2)$, because 
  $\sigma_1(CR_1)$ and $\sigma_2(CR_2)$ are ground.
  
  The same procedure we follow for methods, i.e., a given method $m$ that we find a match in $CR_1(C)$, and $CR_2(C)$, $S_m$ is the set of constraints for the method as result of unifying return type and types of the parameters from the two different class requirements($CR_1, CR_2$). 
\end{proof}

\begin{lemma}
\label{lem:satisfy}
If \textit{CT satisfy }$\sigma_1(CR_1)$, $\sigma_1(CR_1)$ is ground,  and \textit{CT satisfy $\sigma_2 (CR_2)$}, $\sigma_2(CR_2)$ is ground, then \textit{CT satisfy $\sigma(CR)$}, where $\sigma = \sigma_1 \circ \sigma_2$ and $CR_S = merge_{CR}(CR_1, CR_2)$.
\end{lemma}
\begin{proof}
First we have to show that the new set of constraints $S$ generated from merging is solvable, and this holds by Lemma~\ref{lem:mergeCR-cons}.\\ 
Then we show that $CT\ \satisfy\ \sigma(CR)$.
 For sake of brevity we consider clauses common in both requirements sets $CR_1$ and $CR_2$. Let us consider the field $f$, such that $(T_1.f : T'_1, cond_1) \in CR_1$ and $(T_2.f : T'_2, cond_2)$, and by assumption we have that $CT\ satisfy\ \sigma_1((T_1.f : T'_1, cond_1))$ and $CT\ \satisfy \ \sigma_2(T_2.f : T'_2, cond_2)$. By the definition of $merge_{CR}$ 
	 the conditions of these two requirements are updated, i.e., $(T_1.f : T'_1,cond_1 \cup (T_1 \neq T_2 )) $ and $(T_2.f : T'_2,cond_2 \cup (T_1 \neq T_2)) $, and a new requirement is added, i.e., $(T_1.f : T'_1, cond_1 \wedge cond_2 \cup (T_1 = T_2)$. 
	  Suppose that CT satisfies the three of the new and updated requirements, namely all their conditions should hold by rule \rulename{T-Satisfy}, but this is contradiction, because two types cannot be at the same time not equal and equal. Therefore there are two possibilities:
	  \begin{enumerate}[1)]
	  	\item either the conditions of the updated field requirements hold, i.e., $(T_1.f : T'_1,cond_1 \cup (T_1 \neq T_2 )) $, $(T_2.f : T'_2,cond_2 \cup (T_1 \neq T_2)) $, and $(T_1\neq T_2)$ holds.
	  	\item or the conditions of the new field requirement hold, i.e., $(T_1.f : T'_1, cond_1 \wedge cond_2 \cup (T_1 = T_2))$, and $(T_1 = T_2)$ holds. 
	  \end{enumerate}
	  \begin{itemize}
	  	\item If $1)$ is possible then $CT\ \satisfy \sigma_1\circ\sigma_2(T_1.f : T'_1,cond_1 \cup (T_1 \neq T_2 ) \cup T_2.f : T'_2,cond_2 \cup (T_1 \neq T_2))$ because by assumtion $CT\ satisfy\ \sigma_1((T_1.f : T'_1, cond_1))$ and $CT\ \satisfy \ \sigma_2(T_2.f : T'_2, cond_2)$. The new class requirement $(T_1.f : T'_1, cond_1 \wedge cond_2 \cup (T_1 = T_2))$ is satisfiable by default since one of its conditions $(T_1 = T_2)$ does not hold, namely is not a valid requirement.
	  	\item If $2)$ is possible then $CT\ \satisfy \sigma_1\circ\sigma_2(T_1.f : T'_1, cond_1 \wedge cond_2 \cup (T_1 = T_2))$, because $(T_1=T_2)$, $CT\ satisfy\ \sigma_1((T_1.f : T'_1, cond_1))$ and $CT\ \satisfy \ \sigma_2(T_2.f : T'_2, cond_2)$. The updated class requirements $(T_1.f : T'_1,cond_1 \cup (T_1 \neq T_2 )) $ and $(T_2.f : T'_2,cond_2 \cup (T_1 \neq T_2)) $ are satisfiable by default since one of their conditions $(T_1 \neq T_2)$ does not hold, namely are not valid requirements.

	  \end{itemize}
	  As a result $CT$ satisfies the resulting set of requirements after merging for the given field $f$.
	  	
	The same we argue for methods, optional methods, current class, and extend clauses. 
\end{proof}

%
%
%

\begin{prop}[Independent derivation in co-contextual type checking]\label{prop:indep}
Given a set of otherwise independent derivations of class requirement $CR = \{CR_1 \cup \ldots \cup CR_n \}$, $\forall i, j \in [1..n].\ freshU(CR_i)\cap freshU(CR_j) = \emptyset $, where $freshU(CR_i) = \{U_1^i, \ldots U_n^i \}$
\end{prop}

\begin{proof}
	It is straightforward by the rules and how the type checking is performed, i.e., for every rules of the type checking we always introduce fresh class names $U$, therefore $U$s in one derivation do not appear to another independent derivation. 
\end{proof}
 
 \begin{corollary}[Associative feature for substitution]
 \label{cor:assoc}
 Given $CR$, $\sigma_1$	and $\sigma_2$ then it holds that $(\sigma_1 \circ \sigma_2)(CR) = (\sigma_2 \circ \sigma_1)(CR)$
 \end{corollary}
 \begin{proof}
 	Follows directly from Proposition~\ref{prop:indep}.
 \end{proof}


\begin{definition}[Correspondence relation for expressions]
Given judgments \judge{\ctxplus{\Gamma}{CT}}{e}{C}, \cojudge{e}{T}{S}{R}{CR}, and $\func{solve}(\Sigma ,S) = \sigma$, where 
$\projExt{CT}=\Sigma$. 
The correspondence relation between $\Gamma$ and $R$, $CT$ and $CR$, written $(C, \Gamma, CT)\vartriangleright \sigma(T, R, CR)$, is defined as: 
\begin{enumerate}[a)]
	\item $C = \sigma(T)$
	\item $\Gamma \supseteq \sigma(R)$
	\item $CT\ \satisfy \ \sigma(CR)$ 
\end{enumerate}	
\end{definition}
 
\begin{theorem}[Equivalence of expressions: $\Rightarrow$]
Given $e,\ C, \ \Gamma ,\ CT, $ if $\judge{\ctxplus{\Gamma}{CT}}{e}{C}$, 
then there exists $T, \ S,\ R,\ CR,\ \Sigma,\ \sigma$, where $\projExt{CT}=\Sigma$ and $\func{solve}(\Sigma,S)=\sigma$, such that $\cojudge{e}{T}{S}{R}{CR}$ holds, $\sigma$ is a ground solution and $(C, \Gamma, CT)\vartriangleright \sigma(T, R, CR)$ holds. 
\label{theo:ExprR}
 \end{theorem}

\begin{proof}
We proceed by induction on the typing judgment of expression $e$.
\begin{itemize}
\item Case \rulename{T-Var} with \judge{\ctxplus{\Gamma}{CT}}{x}{C}.

    By inversion, $\Gamma(x)=C$.

    Let $U$ fresh, $S=\emptyset$, $R=\{x : U\}$, $CR = \emptyset$ and $\sigma =\{U \smap C\}$.

    Then $\cojudge{e}{C'}{S'}{R} {CR}$ holds by rule \rulename{TC-Var}. Since $S = \emptyset$, then $\sigma$ solves $S$.
    $\sigma$ is ground solution because: 
    \begin{enumerate}[1)]
    	\item $\sigma(U)$ is ground because $\sigma(U)= C$.
    	\item $R = \{x :U\}$  and $\sigma = \{U \smap C \}$ implies $\sigma(R) = \{x : C\}$ is ground. 
    	\item $CR = \emptyset$ implies that $\sigma(CR) = \emptyset$ is ground. 
    \end{enumerate}
    
    The correspondence relation holds because;
    \begin{enumerate}[a)]
    	\item $C = \sigma(U)$
    	\item Since $\Gamma(x) = C$ by inversion, then $\Gamma \supseteq \{x : C\} = \sigma(R)$.
    	\item $CR = \emptyset$ and $\sigma(CR) = \emptyset$ implies that $CT\ \satisfy\ \sigma(CR) $.
    \end{enumerate}
\vskip2ex
  \item Case \rulename{T-Field} with \judge{\ctxplus{\Gamma}{CT}} {e.f_i}{C_i}.
    
   By inversion, \judge {\ctxplus{\Gamma}{CT}} {e}{C_e} and $field(f_i, C_e, \ctxcolor{CT}) = C_i$.     
   By IH, \cojudge {e} {T'_e} {S_e} {R_e} {CR_e}, where $\mathit{solve}(\projExt{CT} , S_e) = \sigma_e$, $\sigma_e(T'_e)$, $\sigma_e(R_e)$, 
  $\sigma_e(CR_e)$ are ground and the correspondence relation holds, i.e.,
   $ C_e = \sigma_e(T'_e)$, $\Gamma \supseteq \sigma_e(R_e)$, $CT\ \satisfy\ \sigma_e(CR_e)$. 
   
   Let $U$ be fresh, $CR|_{\typectxcolor{S_f}} = merge_{CR}(CR_e, (T'_e.f_i : U,\emptyset )) $, $S = S_e \cup \typectxcolor{S_f}$ 
   and $\sigma = \{U \smap C_i\}\circ \sigma_e $.
 
  Then \cojudge {e.f_i} {U} {S} {R_e} {CR} holds by rule \rulename{TC-Field}.\\
  $\sigma$ solves $S$ because it solves $S_e$ and $S_f$ as shown below: 
  \begin{itemize}
  	\item $solve(\projExt{CT}, S_e) =\sigma_e$ by IH and $\sigma = \{U \smap C_i\}\circ \sigma_e $ implies $\sigma$ solves $S_e$
  	\item $\sigma_e(CR_e)$ is ground by $IH$. \\
  	$(\ast)$ $\sigma(T'_e.f_i : U,\emptyset )$ is ground, because
    $\sigma(T'_e.f_i : U) = (\sigma(T'_e).f_i : \sigma(U)) = (C_e.f_i : C_i)$ and $C_e.f_i : C_i$ is ground.\\
     $CT \ \satisfy\ \sigma_e(CR_e)$ by $IH$. \\
     $(\ast\ast)$ $CT\ satisfy\ \sigma(T'_e.f : U,\emptyset )$ because $field(f_i, C_e, \ctxcolor{CT}) = C_i$ hence by rule \rulename{S-Field} holds that $CT\ satisfy\ (C_e.f : C_i,\emptyset )$, and $\sigma(T'_e.f_i : U) = C_e.f_i : C_i$.\\
     As a result by Lemma~\ref{lem:mergeCR-cons} $\sigma$ solves $S_f$.
  \end{itemize}
  $\sigma$ is a ground solution because: 
   \begin{enumerate}[1)]
   	\item  $\sigma(U)$ is ground because $\sigma(U) = C_i$.
   	\item $\sigma(R_e)$ is ground because $\sigma(R_e) = (\{U \rightarrow C_i \} \circ \sigma_e )(R_e) = \{U\rightarrow C_i \} (\sigma_e(R_e) )$,  
   since $\sigma_e(R_e)$ is ground by $IH$ then $\{U\rightarrow C_i \} (\sigma_e(R_e) ) = \sigma_e(R_e)$, i.e., $\sigma(R_e) = \sigma_e(R_e)$. 
   	\item $\sigma(CR_e)$ is ground because $\sigma(CR_e) = (\{U \rightarrow C_i \}\circ \sigma_e )(CR_e) = \{U\rightarrow C_i \} (\sigma_e(CR_e))$, 
   	$\{U\rightarrow C_i \} (\sigma_e(CR_e) ) = \sigma_e(CR_e)$ because $\sigma_e(CR_e)$ is ground by $IH$. 
     	$\sigma(T'_e.f_i : U,\emptyset )$ is ground by $(\ast)$. As a result $\sigma(CR)$ is ground by definition of $merge_{CR}$.
   \end{enumerate} 

   The correspondence relation holds because: 
     \begin{enumerate}[a)]
  	\item $C_i = \sigma(U)$
  	\item $\Gamma \supseteq \sigma(R_e)$, because $\Gamma \supseteq \sigma_e(R_e)$ by $IH$, and from $2)$ $\sigma(R_e) = \sigma_e(R_e)$.
  	\item  $CT\ satisfy\ \sigma(CR_e)$, because $CT\ satisfy\ \sigma_e(CR_e)$ by $IH$, and from $3)$ $\sigma(CR_e) = \sigma_e(CR_e)$.
  	$CT\ satisfy\ \sigma(T'_e.f : U,\emptyset )$ by $(\ast\ast)$. As a result
  	$CT\ satisfy\  \sigma(CR_e) \cup \sigma(T'_e.f : U, \emptyset )$, i.e., $CT\ satisfy\  \sigma(CR)$ by Lemma~\ref{lem:satisfy}.
  \end{enumerate}  
  
  \vskip2ex
  \item Case \rulename{T-Invk} with \judge{\ctxplus{\Gamma}{CT}}{e.m(\seq{e})} {C}.

  By inversion, \judge {\ctxplus{\Gamma}{CT}} {e}{C_e},  $mtype(m, C_e, CT)= \seq{D}\rightarrow C$, \judge{\ctxplus{\Gamma}{CT}} {\seq{e}} {\seq{C}}
   and $\seq{C} <: \seq{D}$.
    
  By IH, \cojudge {e} {T_e} {S_e} {R_e} {CR_e}, where $\mathit{solve}(\projExt{CT} ,S_e) = \sigma_e$, $\sigma_e(T_e)$, $\sigma_e(R_e),\ \sigma_e(CR_e)$ are ground and the correspondence relation hold, i.e, $ C_e = \sigma_e(T'_e)$, 
  $\Gamma \supseteq  \sigma_e(R_e)$, $CT\ satisfy\  \sigma_e(CR_e)$.\\
   By IH \cojudge {\seq{e}} {\seq{T}} {\seq{S}} {\seq{R}} {\seq{CR}}, $\forall i \in [1 .. n]$. $\mathit{solve}(\projExt{CT} ,S_i) = \sigma_i$, $\sigma_i(T_i)$, $\sigma_i(R_i),\ \sigma_i(CR_i)$ are ground, and the 
  correspondence relation holds, i.e., $C_i = \sigma_i(T_i)$, $\Gamma \supseteq  \sigma_i(R_i)$, $CT\ satisfy\  \sigma_i(CR_i)$.
    
  Let $U'$, $\seq{U}$ be fresh, $ R|_{S_r} = merge_R(R_e, R_1, \ldots, R_n) $, \\
  $CR|_{\typectxcolor{S_{cr}}} = merge_{CR}(CR_e,CR_1, \ldots, CR_n, (T_e.m : \seq{U} \rightarrow U', \emptyset )) $, 
  $S =S_e \cup \seq{S} \cup S_r \cup \typectxcolor{S_{cr}} \cup \{\seq{T}<: \seq{U}\}$ and 
  $\sigma = \{U' \smap C\}\circ \{U_i \smap D_i \}_{i\in [1..n]}\circ \sigma_e \circ \{\sigma_i\}_{i\in [1..n] }$
   
     Then \cojudge {e.m(\seq{e})} {U} {S} {R} {CR} holds by rule \rulename{TC-Invk}.\\
      $\sigma$ solves $S$ because it solves $S_e$, $\seq{S}$, $\typectxcolor{S_{r}}$, $S_{cr}$, and $\{\seq{T'}<: \seq{U}\}$ as shown below:
      \begin{itemize}
      \item $solve(\projExt{CT},S_e) = \sigma_e$ and $\sigma = \{U \smap C\}\circ \{U_i \smap D_i \}_{i\in [1..n]}\circ \sigma_e \circ \{\sigma_i\}_{i\in 1..n }$
implies that $\sigma$ solves $S_e$
\item $\{solve(\projExt{CT}, S_i) = \sigma_i\}_{i \in 1..n}$ and $\sigma = \{U \smap C\}\circ \{U_i \smap D_i \}_{i\in [1..n]}\circ \sigma_e \circ \{\sigma_i\}_{i\in [1..n] }$
implies that $\sigma$ solves $\seq{S}$
\item $\sigma$ solves $\typectxcolor{S_{r}}$ by Lemm~\ref{lem:merge-cons}.
\item $\sigma_e(CR_e),\ \forall i \in [1..n].\ \sigma_i(CR_i)$ are ground  by $IH$.\\
$(\ast)$ $\sigma(T_e.m:\seq{U}\rightarrow U',\emptyset )$ is ground because\\
     $\sigma(T_e.m :\seq{U}\rightarrow U') = (\sigma(T_e).m :\sigma(\seq{U})\rightarrow \sigma(U'))  = C_e.m : \seq{D} \rightarrow C$ and $C_e.m : \seq{D} \rightarrow C$ is ground. \\
     $CT\ \satisfy\ \sigma_e(CR_e)$, $\forall i \in [1..n].\ CT\ \satisfy\ \sigma_i(CR_i)$ by $IH$.\\
     $(\ast\ast)$ $CT \ satisfy\ \sigma(T'_e.m : \seq{U} \rightarrow U', \emptyset)$ because $mtype(m, C_e, CT)= \seq{D}\rightarrow C$ hence by rule \rulename{S-Method} holds that
  $CT \ satisfy\ (C_e.m : \seq{D} \rightarrow C, \emptyset)$, and $\sigma(T_e.m :\seq{U}\rightarrow U') = C_e.m : \seq{D} \rightarrow C$.\\
  As a result $\sigma$ solves $S_{cr}$ by Lemma~\ref{lem:mergeCR-cons}.
  \item Since $\{\seq{C}<: \seq{D}\}$ holds and $\sigma(\{\seq{T}<: \seq{U}\})= \{\seq{C}<: \seq{D}\}$, then $\sigma(\{\seq{T}<: \seq{U}\})$ holds 
      \end{itemize}
 $\sigma$ is ground solution because
  \begin{enumerate}[1)]
  	\item $\sigma(U')$ is ground because $\sigma(U') = C$
  	\item  $\sigma(R_e)$ is ground because 
  $\sigma(R_e) = (\{U' \smap C\}\circ \{U_i \smap D_i \}_{i\in [1..n]}\circ \sigma_e \circ \{\sigma_i\}_{i\in [1..n] })(R_e) = 
  (\sigma_e \circ \{\sigma_i\}_{i\in 1..n })(R_e)$ because $U'$, $\seq{U}$ 
  are defined fresh.\\
   $(\sigma_e \circ \{\sigma_i\}_{i\in [1..n]})(R_e) = ( \{\sigma_i\}_{i\in [1..n] }\circ \sigma_e )(R_e) $ by 
   Corollary~\ref{cor:assoc}.\\
  $( \{\sigma_i\}_{i\in [1..n] }\circ \sigma_e )(R_e) = ( \{\sigma_i\}_{i\in [1..n ]})(\sigma_e(R_e) ) = \sigma_e(R_e)$ because 
  $\sigma_e(R_e)$ is ground by $IH$. \\
  $\forall i \in [1..n]$. $\sigma(R_i)$ is ground because 
    $\sigma(R_i) = (\{U' \smap C\}\circ \{U_i \smap D_i \}_{i\in [1..n]}\circ \sigma_e \circ \{\sigma_i\}_{i\in [1..n] })(R_i) = 
  (\sigma_e \circ \{\sigma_i\}_{i\in [1..n] })(R_i)$ because $U'$, $\seq{U}$ 
  are defined fresh.\\
  $(\sigma_e \circ \{\sigma_i\}_{i\in [1..n] })(R_i) = (\sigma_e\circ \{\sigma_j\}_{j\in [1..i-1, i+1..n] }\circ \sigma_i )(R_i) $ by Corollary~\ref{cor:assoc}.\\ 
  $(\sigma_e\circ \{\sigma_j\}_{j\in [1..i-1, i+1..n ]})(\sigma_i(R_i))= \sigma_i(R_i)$ because 
  $\sigma_i(R_i)$ is ground by $IH$. 
  As a result $\sigma(R)$ is ground by definition of $merge_{R}$.
  	\item $\sigma(CR_e)$ is ground because 
  $\sigma(CR_e) = (\{U' \smap C\}\circ \{U_i \smap D_i \}_{i\in [1..n]}\circ \sigma_e \circ \{\sigma_i\}_{i\in [1..n] })(CR_e) = 
  (\sigma_e \circ \{\sigma_i\}_{i\in [1..n] })(CR_e)$ because $U'$, $\seq{U}$ 
  are defined fresh.\\
   $(\sigma_e \circ \{\sigma_i\}_{i\in [1..n] })(CR_e) = ( \{\sigma_i\}_{i\in [1..n] }\circ \sigma_e )(CR_e) $ by Corollary~\ref{cor:assoc}.\\ 
  $( \{\sigma_i\}_{i\in [1..n] }\circ \sigma_e )(CR_e) = ( \{\sigma_i\}_{i\in [1..n] })(\sigma_e(CR_e) ) = \sigma_e(CR_e)$ because 
  $\sigma_e(CR_e)$ is ground by $IH$.\\
   $\forall i \in [1..n]$. $\sigma(CR_i)$ is ground because 
    $\sigma(CR_i) = (\{U' \smap C\}\circ \{U_i \smap D_i \}_{i\in [1..n]}$ \\
    $ \circ \sigma_e \circ \{\sigma_i\}_{i\in [1..n] })(CR_i) = 
  (\sigma_e \circ \{\sigma_i\}_{i\in [1..n] })(CR_i)$ because $U'$, $\seq{U}$ are defined fresh.\\
   $(\sigma_e \circ \{\sigma_i\}_{i\in [1..n] })(CR_i) = (\sigma_e\circ \{\sigma_j\}_{j\in [1..i-1, i+1..n ]}\circ \sigma_i )(CR_i) $ by Corollary~\ref{cor:assoc}.\\ 
  $(\sigma_e\circ \{\sigma_j\}_{j\in [1..i-1, i+1..n] })(\sigma_i(CR_i))= \sigma_i(CR_i)$ because 
  $\sigma_i(CR_i)$ is ground by $IH$.
   $\sigma(T_e.m:\seq{U}\rightarrow U,\emptyset )$ is ground by $(\ast)$. As a result $\sigma(CR)$ is ground by definition of $merge_{CR}$,
   \end{enumerate}
  
  The correspondence relation holds because: 
  \begin{enumerate}[a)]
  	\item $C = \sigma(U)$
  	\item  $\Gamma \supseteq \sigma(R_e)$ because $\Gamma \supseteq \sigma_e(R_e)$ by $IH$, and from $2)$ $\sigma(R_e) = \sigma_e(R_e)$. 
  	$\forall i \in 1\ldots n$. $\Gamma \supseteq \sigma(R_i)$ because $\Gamma \supseteq \sigma_i(R_i)$ by $IH$, and from $2)$ $\sigma(R_i) = \sigma_i(R_i)$. As a result $\Gamma\supseteq \sigma(R)$ by definition of $merge_R$.
   \item  $CT \  satisfy\ \sigma(CR_e)$ because $CT \  satisfy\ \sigma_e(CR_e)$, and from $3)$ $\sigma(CR_e)  = \sigma_e(CR_e)$.
    $\forall i\in 1\ldots n $. $CT\ satisfy\ \sigma(CR_i)$ because $ CT \ satisfy\ \sigma_i(CR_i)$ by $IH$, and from $3)$ 
    $\sigma(CR_i) = \sigma_i(CR_i)$. $CT \ satisfy\ \sigma(T'_e.m : \seq{U} \rightarrow U', \emptyset)$ by $(\ast\ast)$.
  As a result $CT \ satisfy\ \sigma(CR_e) \cup \sigma(CR_1)\ldots \cup\sigma(CR_n) \cup \sigma(T_e.m :\seq{U}\rightarrow U')$, i.e., $CT\ \satisfy\ \sigma(CR)$ by Lemma~\ref{lem:satisfy}.
  \end{enumerate}
  
   \vskip2ex
   \item Case \rulename{T-New} with  \judge {\ctxplus{\Gamma}{CT}}{new\ C(\seq{e})} {C}. 
  
  By inversion,  \judge{\ctxplus{\Gamma}{CT}} {\seq{e}}{\seq{C}},  $fields(C, \ctxcolor{CT}) = C.init(\seq{D})$ and $\seq{C} <: \seq{D}$.
    
  By IH, \cojudge {\seq{e}} {\seq{T}} {\seq{S}} {\seq{R}} {\seq{CR}}, $\forall i \in 1 \ldots n$. $\mathit{solve}(\projExt{CT},S_i) = \sigma_i$, $\sigma_i(T_i)$, 
  $\sigma_i(R_i),\ \sigma_i(CR_i)$ are ground, and the 
  correspondence relation holds, i.e, $ C_i = \sigma_i(T_i) $, $\Gamma \supseteq  \sigma_i(R_i)$, $CT\ satisfy\  \sigma_i(CR_i)$.
    
  Let $\seq{U}$ be fresh, $merge_R(R_1, \ldots, R_n)= R|_{S_r}$, \\$CR|_{\typectxcolor{S_{cr}}} = merge_{CR}(CR_1, \ldots, CR_n, (C.init(\seq{U},\emptyset )))$. $S = \seq{S} \cup S_r \cup \typectxcolor{S_{cr}} \cup $\\ 
  $\{\seq{T} <: \seq{U}\}$ and $\sigma = \{U_i\smap D_i\}_{i\in [1..n]}\circ \{\sigma_i\}_{i\in [1..n]}$.
  
  Then \cojudge {C.init(\seq{e})} {C} {S} {R} {CR} holds by rule \rulename{TC-New}. \\
  $\sigma$ solves $S$ because it solves $\seq{S}$, $\typectxcolor{S_{r}}$, $S_{cr}$ , and $\{\seq{T}<: \seq{U}\}$ as shown below: 
  \begin{itemize}
  \item $\{solve(\projExt{CT}, S_i) = \sigma_i\}_{i \in [1..n]}$ and $\sigma = \{U_i \smap D_i \}_{i\in [1..n]}\circ \{\sigma_i\}_{i\in [1..n] }$
implies that $\sigma$ solves $\seq{S}$
\item $\sigma$ solves $\typectxcolor{S_{r}}$ by Lemma~\ref{lem:merge-cons}
\item $\forall i \in [1..n].\  \sigma_i(CR_i)$ are ground  by $IH$.\\
$(\ast)$ $\sigma(C.init(\seq{U}),\emptyset )$ is ground because
     $\sigma(C.init(\seq{U})) = (\sigma(C).init(\sigma(\seq{U})))  = C.init(\seq{D})$ and $C.init(\seq{D})$ is ground. \\
     $\forall i \in [1..n].\ CT\ \satisfy\ \sigma_i(CR_i)$ by $IH$.\\
     $(\ast\ast)$ $CT \ satisfy\ \sigma(C.init( \seq{U}), \emptyset)$ because $fileds( C,\ctxcolor{ CT})= \overline{C.f : D}$ hence by rule \rulename{S-Constructor} holds that
  $CT \ satisfy\ (C.init(\seq{D}), \emptyset)$, and \\ 
  $\sigma(C.init(\seq{U})) = C.init(\seq{D})$.\\
  As a result $\sigma$ solves $S_{cr}$ by Lemma~\ref{lem:mergeCR-cons}.
  \item Since $\{\seq{C}<: \seq{D}\}$ holds and $\sigma(\{\seq{T}<: \seq{U}\})= \{\seq{C}<: \seq{D}\}$, then $\sigma(\{\seq{T}<: \seq{U}\})$ holds 
  \end{itemize}  
   $\sigma$ is ground solution because:
  \begin{enumerate}[1)]
  	\item $\sigma(C)$ is ground because $C$ is ground.
  	\item  $\forall i \in [1..n]$. $\sigma(R_i)$ is ground because 
    $\sigma(R_i) = (\{U_i \smap D_i \}_{i\in [1..n]}\circ \{\sigma_i\}_{i\in [1..n] })$ \\
    $(R_i) = \{\sigma_i\}_{i\in [1..n] })(R_i)$ because $\seq{U}$ are defined fresh. \\
  $(\{\sigma_i\}_{i\in [1..n] })(R_i) = ( \{\sigma_j\}_{j\in [1..i-1, i+1..n ]}\circ \sigma_i )(R_i) $ by Corollary~\ref{cor:assoc}.\\
  $( \{\sigma_j\}_{j\in [1..i-1, i+1..n ]})(\sigma_i(R_i))= \sigma_i(R_i)$ because 
  $\sigma_i(R_i)$ is ground by $IH$. As a result $\sigma(R)$ is ground by definition of $merge_{R}$.
  	\item $\forall i \in [1..n]$. $\sigma(CR_i)$ is ground because 
    $\sigma(CR_i) = ( \{U_i \smap D_i \}_{i\in [1..n]}\circ \{\sigma_i\}_{i\in [1..n] })(CR_i) = 
  (\{\sigma_i\}_{i\in [1..n] })(CR_i)$ because $\seq{U}$ are defined fresh.\\
  $(\{\sigma_i\}_{i\in [1..n] })(CR_i) = (\{\sigma_j\}_{j\in [1..i-1, i+1..n ]}\circ \sigma_i )(CR_i) $ by Corollary~\ref{cor:assoc}.\\
  $( \{\sigma_j\}_{j\in [1..i-1, i+1..n] })(\sigma_i(CR_i))= \sigma_i(CR_i)$ because $\sigma_i(CR_i)$ is ground by $IH$.
   $\sigma(C.init(\seq{U}),\emptyset )$ is ground by $(\ast)$. As a result $\sigma(CR)$ is ground by definition of $merge_{CR}$.
    \end{enumerate}
  
  The correspondence relation holds because: 
  \begin{enumerate}[a)]
  	\item $C = \sigma(C)$
  	\item $\forall i \in 1\ldots n$. $\Gamma \supseteq \sigma(R_i)$ because $\Gamma \supseteq \sigma_i(R_i)$ by $IH$, and from $2)$ $\sigma(R_i) = \sigma_i(R_i)$. As a result $\Gamma\supseteq \sigma(R)$ by definition of $merge_R$.
   \item  $\forall i\in 1\ldots n $. $CT\ satisfy\ \sigma(CR_i)$ because $ CT \ satisfy\ \sigma_i(CR_i)$ by $IH$, and from $3)$ 
    $\sigma(CR_i) = \sigma_i(CR_i)$. 
    $CT\ satisfy\ \sigma(C.init(\seq{U}), \emptyset)$ by $(\ast\ast)$.\\
  As a result $CT \ satisfy\ \sigma(CR_1)\ldots\cup \sigma(CR_n) \cup \sigma(C.init(\seq{U}), \emptyset)$, i.e., $CT\ \satisfy\ \sigma(CR)$ by Lemma~\ref{lem:satisfy}.

  \end{enumerate}
     \vskip2ex
  \item Case \rulename{T-UCast}with \judge {\ctxplus{\Gamma}{CT}} {(C) e} {C}.
  
  By inversion, \judge{\ctxplus{\Gamma}{CT}} {e}{D} and $D <: C$.
   
   By IH, \cojudge {e} {T_e} {S_e} {R_e} {CR_e}, where $\mathit{solve}(\projExt{CT}, S_e) = \sigma_e$, $\sigma_e(T_e)$, 
   $\sigma_e(R_e)$, $\sigma_e(CR_e)$ are ground and the correspondence relation holds, i.e., $D = \sigma_e(T_e)$, $\Gamma \supseteq \sigma_e(R_e)$, $CT\ \satisfy\ \sigma_e(CR_e)$.
   
   Let $\sigma = \sigma_e$, and $S = S_e \cup \{T_e <: C\}$. 
 
  Then \cojudge {(C) e} {C} {S} {R_e} {CR_e} holds by rule \rulename{TC-UCast}. \\
  $\sigma$ solves $S$, because it solves $S_e$, and $\{T_e <: C\}$ as shown below:
  
  \begin{itemize}
  \item Since $\sigma = \sigma_e$ and $\sigma_e$ solves $S_e$ then $\sigma$ solves $S_e$.
  \item Since $\{D <: C\}$ holds and $\sigma(\{T_e <: C\}) = \{D <: C\}$ then $\sigma(\{T_e <: C\})$ holds.
  \end{itemize}
  $\sigma$ is ground solution because: 
   \begin{enumerate}[1)]
   	\item $\sigma(C)$ is ground because $C$ is ground as a given class in $CT$
   	\item $\sigma(R_e)$ is ground because $\sigma_e(R_e)$ is ground by $IH$ and $\sigma = \sigma_e$
   	\item $\sigma(CR_e)$ is ground because $\sigma_e(CR_e)$ is ground by $IH$ and $\sigma = \sigma_e$
   \end{enumerate} 
   The correspondence relation $(C, \Gamma, CT)\vartriangleright (C, R_e, CR_e, \sigma)$ holds because: 
     \begin{enumerate}[a)]
  	\item $C = \sigma(C)$
  	\item $\Gamma \supseteq \sigma(R_e)$, because $\Gamma \supseteq \sigma_e(R_e)$ by $IH$ and $\sigma = \sigma_e$
  	\item $CT\ satisfy\ \sigma(CR_e)$, because $CT\ satisfy\ \sigma_e(CR_e)$ by $IH$ and $\sigma = \sigma_e$
   \end{enumerate}
  \end{itemize}

  The proof is symmetric for \rulename{T-DCast}, and \rulename{T-SCast}, as in the case of \rulename{T-UCast}.

 \newpage
\begin{definition}[CReqs(CR)]
 $CReqs(CR) = \{ T.extends : T' \WHERE (T.extends:T', cond) \in CR \} \cup \{ T.init(\seq{T}) \WHERE (T.init(\seq{T}), cond) \in CR \}\cup \{ T.f: T' \WHERE (T.f :T', cond) \in CR \} \cup \{ T.m : \seq{T}\rightarrow T' \WHERE (T.m : \seq{T} \rightarrow T' , cond) \in CR \}$
\end{definition}

\begin{definition}[Domain of Class Table Clause]
	\begin{equation}
 domCl(CTcls) =  \begin{cases}
 	(C.extends) & \mbox{if } (CTcls)= (C.extends = D) \\
 	(C.f) & \mbox{if } (CTcls) = (C.f: C_f) \\
    (C.m) & \mbox{if } (CTcls) = (C.m : \seq{C}\rightarrow C_r ) \\
   (C.init) & \mbox{if } (CTcls) = (C.init(\seq{C})) 
 \end{cases}
\end{equation}
\end{definition}

\begin{definition}[Domain of CT]
$\domC = \{ domCl(CTcls) \WHERE CTcls \in CT \}$
\end{definition}



 \begin{definition}[translate a class requirements to class table entries] 
 It is given a ground class requirement clause $CReq$. 
 \begin{equation}
 translate(CReq) = \begin{cases}
 	(C.extends = D) & \mbox{if } (CReq)= (C.extends: D) \\
 	(C.f : C_f) & \mbox{if } (CReq) = (C.f: C_f) \\
    (C.m : \seq{C}\rightarrow C_r ) & \mbox{if } (CReq) = (C.m : \seq{C}\rightarrow C_r ) \\
   (C.init(\seq{C})) & \mbox{if } (CReq) = (C.init(\seq{C})) 
 \end{cases}
\end{equation}
\end{definition}

\begin{definition}[translate a class table entry to a class requirement CReq] 
 It is given a class table clause $CTcls$.\\
 \begin{equation}
 translate^*(CTcls) = \begin{cases}
 	(C.extends: D) & \mbox{if } (CTcls)= (C.extends = D) \\
 	(C.f : C_f) & \mbox{if } (CTcls) = (C.f: C_f) \\
    (C.m : \seq{C}\rightarrow C_r ) & \mbox{if } (CTcls) = (C.m : \seq{C}\rightarrow C_r ) \\
   (C.init(\seq{C})) & \mbox{if } (CTcls) = (C.init(\seq{C})) 
 \end{cases}
\end{equation}
\end{definition}

\begin{definition}[Clauses of supertypes of CReq] 
\begin{equation}
  \{(CReq, CR)\}_{\ll} =  \begin{cases}
        (T.extends : T') & \mbox{for } CReq = (T.extends : T')\\
        \{(T.init(\seq{T'})\} & \mbox{for } (T.init(\seq{T'}) \in CReqs(CR)  \\
		         & \;\;\; \wedge CReq = (T.init(\seq{T})) \wedge \seq T<: \seq T' \\
        \{(T'.f : T'_f)\} & \mbox{for } (T'.f : T'_f) \in CReqs(CR)  \\
		         & \;\;\; \wedge CReq = (T.f : T_f) \wedge T<: T' \\
	    \{(T'.m: \seq{T'}\rightarrow T'_r )\} & \mbox{for } (T'.m: \seq{T'}\rightarrow T'_r) \in CReqs(CR) \\
	             & \;\;\; \wedge CReq = (T.m :\seq{T}\rightarrow T_r)  \wedge T<: T'  
  \end{cases}
\end{equation}
\end{definition}

\begin{definition}[Clauses of subtypes of CReq] 
\begin{equation}
  \{(CReq, CR)\}_{\gg} =  \begin{cases}
          (T.extends : T') & \mbox{for } CReq = (T.extends : T')\\
          \{(T.init(\seq{T'})\} & \mbox{for } (T.init(\seq{T'}) \in CReqs(CR)  \\
		         & \;\;\; \wedge CReq = (T.init(\seq{T})) \wedge \seq T' <: \seq T \\
		\{(T'.f: T'_f)\} & \mbox{for } (T'.f: T'_f) \in CReqs(CR)  \\
		         & \;\;\; \wedge CReq = (T.f : T_f) \wedge T' <: T\\
	    \{(T'.m:\seq{T'}\rightarrow T'_r)\} & \mbox{for } (T'.m: \seq{T'}\rightarrow T'_r) \in CReqs(CR)  \\ 
	             & \;\;\; \wedge CReq = (T.m :\seq{T}\rightarrow T_r)\wedge T' <: T
  \end{cases}
\end{equation}
\end{definition}

\begin{definition}[Clauses of superclasses of CTcls] 
\begin{equation}
  \{(CTcls, CT)\}_{\ll^*} =  \begin{cases}
         (C.extends : D) & \mbox{for } CTcls = (C.extends = D)\\
         \{(C.init(\seq{C'})\} & \mbox{for } (C.init(\seq{C'}) \in CT  \\
		         & \;\;\; \wedge CTcls = (C.init(\seq{C})) \wedge \seq C <: \seq C' \\
		\{(D.f: D')\} & \mbox{for } (D.f : D') \in CT \\
		        & \;\;\; \wedge CTcls = (C.f : C') \wedge C <: D \\
	    \{(D.m: \seq{D}\rightarrow D_r )\} & \mbox{for } (D.m:\seq{D}\rightarrow D_r )\in CT \\
	            & \;\;\; \wedge CTcls = (C.m :\seq{C}\rightarrow C_r) \wedge C <: D
  \end{cases}
\end{equation}
\end{definition}

\begin{definition}[Compatible class requirements] 
Given two class requirements $CReq,\ CReq'$, compatibility of two class requirements $CReq \sim CReq'$ is defined over all cases of clauses:\\
 $
\begin{array}{lll}
	\bullet\ (T.extends: T_1) \sim (T'.extends : T_2) & if\ (T= T')\wedge (T_1 =T_2) \\
	\bullet\ T.init(\seq{T})\sim (T'.init(\seq{T'})) & if\ (T = T') \\
	\bullet\ (T.f : T_f) \sim (T'.f : T'_f) & if \ (T <: T')\vee (T:> T') \\
	\bullet\ (T.m : \seq{T}\rightarrow T_r)\sim (T'. m : \seq{T'}\rightarrow T'_r ) & if\ (T <: T')\vee (T:> T')
\end{array}
$
\end{definition}

\begin{definition}[Compatibility between a class requirement and a class table clause]
Given a class table clause $CTcls$, a class requirement $CReq$, and a ground solution $\sigma$, such that $\sigma(CReq)$ ground, compatibility
 $CReq \sim CReq'$ is defined over all cases of clauses:\\
 $
\begin{array}{lll}
    & \bullet\ \sigma(T.extends: T') \sim (C.extends = D) & if\ (\sigma(T)= C)\wedge (\sigma(T) = D) \\
    & \bullet\ \sigma(T.init(\seq{T}))\sim (C.init(\seq{C})) & if\ (\sigma(T ) = C) \\
    & \bullet\ \sigma(T.f : T' )\sim (C.f : C') & if\ (\sigma(T ) :> C) \vee (\sigma(T ) <: C)\\
    & \bullet\ \sigma(T.m : \seq{T}\rightarrow T' )\sim (C.m : \seq{C}\rightarrow C' ) & if\ (\sigma(T ) :> C) \vee (\sigma(T ) <: C)\\
\end{array}
$
\end{definition}

\begin{lemma}[Weakening for context]\label{lem:conxtWeak}
	If \judge{\Gamma}{t}{T}, and $x \notin dom(\Gamma )$, then \judge{\ctxplus{\Gamma}{x: C}}{t}{T}.
\end{lemma}
\begin{proof}
    Straightforward induction on typing derivations.
\end{proof}

\begin{lemma}[Weakening for a single class requirement]\label{lem:oneCTweak}
 Given $CT$, a class table clause $CTcls$, a class requirement $(CReq, cond)$ and $\sigma$, 
 such that $\sigma(CReq, cond)$ is ground, \\
 if $CT\ satisfy\ \sigma(CReq, cond)$ and  $\forall CTcls' \in \{(CTcls, CT )\}_{\ll^*}\ such\ that$\\
 $ CTcls' \notin CT$, then $CT\cup CTcls\ satisfy\ \sigma(CReq, cond)$. 

\end{lemma}
 \begin{proof}
We proceed by case analysis on the definition of $CReq$.
\begin{itemize}
	\item Case $CReq = (T.f : T')$. We consider $\sigma(T.f : T', cond) = (C.f : C', cond_g)$. 
	
	We have to show that $CT\cup CTcls\ \satisfy\ (C.f : C', cond_g)$. 
	
	It is given that $CT\ \satisfy\ (C.f : C', cond_g)$, therefore by inversion \\
	$field(f, C, CT)= C'$ (rule \rulename{S-Field}). 
	To show that the extended class table still satisfies the given class requirement, we distinguish the following cases on the definition of $CTcls$:
	\begin{enumerate}[1)]
		\item $CTcls = (D.g : D')$, such that $f \neq g$. 
		Moreover, consider the class table $CT\cup (D.g : D')$. 
		We know that since $f$ is not the same as $g$:
		
        \vspace{2mm}
		$(\ast)\; field(f, C, CT)= field(f, C, CT\cup (D.g : D')) = C'$.
        \vspace{2mm}
        
		As a result $CT\cup CTcls \ \satisfy\ (C.f : C', cond_g)$ by rule \rulename{S-Field} and $(\ast)$.
		 
	 	\item $CTcls = (A.f : A')$.
		Since by inversion $field(f, C, CT) = C'$, then there exists $D$, such that $C<:D$ and $(D.f: C')\in CT$. To proceed with the proof we distinguish two subcases: 
		\begin{enumerate}[a)]
			\item $A$ and $D$ belong to the same class hierarchy (subtyping relation).
	   			\subitem $A <: D$  \\
	   				This case does not hold by the assumption that $ \forall (CTcls') \in \{(A.f : A', CT )\}_{\ll^*}$ $such \ that$ $CTcls' \notin\ CT$, i.e., $D$ is a supertype of $A$, and $D. f : C'$ is an existing clause of the class table. 
	   			\subitem $A :> D$ \\
	   			 Since $C <: D$, then by transitivity we have $ C<:A$. Thus the type of $C.f$ does not depend on the type of $A.f$, because by field lookup rule, the type of $C.f$ is defined by the first supertype we find starting from left to right; since $C <: D <: A$, then $D.f$ is considered to define the type of $C.f$.
	    			Moreover, consider the class table $CT\cup (A.f : A')$. We know that since $A :> D$, $A :> C$, from field lookup definition:

       			 \vspace{2mm}
	   			 $(\ast)$ $field(f, C, CT)= field(f, C, CT\cup (A.f : A')) = C'$.
	   			 \vspace{2mm}

        		 As a result $CT\cup CTcls \ \satisfy\ (C.f : C', cond_g)$ by $(\ast)$ and rule \rulename{S-Field}.
       \item  $A$ and $D$ do not belong to the same class hierarchy (subtyping relation).
   		 We consider the class table $CT\cup (A.f : A')$. Since the field declaration for $f$ of class $A$ is unnecessary to define the type of $C.f$, because $C <:D$, and $D \nless : A $, $D \ngtr : A$, as a result $C \nless : A $, $C \ngtr : A$, then :
          
          \vspace{2mm}
		  $(\ast)$ $field(f, C, CT)= field(f, C, CT\cup (A.f: A')) = C'$.
		  \vspace{2mm}

          As a result $CT\cup CTcls \ \satisfy\ (C.f : C', cond_g)$ by $(\ast)$ and rule \rulename{S-Field}.
	    \end{enumerate}

     \item \textit{CTcls is different from a field clause}.\\
		We consider the class table $CT\cup translate(CReq')$. We know that since $CReq'$ is different from field clause for class requirements:

		 \vspace{2mm}
		$(\ast)$ $field(f, C, CT)= field(f, C, CT\cup CTcls) = C'$.
		 \vspace{2mm}
		
		As a result $CT\cup CTcls \ \satisfy\ (C.f : C', cond_g)$ by $(\ast)$ and rule \rulename{S-Field}.
	\end{enumerate}	
	\item $CReq = (T.m : \seq{U}\rightarrow U' )$	
 \end{itemize}
\end{proof}

\begin{lemma}[Class Table Weakening]\label{lem:CTweak}
Given $CT$, a class table clause $CTcls$, a set of class requirements $CR$, 
and a ground solution $\sigma$, such that $\sigma(CR)$ is ground, if $CT\ satisfy\ \sigma(CR)$ and 
 $\forall\ CTcls' \in \{(CTcls, CT )\}_{\ll^*}\ such\ that\ \ CTcls' \notin CT$, then  $CT\cup CTcls\ satisfy\ \sigma(CR)$. 
\end{lemma}
\begin{proof}
	We proceed by mathematical induction on the set of class requirements $CR$.\\
	\textbf{Initial step}: Show that the lemma holds for one single class requirement, i.e., $CR = \{(CReq, cond)\}$. It is given a class table clause $CTcls$, $\sigma(CReq, cond)$ is ground and $CT\ \satisfy\ \sigma(CReq, cond)$, then $CT\cup CTcls\ \satisfy\ \sigma(CReq, cons)$ by Lemma~\ref{lem:oneCTweak}.\\
	\textbf{Inductive step}: We suppose that the lemma is true for a set of class requirements $CR = CR'$, i.e., $CT\cup CTcls \ \satisfy\ \sigma(CR')$, where $\sigma(CR')$ is ground. \\
	We prove the lemma for 
	$CR = (CReq, cond) \cup CR'$, i.e., $CT\cup CTcls \ \satisfy\ \sigma(CR)$. \\
	 Union of class requirements is realized by $merge_{CR}$ function, i.e., \\
	 $CR|_{S} = merge(CR', (CReq, cond))$.
	 $\sigma(CReq, cond)$ in ground from the initial step and $\sigma(CR')$ is ground from the inductive step, then $sigma$ solve $S$ by Lemma~\ref{lem:mergeCR-cons}. 
	$CT\cup CTcls \ \satisfy\ \sigma(CReq, cond)$ from the initial step, and $CT\cup CTcls \ \satisfy \sigma(CR')$ from the inductive step, as a result $CT\cup CTcls \ \satisfy\ \sigma(Creq, cond)\cup \sigma(CR')$, i.e., $CT\cup CTcls \ \satisfy\ \sigma(CR)$ by Lemma~\ref{lem:satisfy}.
	\end{proof}
%
%
\begin{lemma}[Compatible clause in CT and not in CR]\label{lem:cCTnCR}
 Given $ CT', CR', (CReq \emptyset )$, $\sigma$, such that 
 $CR|_{S} = merge(CR', (CReq, \emptyset ))$, $\sigma$ solves $S$, and $\sigma(CR)$ is ground,
 if $CT'\ satisfy\ \sigma(CR')$, $\exists (CReq', cond)\in CR'.\ CReq \sim CReq'$, and $\exists CTcls \in CT'.\ \sigma(CReq) \sim CTcls $, then there exists a class table $CT$, such that $CT\ satisfy\ \sigma(CR)$.
   \end{lemma}
 \begin{proof} We proceed by case analyses on the definition of CReq.
   \begin{itemize}
   	\item $CReq = (T.f : U)$, and $(D.f : D') \in CT'$ for some $D$, by assumption. 
   	 We distinguish two cases regarding the subtyping relation between the CReq and the class table clause:
 \begin{enumerate}[1)]
 \item $D :> \sigma(T)$.
  Since $(D.f : D') \in CT$ is already a member of the class table, and $D$ is supertype of $\sigma(T)$, then $\sigma(U) =  D'$.
   We take $CT = CT'$. $field (f, \sigma(T), CT ) = D'$, therefore $CT \ \satisfy \ (\sigma(T).f : D',\emptyset )$ by rule \rulename{S-Field}, i.e., 
   $CT \ \satisfy \ (\sigma(T).f : D',\emptyset )$, and $CT\ \satisfy\ \sigma(CR')$, as a result $CT\ \satisfy\ \sigma(CR)$ by Lemma~\ref{lem:satisfy}.
  \item $D <: \sigma(T)$. We take $CT = CT'\cup translate(\sigma(T. f: U))$, then \\
   $CT\ \satisfy \ \sigma(T.f : U , \emptyset )$ by construction and 
  $CT\ \satisfy \sigma(CR')$ by Class Table Weakening Lemma~\ref{lem:CTweak}. As a result $CT\ \satisfy\ \sigma(CR)$ by Lemma~\ref{lem:satisfy}.
 \end{enumerate}
	\item $CReq = (T.m : \seq{U}\rightarrow U)$
      Analogous to the case of field clause.
   \end{itemize}
   \end{proof}

%

\begin{lemma} [Compatible clause in CT and in CR]\label{lem:cCTcCR} Given $ CT', CR', \\
(CReq, \emptyset )$, $\sigma$, such that
 	$CR|_{S} = merge(CR', (CReq, \emptyset))$, $\sigma$ solves $S$ and $\sigma(CR)$ is ground, 
 	if  $CT'\ satisfy\ \sigma(CR')$, $\exists (CReq',cond)\in CR'.\ CReq \sim CReq'$, $\exists CTcls \in CT'.\ \sigma(CReq) \sim CTcls$,  
   then there exists a class table $CT$, such that $CT\ satisfy\ \sigma(CR)$.  
\end{lemma}

%
 \begin{proof} We proceed by case analyses on the definition of CReq.
 \item $CReq = (T.f : U )$. \\
 By assumption $(T'.f : T_f, cond' ) \in CR'$ for some $T'$, and $(D.f : D')\in CT'$, for some $D$, $\sigma(T') <: D$, 
To show that $CT\ \satisfy\ \sigma(CR)$ we consider the case where $\sigma(cond')\Downarrow true$\footnote{We do not consider when it is false because the requirement is not valid requirement and it is a case as in Lemma~\ref{lem:cCTnCR} and the proof follows the same}.
$\sigma(cond) \Downarrow true$, i.e., all conditions in $cond$ do hold. $CT'\ satisfy \ \sigma(CR')$, and $(T'.f : T_f, cond')\in CR'$, therefore $CT'\ \satisfy\ \sigma(T'.f : T_f, cond')$, by inversion $field(f, T', CT') = D'$ (rule \rulename{S-Field}), where  $\sigma(T_f) = D'$. We distinguish to cases with respect to the subtyping relation between $D$ and $\sigma(T)$: 
	\begin{enumerate}[1)]
		\item $D > \sigma (T)$\\
		$D :> \sigma(T)$, $D :> \sigma(T')$, let us consider 
		$(\ast)$ $(\sigma(T).extends = D\in CT' $, $(\sigma(T').extends = D)\in CT'$ and $(D.f : D')\in CT'$ . The class requirements we are interested in are $(T.f : U , cond)$, $(T'.f : U', cond')$. After applying merging for the two requirements and remove for the two extend clauses the resulting valid requirements, that is the requirements where their conditions hold, are $(D.f : U, cond_t)$ and $(D.f : U', cond_{t'})$ (for sake of brevity we omit the detailed steps and the non interesting requirements for us). Then after applying remove for the field clause results that $\sigma(U) = \sigma(U') = D'$. $field (f, \sigma(T'), CT') = D' = \sigma(U')$, $field (f, \sigma(T), CT')   = D' =\sigma(U)$, therefore $CT'\ \satisfy\ \sigma(T.f : U,\emptyset )$ by rule \rulename{S-Field}. We take $CT = CT'$. $CT\ \satisfy\ \sigma(CR')$, and $CT \ \satisfy\ \sigma(T.f: U, \emptyset)$, as a result $CT\ \satisfy\ \sigma(CR)$ by Lemma~\ref{lem:satisfy}.
		\item $D <: \sigma(T)$\\
		By transitivity $\sigma(T') <: \sigma(T)$.  $D$ is subtype of $\sigma(T)$ and $D.f$ is unnecessary to determine the type of $\sigma(T).f$ by field lookup rule. We take $CT = CT'\cup translate(\sigma(T. f: U))$. $CT\ \satisfy\ \sigma(T.f U, \emptyset )$ by class table construction, and $\sigma(T') <: D <: \sigma(T)$ then $CT\ \satisfy\ \sigma(CR')$ by Class Table Weakening Lemma~\ref{lem:CTweak}.
		  \\ As a result $CT \ \satisfy\ \sigma(CR)$ by Lemma~\ref{lem:satisfy}.
	\end{enumerate} 
\item $CReq = (T. m :\seq{U}\rightarrow U )$	
Proof is analogous to case field clause.
\end{proof}


%
\begin{lemma}[Add Clause Definition in CT]\label{lem:AddCTcls}
	Given a class table clause $CTcls$ declaration, a class table $CT$ and a ground set of requirements $CR$, if $CTcls\notin CT$, and 
	$CT\ \satisfy \ CR$, then $CT\cup CTcls\ \satisfy\ CR$
\end{lemma}
\begin{proof}
Tedious but straightforward. 
\end{proof}

 \begin{theorem}[Equivalence of expressions: $\Leftarrow$]\label{theo:ExprL}
	Given $e,\ T,\ S,\ R ,\ CR,\ \Sigma,$ if $\cojudge{e}{T}{S}{R}{CR}$, $\func{solve}(\Sigma ,S) = \sigma$, and $\sigma$ is a ground solution, then there exists $C,\ \ctxcolor{ \Gamma,\ CT}$, such that \\
$\judge{\ctxplus{\Gamma}{CT}}{e}{C}$, $(C, \Gamma, CT)\vartriangleright\sigma(T, R, CR)$ and  $\projExt{CT}=\Sigma$.
 \end{theorem}

 We proceed by induction on the typing derivation.
\begin{itemize}
   \item Case \rulename{TC-Var} with \cojudge{x}{U}{\emptyset}{x:U}{\emptyset}
   
   Let $\sigma$ be a ground solution, such that $\sigma(U)$ is ground by assumption. 
   
   By inversion, $U$ is fresh, $S= \emptyset$, $R =\{x: U\}$, $CR= \emptyset$.
   
   By IH, $\Gamma_x = \{x : \sigma(U)\} $
   
   Let $\sigma(U) = C$, for some $C$ we know it is ground. 
   
   Then \judge{\ctxplus{\Gamma}{CT}}{x}{C} by rule $T-Var$, and the correspondence relation holds:
   \begin{enumerate}[a)]
   	\item $\sigma(U) = C$
   	\item We take $\Gamma = \Gamma_x$, and $\Gamma = \{x :C\} \supseteq \sigma(R) = \sigma(\{x : U\})$.
   	\item We take $CT = \emptyset$, since $CR = \emptyset$, and $\sigma(CR) = \emptyset$, then $CT\ \satisfy \ \sigma(CR)$
   \end{enumerate}
   
    \item Case \rulename{TC-Field} with \cojudge{e.f_i}{U}{S}{R_e}{CR}
    
    Let $S = S_e \cup S_f$, $\sigma$ be a ground solution, such that $solve(S,\Sigma )=\sigma$, i.e., it solves $S_e$, $S_f$, and $\sigma(U)$, $\sigma(R_e)$, $\sigma(CR)$ are ground by assumption. 

    By inversion, \cojudge {e}{T_e}{S_e}{R_e}{CR_e}, $\sigma(T_e)$, $\sigma(R_e)$, $\sigma(CR_e)$ are ground.
    $CR|_{\typectxcolor{S_f}} = merge_{CR}(CR_e, (T_e.f_i : U, \emptyset )) $, and $U$ is fresh.
    

    By IH, \judge{\ctxplus{\Gamma_e}{CT_e}} {e} {C_e}, the correspondence relation holds, i.e., $ C_e = \sigma(T_e)$, $\Gamma_e \supseteq \sigma(R_e)$, $CT_e\ satisfy\ \sigma(CR_e)$. $\projExt{CT_e}= \Sigma_e$ 
        
    Let $C_i = \sigma(U)$, for some $C_i$ we know is ground.\\
     We consider three cases to construct the class table $CT$:
     \begin{enumerate}[(1)]    	
    	\item  $\{(C_e.f : C_i, CT_e)\}_{\ll^*} = \emptyset$. Since no entry of class $C_e$ or its superclasses exist for field f in the given class table $CT_e$, we add a new entry in the class table, i.e., $CT = CT_e \cup (C_e.f : C_i)$.
    	\item $\{(T'_e.f : U, CR_e)\}_{\ll}\cup \{(T'_e.f : U, CR_e)\}_{\gg} = \emptyset$, $(D.f : D')\in CT_e$ for some $D, D'$, then by Lemma~\ref{lem:cCTnCR} $CT$ is constructed. 
    	\item $(T'.f: T_f, cond' )\in CR_e$, for some $T'$, $cond'$, $(D.f :D') \in CT_e$ for some $D$, $D'$, $\sigma(T') <:D$, then by Lemma~\ref{lem:cCTcCR} $CT$ is constructed.  
    \end{enumerate}
    From above we have that $field(f_i, C_e, CT) = C_i$, and no extends clauses are added to the class table $CT_e$, 
    therefore $\projExt{CT} = \Sigma_e = \Sigma$. 
       
      Then \judge{\ctxplus{\Gamma}{CT}}{e.f_i}{C_i} holds by rule \rulename{T-Field} , and the correspondence relation holds because: 
    \begin{enumerate}[a)]
        \item $\sigma(U) = C_i$ 
      	\item  We take $\Gamma = \Gamma_e$, and $\Gamma \supseteq \sigma(R_e)$ by $IH$.
    	\item  What it is left to be shown is that $CT\ satisfy\ \sigma(CR)$, we distinguish the following cases depending on the class table construction: 
    	\begin{enumerate}[(1)']
    		\item In addition to $(1)$, $\sigma(T'_e.f : U )= \sigma(T'_e).f : \sigma(U) = C_e.f : C_i$, therefore \\
    	 $CT\ satisfy\ \sigma(T'_e.f : U,\emptyset )$ by construction of $CT$.\\
    	  $CT_e\ \satisfy\ \sigma(CR_e)$ by $IH$, and $\{(C_e.f : C_i)\}_{\ll^*} \notin CT_e$ therefore $CT\ \satisfy\ \sigma(CR_e)$ by Class Table Weakening Lemma~\ref{lem:CTweak}. \\
    	As a result $CT\ satisfy\  \sigma(CR_e) \cup \sigma(T'_e.f : U, \emptyset )$, i.e., 
    	$CT\ satisfy\  \sigma(CR)$ by Lemma~\ref{lem:satisfy}. 
    	\item In addition to $(2)$, $CT_e\ satisfy\ \sigma(CR_e) $ by $IH$, then there is $CT$, $CT\ satisfy\ \sigma(CR)$ by Lemma~\ref{lem:cCTnCR}.
  	   \item In addition to $(3)$, $CT_e\ satisfy\ \sigma(CR_e) $ by $IH$, then there is $CT$, $CT\ satisfy\ \sigma(CR)$ by Lemma~\ref{lem:cCTcCR}.  
  	 \end{enumerate}
   \end{enumerate}
    
   \vskip2ex  
     \item Case \rulename{TC-Invk} with \cojudge{e.m(\seq{e})}{U}{S}{R}{CR}.
     
     Let $S = S_e \cup \seq{S} \cup S_r\cup S_s \cup \typectxcolor{S_{cr}} \cup \{\seq{T} <: \seq{U}\}$, and $\sigma$ be a ground solution, such that it solves $S$, i.e., $\sigma$ solves $S_e, \seq{S}$, $S_s$, $S_{cr}$, $\{\seq{T} <: \seq{U}\}$, and $\sigma(U')$, $\sigma(R)$, $\sigma(CR)$ are ground.

     By inversion, \cojudge{e}{T_e}{S_e}{R_e}{CR_e}, $\sigma(T_e)$, $\sigma(R_e)$, $\sigma(CR_e)$ are ground, 
      \cojudge{\seq{e}}{\seq{T}}{\seq{S}}{\seq{R}}{\seq{CR}}, $\forall i \in [1..n] $. $\sigma(T_i)$, $\sigma(R_i)$, $\sigma(CR_i)$ are ground,
     $ R|_{S_r} = merge_R(R_e, R_1, \ldots, R_n)$,
   $CR'|_{\typectxcolor{S_{s}}} = merge_{CR}(CR_e, CR_1, \ldots, CR_n)$, \\
   $CR|_{\typectxcolor{S_{cr}}} = merge_{CR}(CR', (T_e.m: \seq{U} \rightarrow U',\emptyset ))$, and $U'$, $\seq{U}$ are fresh.
                  
     By IH, \judge{\ctxplus{\Gamma_e}{CT_e}}{e}{C_e}, the correspondence relation holds, with $ C_e = \sigma(T_e)$, $\Gamma_e \supseteq \sigma(R_e)$, 
     $CT_e\ satisfy\ \sigma(CR_e)$. $\projExt{CT_e} = \Sigma_e$\\
     By IH, \judge{\ctxplus{\seq{\Gamma}}{\seq{CT}}}{\seq{e}}{\seq{C}}, the correspondence relation holds, $\forall i\in [1..n]$. $C_i = \sigma(T'_i)$, $\Gamma_i \supseteq  \sigma(R_i)$, $CT_i\ satisfy\  \sigma(CR_i)$. $\projExt{CT_s} = \Sigma_s$, where:  \begin{equation*}
     	\Gamma_s = \bigcup_{i\in [1..n]} \{\Gamma_i \} \qquad CT_s = \bigcup_{i\in [1..n]} \{CT_i\}
     \end{equation*}
    $(\ast)$ $\{freshU(CT_e)\cap freshU(CT_s)\} = \emptyset$, and \( \bigcap_{i\in [1..n]} \{freshU(CT_i) \}= \emptyset \), by Proposition~\ref{prop:indep}. 
   $\{CT_e \cup CT_s\}\ \satisfy\ \sigma(CR_e)$. \\
   $\forall i\in [1..n]$. $\{CT_e \cup CR_s \} \ \satisfy\ \sigma(CR_i)$ by Class Table Weakening Lemma~\ref{lem:CTweak}, 
   therefore $\{CT_e \cup CT_s\}\ satisfy\ \sigma(CR_e)\cup \sigma(CR_1)\ldots \cup \sigma(CR_n)$, i.e., \\
    $\{CT_e\cup CT_s\}\ satisfy\ \sigma(CR')$ by Lemma~\ref{lem:satisfy}.
     
  Let $C = \sigma(U')$, $\seq{D} = \sigma(\seq{U})$ for some $C$, $\seq{D}$ we know are ground. 
      $\seq{C}<: \seq{D}$ holds because $\sigma(\{\seq{T}<: \seq{U}\})$ holds. \\
      We consider three cases to construct the class table $CT$:
     \begin{enumerate}[(1)]    	
    	\item  $\{(C_e.m : \seq{D} \rightarrow C , \{CT_e \cup CT_s\})\}_{\ll^*} = \emptyset$. Since no entry of class $C_e$ or its superclasses exist for method $m$ in the given class table $\{CT_e \cup CT_s\}$, we add a new entry in the class table, i.e., 
    	$CT = \{CT_e \cup CT_s\}\cup (C_e.m :\seq{D}\rightarrow C)$.
    	\item $\{(T_e. m:\seq{U}\rightarrow U' , CR')\}_{\ll}\cup \{(T_e.m:\seq{U}\rightarrow U', CR')\}_{\gg} = \emptyset$, $(D.m :\seq{D}\rightarrow D')\in \{CT_e \cup CT_s\}$ for some $D, \seq{D}, D'$, then by Lemma~\ref{lem:cCTnCR} $CT$ is constructed. 
    	\item $(T'.m:\seq{T}\rightarrow T_r, cond') \in CR'$, for some $T', \seq{T}, T_r, cond'$, $(D.m:\seq{D}\rightarrow D') \in \{CT_e \cup CT_s\}$ for some $D, \seq{D}, D'$, $\sigma(T') <:D$, then by Lemma~\ref{lem:cCTcCR} $CT$ is constructed.  
    \end{enumerate}
    From above we have that $mtype(m, C_e, CT)= \seq{D}\rightarrow C$, and no extends clauses are added to the class table $\{CT_e\cup CT_s\}$, therefore $\projExt{CT} = \Sigma_e \cup \Sigma_s = \Sigma$.

%

      Then \judge{\ctxplus{\Gamma}{CT}}{e.m(\seq{e})}{C} holds by rule \rulename{T-Invk}, and the correspondence relation holds because:      
       \begin{enumerate}[a)]
      	\item  $C = \sigma(U)$
      	\item  We take $\Gamma = \Gamma_e \cup \Gamma_s$. $\Gamma \supseteq \sigma(R_e)$, because $\Gamma_e \supseteq \sigma(R_e)$ by $IH$ and Context Weakening Lemma~\ref{lem:conxtWeak}, $\Gamma \supseteq \sigma(R_1) \ldots \Gamma \supseteq \sigma(R_n)$, because $\Gamma_i \supseteq R_i$ by $IH$ and Context Weakening Lemma~\ref{lem:conxtWeak}, therefore $\Gamma \supseteq \sigma(R)$ by definition of $merge_R$.
      	\item What is left to be shown is that $CT\ satisfy\ \sigma(CR)$. We distinguish the following cases:
      	  \begin{enumerate}[(1)']
      		\item In addition to $(1)$, $\sigma(T_e.m :\seq{U}\rightarrow U')= \sigma(T_e).m :\sigma(\seq{U})\rightarrow \sigma(U') = C_e.m :\seq{D}\rightarrow C$ therefore 
    	 $CT\ satisfy\ \sigma(T_e.m :\seq{U}\rightarrow U',\emptyset )$ by construction of $CT$. 
    	 $\{CT_e\cup CT_s\}\ \satisfy\ \sigma(CR')$ by $(\ast)$ therefore $CT\ \satisfy\ \sigma(CR')$ by Class Table Weakening Lemma~\ref{lem:CTweak}. \\
    	As a result $CT\ satisfy\  \sigma(CR') \cup \sigma(T'_e.m :\seq{U}\rightarrow U, \emptyset )$, i.e., \\
    	$CT\ satisfy\  \sigma(CR)$ by Lemma~\ref{lem:satisfy}.
  			\item  In addition $(2)$, $\{CT_e\cup CT_s\}\ satisfy\ \sigma(CR') $ by $(\ast)$, then there is $CT$,\\ $CT\ satisfy\ \sigma(CR)$ by Lemma~\ref{lem:cCTnCR}.
  	   \item In addition to $(3)$ , $\{CT_e\cup CT_s \}\ satisfy\ \sigma(CR') $ by $(\ast)$, then there is $CT$, $CT\ satisfy\ \sigma(CR)$ by Lemma~\ref{lem:cCTcCR}.  
		 \end{enumerate}
      \end{enumerate}
             
    \vskip2ex
    \item Case \rulename{TC-New} with \cojudge{new\ C(\seq{e})}{C}{S}{R}{CR}
   
	Let $S  = \seq{S} \cup S_r \cup \typectxcolor{S_{cr}} \cup \{\seq{T} <: \seq{U}\}$, $\sigma$ be a ground solution, such that it solves $S$, i.e., $\sigma$ solves $\seq{S}, S_r$, $S_{cr}$, $\{\seq{T} <: \seq{U}\}$, and $\sigma(C)$, $\sigma(R)$, $\sigma(CR)$ are ground. 

    By inversion, \cojudge{\seq{e}}{\seq{T}}{\seq{S}}{\seq{R}}{\seq{CR}}, $\forall i \in [1..n] $. $\sigma(T_i)$, 
     $\sigma(R_i)$, $\sigma(CR_i)$ are ground, $ R|_{S_r} = merge_R(R_1, \ldots, R_n)$,
   $CR_s|_{\typectxcolor{S_{s}}} = merge_{CR}(CR_1, \ldots, CR_n) $, $CR|_{\typectxcolor{S_{cr}}} = merge_{CR}(CR_s, \\(C.init(\seq{U} ),\emptyset )) $, and $\{U_i\}_{i \in [1..n] }$ are fresh.\\
     By IH, \judge{\ctxplus{\seq{\Gamma}}{\seq{CT}}}{\seq{e}}{\seq{C}}, the correspondence relation holds, 
     $\forall i\in [1..n]$. $C_i = \sigma(T_i)$, $\Gamma_i\supseteq  \sigma(R_i)$, $CT_i\ satisfy\  \sigma(CR_i)$.  
     $\projExt{CT_s} = \Sigma_s$, where:
    \begin{equation*}
     	\Gamma_s = \bigcup_{i\in [1..n]} \{\Gamma_i \} \qquad CT_s = \bigcup_{i\in [1..n]} \{CT_i\}
     \end{equation*}
    $(\ast)$ $\bigcap_{i \in [1 .. n]}\{freshU(CT_i)\}= \emptyset$, by Proposition~\ref{prop:indep}. \\
    $\forall i\in 1\ldots n$. $CT_s\ \satisfy\ \sigma(CR_i)$ by Class Table Weakening Lemma~\ref{lem:CTweak}, therefore \\
    $CT_s\ \satisfy\ \sigma(CR_1)\ldots \cup \sigma(CR_n)$, i.e., $CT_s \ \satisfy\ \sigma(CR_s)$ by Lemma~\ref{lem:satisfy}.
  
      Let $\{U_i = D_i \}_{i \in [1..n]}$ for some $C$, $\seq{D}$ we know are ground. 
      $\seq{C}<: \seq{D}$ holds because $\sigma(\{\seq{T}<: \seq{U}\})$ holds.\\
       We consider three cases to construct the class table $CT$:
     \begin{enumerate}[(1)]    	
    	\item  $\{(C.init( \seq{D}) , CT_s)\}_{\ll^*} =\emptyset$. Since no entry of class $C$ exist for the constructor $init$ in the given class table $CT_s$, we add a new entry in the class table, i.e., $CT = CT_s\cup (C.init(\seq{D}))$.
    	\item $\{(C.init(\sigma(\seq{U})), \sigma(CR_s))\}_{\ll}\cup \{(C.init(\sigma(\seq{U})), \sigma(CR_s))\}_{\gg} = \emptyset$, \\ $(C.init(\seq{D'}))\in CT_s$, for some $\seq{D'}$, then by Lemma~\ref{lem:cCTnCR} $CT$ is constructed. 
    	\item $(C.init(\sigma(\seq{U'})), \sigma(cond')) \in \sigma(CR_s)$, for some $\seq{U'}, cond'$, $(C.init(\seq{D'})) \in  CT_s$, for some $\seq{D'}$, then by Lemma~\ref{lem:cCTcCR} $CT$ is constructed.  
    \end{enumerate}
    From above we have that $fields(C, CT) =C.init(\seq{D})$, and no extends clauses are added to the class table $CT_s$, 
    therefore $\projExt{CT} = \Sigma_s = \Sigma$. 
      
%
      Then \judge{\ctxplus{\Gamma}{CT}}{C.init(\seq{e})}{C} holds, the correspondence relation holds because:      
      \begin{enumerate}[a)]
      	\item  $C = \sigma(C)$
      	\item  We take $\Gamma = \Gamma_s$. $\Gamma_1 \supseteq \sigma(R_1) \ldots \Gamma_n \supseteq \sigma(R_n)$ by $IH$, then by Context Weakening Lemma~\ref{lem:conxtWeak} $\Gamma \supseteq \sigma(R)$ by definition of $merge_R$.
      	\item What is left to be shown is that $CT\ satisfy\ \sigma(CR)$. We distinguish the following cases:
      	  \begin{enumerate}[(1)']
      		\item In addition to $(1)$, $\sigma(C.init(\seq{U}))= \sigma(C).init(\sigma(\seq{U})) = C.init(\seq{D})$ therefore \\
    	 $CT\ satisfy\ \sigma(C.init(\seq{U}),\emptyset )$ by construction of $CT$. \\
    	  $CT_s\ \satisfy\ \sigma(CR_s)$ by $(\ast)$, therefore $CT\ \satisfy\ \sigma(CR_s)$ by Class Table Weakening Lemma~\ref{lem:CTweak}.
    	As a result $CT\ satisfy\  \sigma(CR_s) \cup \sigma(C.init(\seq{U}), $\\ 
    	$\emptyset )$, i.e., $CT\ satisfy\  \sigma(CR)$ by Lemma~\ref{lem:satisfy}.
  			\item  In addition to $(2)$, $CT_s\ satisfy\ \sigma(CR_s) $ by $(\ast)$, then there is $CT$, $CT\ satisfy\ \sigma(CR)$ by Lemma~\ref{lem:cCTnCR}.
  	   \item In addition to $(2)$, $\sigma(\seq{U'}) <: \seq{D'}$, and $ CT_s\ satisfy\ \sigma(CR_s) $ by $(\ast)$, then there is $CT$, $CT\ satisfy\ \sigma(CR)$ by Lemma~\ref{lem:cCTcCR}.  
		 \end{enumerate}
      \end{enumerate}
  
  \vskip2ex
  \item Case \rulename{TC-UCast} with \cojudge{(C) e}{C}{S}{R_e}{CR_e}
  
  Let $S= S_e \cup \{T'_e <: C\}$, $\sigma$ be a ground solution, such that it solves $S$, i.e., $\sigma$ solves $S_e$, $\{T'_e <: C\}$, and $\sigma(C)$, $\sigma(R_e)$, $\sigma(CR_e)$ are ground.

    By inversion, \cojudge {e}{T_e}{S_e}{R_e}{CR_e}, $\sigma(T_e)$, $\sigma(R_e)$, $\sigma(CR_e)$ are ground.
         

  By IH, \judge{\ctxplus{\Gamma_e}{CT_e}} {e} {D}, and the correspondence relation holds, i.e., $ C_e = \sigma(T_e)$, and $\Gamma_e \supseteq \sigma(R_e)$, $CT_e\ satisfy\ \sigma(CR_e)$. $\projExt{CT_e} = \Gamma_e$       
   
    $D<:C$ holds because $\sigma(\{T_e<: C\})$ holds.
    

	Then \judge{\ctxplus{\Gamma}{CT}}{(C) e}{C} holds by rule \rulename{T-UCast}, the correspondence relation holds because: 
	\begin{enumerate}[a)]
		\item  $C = \sigma(C)$ 
		\item $\Gamma = \Gamma_e$, $\Gamma \supseteq \sigma(R_e)$ by $IH$
		\item $CT = CT_e$, $CT\ satisfy\  \sigma(CR_e)$ by $IH$
	\end{enumerate}
	From above we have that no extends clauses are added to the class table $CT_e$, 
    therefore $\projExt{CT} = \Sigma_e = \Sigma$. 

\end{itemize}

The proof is symmetric for \rulename{T-DCast}, and \rulename{T-SCast}, as in the case of \rulename{T-UCast}.

\end{proof}

\newpage
 \begin{definition}[Correspondence relation for methods]
 Given judgments \judgeM{\ctxplus{C}{CT}}{C_0 \ m(\seq{C}\  \seq{x}) \{return\ e\}\ OK},   
  $\cojudgeOM{C_0 \ m(\seq{C}\  \seq{x}) \{return\ e\} \\ OK}{S}{T}{CR}$, and $\func{solve}(\Sigma, S)=\sigma$, where $\projExt{CT}=\Sigma$. The correspondence relation between CT and CR, written $(C, CT)\vartriangleright_m \sigma(T, CR)$, is defined as  
  \begin{enumerate}[a)]
  	\item $C = \sigma(T)$
  	\item $CT\ satisfy\ \sigma(CR)$
  \end{enumerate}
 \end{definition}
 \begin{theorem}[Equivalence of methods: $\Rightarrow$]
Given $m,\ C,\ CT, $ if $\judgeM{\ctxplus{C}{CT}}{C_0 \ m(\seq{C}\  \seq{x}) \{return\ e\}\\ OK}$, 
then there exists $ S,\ T,\ CR,\ \Sigma,\ \sigma$, where $\projExt{CT}=\Sigma$ and $\func{solve}(\Sigma, S)= \sigma$, such that
$\cojudgeOM {C_0\ m(\seq{C}\ \seq{x})\ \{ return\ e_0\} \ OK}  {S} {T}{CR}$ holds, $\sigma$ is a ground solution and \\
$(C, CT)\vartriangleright_m \sigma(T, CR)$ holds. 
\label{theo:MR}
 \end{theorem}

\begin{proof} By induction on the typing judgment.	

  Case \rulename{T-Method} with  \judgeM{\ctxplus{C}{CT}} {C_0 \ m(\seq{C}\  \seq{x}) \{return\ e\}\ OK}.
  
  By inversion, \judge{\ctxplus{\seq{x}: \seq{C}; this: C}{CT}} {e} {E_0},  $ \{E_0 <: C_0\}$, $extends(C, CT) = D$, i.e., $(C.extends = D) \in CT$ by rule \rulename{Extends}, and  
  $ if\ mtype(m, D, CT) = \seq{D}\rightarrow D_0, then\ \seq{C} = \seq{D};\ C _0= D_0$.
  
  By Theorem~\ref{theo:ExprR}
, \cojudge {e_0} {T_e} {S_e} {R_e} {CR_e}, where $\mathit{solve}(\projExt{CT}, S_e) = \sigma_e$, $\sigma_e(T_e)$, $\sigma_e(R_e)$, $\sigma_e(CR_e)$ are ground and the  relation holds, i.e., $E_0 = \sigma_e(T_e)$, 
  $\{\seq{x} : \seq{C} ; this: C\} \supseteq  \sigma_e(R_e)$, $CT\ satisfy\  \sigma_e(CR_e)$.
  
   We define the set of constraints $S'$ and the solution $\sigma'$ depending on the occurrence of $\seq{x}, this$ in $R_e$, and $U_c$ is fresh.
      
  \begin{itemize}
  \item If $\seq{x} \in dom(R_e)$ and $this \in dom(R_e)$, then $\{R_e(x_i) = U_i\}_{i\in [1..n]}$, $R_e(this) = U_c$, for $\seq{U}$ fresh. 
  We choose $S' =  \{C_i = R_e(x_i)\}_{i\in [1..n]}; \{C = R_e(this)\}$, $\sigma' = \{U_c\smap C\}\circ \{U_i \smap C_i\}_{i\in [1..n]}$.
     \item If $\seq{x} \in dom(R_e)$ and $this \not\in dom(R_e)$, then $\{R_e(x_i) = U_i\}_{i\in [1..n]}$, for $\seq{U}$ fresh. 
  We choose $S' = \{C_i = R_e(x_i)\}_{i\in [1..n]}$, $\sigma' = \{U_c\smap C\}\circ \{U_i \smap C_i\}_{i\in [1..n]}$.
   \item If $\seq{x} \not\in dom(R_e)$ and $this \in dom(R_e)$, then $R_e(this) = U_c$. 
  We choose $S' = \{C = R_e(this)\}$, $\sigma' = \{U_c\smap C\}$.
     \item If $\seq{x} \not\in dom(R_e)$ and $this \not\in dom(R_e)$.
  We choose $S' = \emptyset$, $\sigma' =\{U_c \smap C\}$.
  \end{itemize}
In all the cases above we have $\{U_c \smap C\}$, regardless the occurrence of $this$ in $R_e$, because $U_c$ serves as a placeholder for the current class where the method $m$ is declared as part of.
     
   Let $U_d$ be fresh, $R = R_e - this - \seq{x}$, $CR|_{S_{cr}} = merge_{CR}(CR_e, (U_c.extends: U_d, \emptyset), (U_d.m : \seq{C}\rightarrow C_0, \emptyset )_{opt})$, $S = S_e \cup \{T_e <: C_0\} \cup S_{cr} \cup  S'$, 
   \\$\sigma = \{U_d\smap D \} \circ \sigma' \circ \sigma_e $.
     

 We show why $R$ is $\emptyset$. The intuition behind it is that we know the actual types of the parameters since we have method declaration for $m$, and we know the actual type of $this$ since it is given the current class $C$ where method $m$ is declared as part of. $\Gamma \supseteq \sigma(R_e)$ by $IH$, i.e., all possible elements in $R_e$ are $\seq{X},\ this$ and 
	$\Gamma = \{\seq{x} : \seq{C} ; this: C\} - \seq{x} - this = \emptyset$, therefore $R= R_e -\seq{x} - this = \emptyset$. 

Then \cojudgeOM {C_0\ m(\seq{C}\ \seq{x})\ \{ return\ e_0\} \ OK}  {S} {U_c}{CR} holds by rule \rulename{T-Method}. \\
$\sigma$ solves $S$ because it solves $S_e,\ S'$, $S_{cr}$, and $\{T_e <: C_0\}$ as shown below:
\begin{itemize}
 \item $solve(\projExt{CT}, S_e) = \sigma_e$ and $\sigma = \{U_d\smap D \} \circ \sigma' \circ \sigma_e$
implies that $\sigma$ solves $S_e$
\item $\sigma$ solves $S'$ by Lemma~\ref{lem:merge-cons}.
\item $\sigma_e(CR_e)$ is ground  by Theorem~\ref{theo:ExprR}.\\
$(\ast)$ $\sigma(U_c.extends: U_d, \emptyset )$ is ground because $\sigma(U_c.extends : U_c)= (C.extends: D)$ and $C.extends : D$ is ground. \\
$(\ast\ast)$ 	$\sigma((U_d.m :\seq{C} \rightarrow C_0, \emptyset )_{opt})$ is ground because $\sigma(U_d.m: \seq{C} \rightarrow C_0 ) = 
  (\sigma(U_d).m: \sigma(\seq{C}) \rightarrow \sigma(C_0) ) = D.m: \seq{C}\rightarrow C_0$ and $D.m : \seq{C}\rightarrow C_0$ is ground. \\
$CT\ \satisfy\ \sigma_e(CR_e)$ by Theorem~\ref{theo:ExprR}.\\
$(\star)$ $CT\ satisfy\ \sigma(U_c.extends: U_d, \emptyset )$ because 
	$(C.extends = D) \in CT$ hence by rule \rulename{S-Extends} holds that 
	 $CT\ satisfy\ (C.extends: D,\emptyset )$, and $\sigma(U_c.extends : U_d)= (C.extends: D)$.\\
To show that $CT\ satisfy\ \sigma((U_d.m :\seq{C}\rightarrow C_0,\emptyset )_{opt})$ we distinguish the following cases:
\begin{itemize}
 \item[] $(\star')$ if $mtype(m, D, CT) = \seq{D}\rightarrow D_0$ is true then the optional class requirement $(U_d.m :\seq{C}\rightarrow C_0 )_{opt}$ is considered and $\seq{C}= \seq{D}, C_0 = D_0$, i.e., type of $m$ declared in $D$ is the same as the type of $m$ declared in $C$. 
$\sigma(U_d.m: \seq{D} \rightarrow D_0 ) = D.m: \seq{D}\rightarrow D_0$ and $mtype(m, D, CT) = \seq{D} \rightarrow D_0 = \seq{C}\rightarrow C_0$, therefore by rule \rulename{S-Method} holds that $CT \ satisfy\ \sigma(U_d.m : \seq{C}\rightarrow C_0,\emptyset)$. 
 \item[] $(\star'')$ if $mtype(m, D, CT) = \seq{D}\rightarrow D_0$  is false, then the optional class requirement $(U_d.m :\seq{C}\rightarrow C_0 )_{opt}$ is not considered. It is satisfiable by default since it is a not valid class requirement. 
   \end{itemize} 
 As a result $\sigma$ solves $S_{cr}$ by Lemma~\ref{lem:mergeCR-cons}.
  \item Since $\{E_0<: C_0\}$ holds and $\sigma(\{T_e<: C_0\})= \{E_0<: C_0\}$, then $\sigma(\{T_e<: C_0\})$ holds.
\end{itemize}

 $\sigma$ is ground solution because: 
\begin{enumerate}[1)]
	\item $\sigma(U)$ is ground because $\sigma(U) = C$ and $C$ is ground.
	\item $\sigma(CR_e)$ is ground because $\sigma(CR_e)= (\{U_d\smap D\}\circ \sigma' \circ\sigma_e )(CR_e) =(\{U_d\smap D\}\circ \sigma' \})(\sigma_e(CR_e)) \\ = \sigma_e(CR_e)$ because $\sigma_e(CR_e)$ is ground by Theorem~\ref{theo:ExprR}. \\
	$\sigma(U_c.extends: U_d, \emptyset )$ is ground by $(\ast)$. 	$\sigma((U_d.m :\seq{C} \rightarrow C_0, \emptyset )_{opt})$ is ground by $(\ast\ast)$. As a result $\sigma(CR)$ is ground by definition of $merge_{CR}$.
\end{enumerate}

The correspondence relation holds because:
\begin{enumerate}[a)]
	\item $C = \sigma(U_c)$
	\item $CT\ \satisfy\ \sigma(CR_e)$ because $CT\ \satisfy\ \sigma_e(CR_e)$ by Theorem~\ref{theo:ExprR} and from $3)$ $\sigma(CR_e) = \sigma_e(CR_e)$. 
	$CT\ satisfy\ \sigma(U_c.extends: U_d, \emptyset )$ by $(\star)$. To show that $CT\ satisfy\ \sigma(CR)$, is left to scrutinize $CT\ satisfy\ \sigma((U_d.m :\seq{C}\rightarrow C_0,\emptyset )_{opt})$. 
	 We distinguish the following cases:
	 \begin{itemize}
	 	\item if $mtype(m, D, CT) = \seq{D}\rightarrow D_0$ is true then $CT \ satisfy\ \sigma(U_d.m : \seq{C}\rightarrow C_0,\emptyset)$
	 	by $(\star')$. As a result $CT \satisfy\ \sigma(CR_e)\cup \sigma(U_c.extends: U_d, \emptyset ) \cup \sigma(U_d.m:\seq{C}\rightarrow C_0, \emptyset )$, i.e., $CT\ satisfy\ \sigma(CR)$ by Lemma~\ref{lem:satisfy}.
        \item if $mtype(m, D, CT) = \seq{D}\rightarrow D_0$  is false, then is not considered from $(\star'')$. As a result $CT \satisfy\ \sigma(CR_e)\cup \sigma(U_c.extends: U_d, \emptyset )$, i.e., \\
        $CT\ satisfy\ \sigma(CR)$ by Lemma~\ref{lem:satisfy}.
 	 \end{itemize}
	 
%
\end{enumerate}
\end{proof}

  \begin{theorem}[Equivalence of methods: $\Leftarrow$]
Given $m,\ T,\ S,\ CR,\ \Sigma, $ if $\cojudgeOM {C_0\ m(\seq{C}\ \seq{x})\ \{ return\ e_0\} \\ OK}  {S} {T}{CR}$, $\func{solve}(\Sigma, S) = \sigma$, and $\sigma$ is a ground solution, then there exists $ C,\ CT$, such that
$\judgeM{\ctxplus{C}{CT}}{C_0 \ m(\seq{C}\  \seq{x}) \{return\ e\}\ OK}$ holds, $(C, CT)\vartriangleright_m \sigma(T, CR)$ and $\projExt{CT}=\Sigma$. 
\label{theo:ML}
 \end{theorem}
 
 \begin{proof} By induction of the typing judgment, and case analysis of the class table construction. \\
Case \rulename{TC-Method} with \cojudgeOM{C_0 \ m(\seq{C}\  \seq{x}) \{return\ e\}\ OK} {S}{U_c}{CR}.
  
  Let  $S =S_e \cup \{T_e<: C_0\}  \cup S_c \cup S_{cr} \cup S_x $, $\sigma$ be a ground solution, such that $solve(\Sigma, S) =\sigma$, i.e., $\sigma$ solves $S_e, S_x, S_{cr}, S_c, \{T_e<: C_0\}$, and $\sigma(U_c)$, $\sigma(CR)$ are ground.
  
  By inversion, \cojudge {e} {T_e} {S_e} {R_e} {CR_e}, $\sigma(T_e)$, $\sigma(R_e)$, $\sigma(CR_e)$ are ground. 
   $CR'|_{S'} = merge_{CR}(CR_e, (U_c.extends: U_d , \emptyset ))$
  $CR|_{S_{cr}} = merge_{CR}(CR', (U_d.m :\seq{C}\rightarrow C_0,\emptyset )_{opt})$, where $U_c$, $U_d$ are fresh, 
  $ R_e - this - \seq{x} = \emptyset$. $\sigma$ solves $S'$ by Lemma~\ref{lem:merge-cons}. 
    
  By Theorem~\ref{theo:ExprL}
, \judge{\ctxplus{\Gamma_e}{CT_e}} {e} {E_0}, the correspondence relation holds, i.e., $E_0 = \sigma(T_e)$, 
  $\Gamma_e \supseteq  \sigma(R_e)$, $CT_e \ satisfy\ \sigma(CR_e)$. $\Gamma_e = \{\seq{x} : \seq{C} ; this: C\}$, because $ R_e - this - \seq{x} = \emptyset$ and $\Gamma_e\supseteq \sigma(CR_e)$. $\projExt{CT_e} =\Sigma_e$
  
   Let $C = \sigma(U_c)$, $D = \sigma(U_d)$ for some $C,\ D$ we know are ground. $E_0 <: C_0$ holds because $\sigma(T_e <: C_0)$ holds.\\
    $Context$ empty because $\Gamma_e - \{\seq{x}: \seq{C}; this: C\} =\emptyset$.\\
   
   We proceed by construction of the class table in steps. \\   
 First we consider three cases to construct the class table CT' with respect to the requirement for $extends$:
 \begin{enumerate}[(1)]
 	\item\textbf{clause not in class table}. $(C.extend = D) \notin CT_e$, then \\
 	$CT' = CT_e;\ (C.extends = D)$.
 	\item \textbf{clause in class table, but not in class requirements}. $(C.extends = D) \in CT_e$, and $(U_c.extends : U_d, \emptyset) \notin CR_e$, then $CT' = CT_e$. 
 	\item \textbf{clause in class table and class requirements}. $(C.extends = D) \in CT_e$, and $(U_c.extends : U_d, \emptyset) \in CR_e$ is not a valid case, because $U_c$ is defined fresh and in $CR_e$ we do not have existing requirements regarding $U_c$ for extend, because in method body we can have recursive method call or field access and not in extends, i.e., $this$ can invoke the method itself or other methods and fields but not extends.
\end{enumerate}
%
From above and by rule \rulename{Extends} we have that $extends(C, CT') = D$, an extends is added to the class table $CT_e$, therefore $\projExt{CT'} = \Sigma_e \cup (C,D) = \Sigma'$\\\\
Second we consider three cases to construct the class table CT' with respect to the requirement for method $m$:
 \begin{enumerate}[(1)]
   \setcounter{enumi}{3}
  	\item \textbf{clauses of superclasses not in class table}. $\{(D.m : \seq{D}\rightarrow D_0 , CT')\}_{\ll^*} = \emptyset$, then $CT = CT' $
  	\item \textbf{compatible clauses in class table, but not in class requirements}. 
  	$\{(D. m:\seq{C}\rightarrow C' , CR_e)\}_{\ll}\cup \{(D.m:\seq{D}\rightarrow D', CR_e)\}_{\gg} = \emptyset$, 
  	$(D'.m :\seq{D'}\rightarrow D_r)\in CT'$ for some $D', \seq{D'}, D_r$, then by Lemma~\ref{lem:cCTnCR} $CT'$ is constructed. 
  	\item \textbf{compatible clauses in class table, and in class requirements}. $(D'.m : \seq{D'}\rightarrow D_r) \in CT'$ for some $D'$, $\seq{D'}, D_r$, and $(U_d.m \seq{C}\rightarrow C_0, \emptyset) \in CR_e$ is not a valid case because $U_d$ is defined fresh and $U_d \neq R_e(this)$, i.e., it is possible to have in the body of $m$ $this.m$, but it is impossible to have recursive call of $m$ invoked by $U_d$, as it is defined fresh and different type than $this$.
 \end{enumerate}
 From above we have that $if\ mtype(m, D, CT) = \seq{D}\rightarrow D_0\ then\ \seq{C}= \seq{D}; C_0 = D_0$, no extends clauses are added to the class table $CT'$, therefore $\projExt{CT} = \Sigma' = \Sigma$\\ \\
%
  Then \judgeM{\ctxplus{C}{CT}} {C_0 \ m(\seq{C}\  \seq{x}) \{return\ e\}\ OK} holds by rule \rulename{T-Method}, the correspondence relation holds because: 
  \begin{enumerate}[a)]
  	\item $C = \sigma(U_c)$
  	\item What is left to be shown is that $CT\ satisfy \sigma(CR)$, first we start by showing \\
  	$CT'\ satisfy\ \sigma(CR')$ and we distinguish the following cases: 
  	\begin{enumerate}[(1)']
  		\item In addition to $(1)$ $\sigma(U_c.extends: U_d)= \sigma(U_c).extends : \sigma(U_d) = C.extends : D$ therefore
    	 $CT'\ satisfy\ \sigma(U_c.extends : U_d,\emptyset )$ by construction. $CT_e\ \satisfy\ \sigma(CR_e)$ by Theorem~\ref{theo:ExprL}, and $\sigma(U_c.extends : U_d)  \notin CT_e$, therefore 
    	 $CT'\ \satisfy\ \sigma(CR_e)$ by Class Table Weakening Lemma~\ref{lem:CTweak}.\\
    	As a result $CT'\ satisfy\  \sigma(CR_e) \cup \sigma(U_c.extends : U_d, \emptyset )$, i.e., \\
    	$CT'\ satisfy\  \sigma(CR)$ by Lemma~\ref{lem:satisfy}.
  	\item  In addition to $(2)$, $CT' \ \satisfy\ (C.extends: D, \emptyset )$ by rule \rulename{S-Extend}, and $(C.extends : D) = \sigma(U_c.extends: U_d)$, therefore $CT' \ \satisfy\ (U_c.extends: U_d, \emptyset)$. \\
  	$CT'\ \satisfy\ \sigma(CR_e)$ by Theorem~\ref{theo:ExprL}. 
  	As a result $CT'\ satisfy\  \sigma(CR_e) \cup \sigma(U_c.extends : U_d, \emptyset )$, i.e., 
    	$CT'\ satisfy\  \sigma(CR)$ by Lemma~\ref{lem:satisfy}.\\
    	
 \hspace{-2.4em} Second we show that $CT\ satisfy\ \sigma(CR)$, and we distinguish the following cases:
   	  	\item In addition to $(4)$, the class requirement $(U_d.m : \seq{C}\rightarrow C_0 )_{opt}$ is not considered since it is an optional requirement, therefore $CR = CR'$, $CT'\ satisfy\ \sigma(CR')$. As a result $CT\ \satisfy\ \sigma(CR)$.
  	 \item  In addition to $(5)$, $CT\ satisfy\ \sigma(CR)$ by Lemma \ref{lem:cCTnCR}.
  	\end{enumerate} 
  	\vskip 1em
  	Method declaration consist in adding method clause $m$ in $CT$, whether it is already member of the $CT$ or not. Also, adding the method $m$ in $CT$ does not affect the satisfaction of the class requirements. We are interested that the clause $m$ with its actual ty
  	pe is part of class table. Namely resulting class table $CT_r$, such that $(C.m: \seq{C}\rightarrow C_0) \in CT_r$, $CT_r \satisfy\ \sigma(CR)$. \\
    Lastly we show that $CT_r\ satisfy\ \sigma(CR)$ and we distinguish the following cases: 
  	\begin{description}		 
  	   \item[$\bullet$] $(C.m : \seq{C}\rightarrow C_0) \notin CT$ then we add declaration in the class table, i.e., $CT_r = CT\cup\ (C.m:\seq{C}\rightarrow C_0)$ and $CT_r\ satisfy\ \sigma(CR)$ by Lemma \ref{lem:AddCTcls}.
  		\item[$\bullet$] $C.m \in dom(CT)$ then $CT_r = CT$. Hence, $CT_r\ \satisfy\ \sigma(CR)$.
  	\end{description} 
  \end{enumerate}
 \end{proof}

\newpage
\begin{definition}[Correspondence relation for classes]
 Given \judgeM{CT}{class\ C\ extends\ D\ \{\seq{C}\ \seq{f}; K\ \seq{M}\}\ OK} and 
  $\cojudgeOC {class\ C\ extends\ D\ \{\seq{C}\ \seq{f};\ K\ \seq{M}\}\ OK} {S} {CR}$, and $\func{solve}(\Sigma,S)=\sigma$, where $\projExt{CT}=\Sigma$. The correspondence relation between CT and CR, written $(CT)\vartriangleright_c \sigma(CR)$, is defined as:
  \begin{enumerate}[a)]
  \item $CT\ satisfy\ \sigma(CR)$
 \end{enumerate}
 \end{definition}
 \begin{theorem}[Equivalence of classes: $\Rightarrow$]
Given $C,\ CT, $ if \judgeM{CT}{class\ C\ extends\ D\ \{\seq{C}\ \seq{f}; K\\\seq{M}\}\ OK},
then there exists $ S,\ CR,\ \Sigma,\ \sigma$, where $\projExt{CT}=\Sigma$ and $\func{solve}(\Sigma, S)=\sigma$, such that
\cojudgeOC{class\ C\ extends\ D \{\seq{C}\ \seq{f};\ K\ \seq{M}\}\ OK} {S} {CR} holds, $\sigma$ is a ground solution and 
$(CT)\vartriangleright_c \sigma(CR)$ holds. 
\label{theo:CR}
 \end{theorem}
\begin{proof} By induction on the typing judgment. \\
   Case \rulename{T-Class} with $\judgeM{CT}{class\ C\ extends\ D\ \{\seq{C}\ \seq{f}; K\ \seq{M}\}\ OK}$.

By inversion, $K=C(\seq{D'}\ \seq{g}, \seq{C'}\ \seq{f})\{super(\seq{g}); this.\seq{f} = \seq{f}\}$, i.e., the constructor initializes all fields of  $fields(D, \ctxcolor{CT})=D.init(\seq{D'})$, and \judgeM{\ctxplus{C}{CT}}{\seq{M}\ OK}. 

By Theorem~\ref{theo:MR}
, \cojudgeOM {\seq{M} \ OK}  {\seq{S}} {\seq{U}} {\seq{CR}}, $\forall i \in 1\ldots n$. $\mathit{solve}(\projExt{CT}, S_i) = \sigma'_i$, $\sigma_i = \{U_i \smap C\} \circ \sigma'_i $, $\sigma_i(U_i)$, $\sigma_i(CR_i)$ are ground and the correspondence relation holds, i.e., $C = \sigma_i(U_i)$, $CT\ satisfy\ \sigma_i(CR_i)$. 

Let $CR|_{\typectxcolor{S_{cr}}} = merge_{CR}(CR_1, \ldots, CR_n, D.init(\seq{D'})) $, $S= \seq{S} \cup\typectxcolor{S_{cr}}\cup \{U_i = C\}_{i\in [1..n]} \cup \{C_i = D'_i\}_{i\in 1..k} \cup \{ C_i = C'_i\}_{i\in k..n}$,  where $k = | \seq{D'} |$, $n = | \seq{C}|$, $n - k = |\seq{C'}|$, and $\sigma = \{\sigma_i\}_{i\in [1..n]}$.\\
%
%
Then \cojudgeOC {class\ C\ extends\ D \{\seq{C}\ \seq{f};\ K\ \seq{M}\}\ OK} {S} {CR} holds by rule \rulename{TC-Class}.\\
 $\sigma$ solves $\seq S$, $\typecolor{S_{cr}}$, $\{U_i = C\}_{i\in [1..n] }$ and $\{C_i = D'_i\}_{i\in 1..k} \cup \{ C_i = C'_i\}_{i\in k..n}$  as shown below:
 \begin{itemize}
 	\item $\sigma$ solves $\seq S$ because $\sigma =\{\{U_i \smap C\} \circ \sigma'_i\}_{i \in [1..n]}$, and $\forall i \in [1..n].\ $ \\$solve(\projExt{CT}, S_i) = \sigma_i$.
 	\item $\forall i \in [1..n].\ \sigma_i(CR_i)$ are ground by Theorem~\ref{theo:MR}.\\
 	$(\ast)$ $\sigma(D.init(\seq{D'}))$ is ground because $(D.init(\seq{D'}))$ is ground.\\
 	$\forall i \in [1..n].\ CT\ \satisfy\ \sigma_i(CR_i)$ by Theorem~\ref{theo:MR}.\\
 	$(\ast \ast)$ $CT\ \satisfy\ \sigma(D.init(\seq{D}), \emptyset )$ because $fields(D, CT) = D.init(\seq{D'})$ hence by rule \rulename{S-Constructor} holds that  $CT\ \satisfy\ \sigma(D.init(\seq{D'}),\emptyset)$.
 	As a result $\sigma$ solves $\typecolor{S_{cr}}$ by Lemma~\ref{lem:mergeCR-cons}.
 	\item $\sigma$ solves $\{U_i = C\}_{i\in [1..n] }$ because $\sigma = \{\{U_i \smap C\} \circ \sigma'_i\}_{i \in [1..n]}$.
 	\item $\{C_i = D'_i\}_{i\in 1..k} \cup \{ C_i = C'_i\}_{i\in k..n}$ holds because $K$ initializes all fields of class $C$ as it is given by inversion.
 \end{itemize}
$\sigma$ is ground solution because: 
\begin{enumerate}[1)]
	\item $\forall i\in 1 \ldots n $. $\sigma(CR_i)$ is ground because 
	$\sigma(CR_i) = (\{\sigma_j\}_{j\in [1..i-1, i+1..n] }\circ \sigma_i )(CR_i)$ by Corollary~\ref{cor:assoc}.\\
	$(\{\sigma_j\}_{j\in [1..i-1, i+1..n] }) (\sigma_i(CR_i)) = \sigma_i(CR_i)$ because $\sigma_i(CR_i)$ is ground by Theorem~\ref{theo:MR}. \\
  $\sigma(D.init(\seq{D'}))$ is ground by $(\ast)$. As a result $\sigma(CR)$ is ground by definition of $merge_{CR}$
\end{enumerate}
The correspondence relation holds because:
\begin{enumerate}[a)]
	\item $\forall i\in 1\ldots n$. $CT\ satisfy\ \sigma(CR_i)$ because $CT\ \satisfy\ \sigma_i(CR_i)$ by Theorem~\ref{theo:MR}, and from $1)$ $\sigma(CR_i) = \sigma_i(CR_i)$. $CT\ \satisfy\ \sigma(D.init(\seq{D}), \emptyset )$ by $(\ast\ast)$. As a result $CT\ \satisfy\ \sigma(CR_1)\ldots \cup \sigma(CR_n) \cup\sigma(D.init(\seq{D'}))$, i.e., $CT\ \satisfy\ \sigma(CR)$ by Lemma~\ref{lem:satisfy}. 
\end{enumerate}
\end{proof}
 
\begin{theorem}[Equivalence of classes: $\Leftarrow$]
Given $C,\ CR,\ \Sigma,$ if \cojudgeOC {class\ C\ extends\ D \{\seq{C}\ \seq{f};\ K\ \seq{M}\}\\ OK} {S} {CR}, 
$\func{solve}(\Sigma,S)=\sigma$, and 
$\sigma$ is a ground solution, then there exists $CT$, such that \\
 \judgeM{CT}{class\ C\ extends\ D\ \{\seq{C}\ \seq{f}; K\ \seq{M}\}\ OK} holds, $(CT)\vartriangleright_c \sigma(CR)$ holds and $\projExt{CT}=\Sigma$. 
 \label{theo:CL}
 \end{theorem} 

\begin{proof} By induction on the typing judgment.\\
   Case \rulename{TC-Class} with  \cojudgeOC {class\ C\ extends\ D \{\seq{C}\ \seq{f};\ K\ \seq{M}\}\ OK} {S} {CR}.

Let  $S= \seq{S} \cup\typectxcolor{S_{cr}}\cup \{U_i = C\}_{i\in [1..n]} \cup \{C_i = D'_i\}_{i\in 1..k} \cup \{ C_i = C'_i\}_{i\in k..n}$,  where $k = | \seq{D'} |$, $n = | \seq{C}|$, $n - k = |\seq{C'}|$, $\sigma$ be a ground solution, such that it solves $S$ and $\sigma(CR)$ is ground.

By inversion, \cojudgeOM {\seq{M} \ OK}  {\seq{S}} {\seq{U}} {\seq{CR}}, $\forall i\in 1\ldots n$. $\sigma(U_i)$, $\sigma(CR_i)$ are ground.\\
  $merge_{CR}(CR_1, \ldots, CR_n) = CR'|_{\typectxcolor{S_{c}}}$, 
 $merge_{CR}((D.init(\seq{D'})), CR') = CR|_{\typectxcolor{S_{cr}}}$.

%
Let $\forall i\in i\ldots n$. $\sigma(U_i) = C$ for $C$ we know it is ground. 

By Theorem~\ref{theo:ML}
 \judgeM{\ctxplus{C}{\seq{CT}}}{\seq{M}\ OK}, the correspondence relation holds, $\forall i\in 1\ldots n$. $C = \sigma(U_i)$, $CT_i \ \satisfy\ \sigma(CR_i)$. $\projExt{CT'} = \Sigma'$, where $CT' = \bigcup_{i\in [1..n]}\{CT_i\}$. \\
$(\ast)$ $\bigcap_{i\in [1..n]}\{freshU(CR_i)\}=\emptyset$ by Proposition~\ref{prop:indep}. $\forall i\in 1\ldots n$. $CT'\ \satisfy\ \sigma(CR_i)$ by Class Table Weakening Lemma, therefore $CT' \ \satisfy\ \sigma(CR')$ by Lemma~\ref{lem:satisfy}.  

The constructor $K$ initializes all fields of class $C$, i.e.,  $K=C(\seq{D'}\ \seq{g}, \seq{C'}\ \seq{f})\{ $\\
$super(\seq{g}); this.\seq{f} = \seq{f}\}$, because $\sigma$ solves $\{C_i = D'_i\}_{i\in 1..k} \cup \{ C_i = C'_i\}_{i\in k..n}$.

We consider three cases to construct the class table $CT$:
\begin{enumerate}[(1)]
\item  $\{(D.init( \seq{D'}) , CT')\}_{\ll^*} =\emptyset$. Since no entry of class $D$ exist for the constructor $init$ in the given class table $CT'$, we add a new entry in the class table, i.e., $CT = CT'\cup (D.init(\seq{D'}))$.
    	\item $\{(D.init( \seq{D'}) , \sigma(CR'))\}_{\ll} \cup \{(D.init( \seq{D'}) , \sigma(CR'))\}_{\gg} = \emptyset$, $(D.init(\seq{D''}))\in CT'$, for some $\seq{D''}$, then by Lemma~\ref{lem:cCTnCR} $CT$ is constructed. 
    	\item $(D.init(\seq{A}) cond') \in \sigma(CR')$, for some $\seq{A}, cond'$, $(D.init(\seq{D''})) \in  CT'$, for some $\seq{D''}$, then by Lemma~\ref{lem:cCTcCR} $CT$ is constructed.  

\end{enumerate}
From above we have that $fields(D, CT) = D.init(\seq{D'})$, no extends clauses are added to the class table $CT'$, therefore $\projExt{CT} = \Sigma' = \Sigma$\\
%
%
Then \judgeOK{CT}{class\ C\ extends\ D\ \{\seq{C}\ \seq{f}; K\ \seq{M}\}\ OK}{\seq{S}} holds by rule \rulename{T-Class}. \\
The correspondence relation holds because: 
\begin{enumerate}[a)]
	\item We have to show is that $CT\ satisfy\ \sigma(CR)$, and we distinguish the following cases:
	\begin{enumerate}[(1)']
	\item In addition to $(1)$ $CT\ \satisfy\ \sigma(D.init(\seq{D'}))$ by construction, \\
	 $CT'\ \satisfy\ \sigma(CR')$ by $(\ast)$, therefore $CT\ \satisfy\ \sigma(CR')$ by Class Table Weakening Lemma~\ref{lem:CTweak}. As a result $CT\ \satisfy\ \sigma(D.init(\seq{D'})) \cup \sigma(CR')$, i.e., $CT\ \satisfy\ \sigma(CR)$ by definition of $merge_{CR}$.
	\item In addition to $(2)$, $CT'\ \satisfy\ \sigma(CR')$ by $(\ast)$, then there is $CT$, 
	 $CT\ \satisfy \ \sigma(CR)$ by Lemma~\ref{lem:cCTnCR}.	
	\item In addition to $(3)$, $CT'\ \satisfy\ \sigma(CR')$ by $(\ast)$, then there is $CT$,
	 $CT\ \satisfy \ \sigma(CR)$, by Lemma~\ref{lem:cCTcCR}.\\
\end{enumerate} 
\vspace{-0.5em}
   	Class declaration consists in adding the class $C$ with all of its fields, methods, constructor and extend clauses in the class table, whether they are already member of the $CT$ or not. Also, adding these clauses does not affect the satisfaction of the class requirements. Namely resulting class table $CT_r$, such that $C.extends= D\in CT_r, K \in CT_r, \{C.f_i : C_i\}_{i\in [1..n]} \in CT_r, \seq M \in CT_r$, $CT_r \satisfy\ \sigma(CR)$. We distinguish the following cases: 
  	\begin{description}
     \item[$\bullet$] $(C.extends= D) \notin CT,\ or\ (C.init(\seq{C})) \notin CT ,\ or\ \{C.f_i : C_i\}_{i\in [1..n]} \notin CT,\ or\  \{C.m_i: \seq{C}\rightarrow C_0 \}_{i\in [1..n]} \notin CT $ then $CT_r = CT\cup (C.extends= D); (C.init(\seq{C})) \cup \{C.f_i : C_i\}_{i\in [1..n]} \cup \{C.m_i: \seq{C}\rightarrow C_0 \}_{i\in [1..n]}$, and 
     $CT_r \ \satisfy\ \sigma(CR)$ by Lemma \ref{lem:AddCTcls}.
          \item[$\bullet$] $\forall CTcls \in \{(C.extends= D) \cup (C.init(\seq{C})) \cup \{C.f_i : C_i\}_{i\in [1..n]}\cup $
          $\{C.m_i: \seq{C}\rightarrow C_0 \}_{i\in [1..n]} \}$ such that $domCl(CTcls) \in \domC$ then $CT_r = CT$. \\
           Hence, $CT_r\ \satisfy\ \sigma(CR)$.
	\end{description}
\end{enumerate} 
\end{proof}

\newpage

\begin{lemma}
\label{lem:remove-cons}
Given a complete class table $CT$ constructed from all possible class declarations $\seq L$, a set of requirements $CR$, 
$\biguplus_{ L' \in \seq{L}}( \func{removeMs}(CR, L')\uplus \func{removeFs}(CR, L')\uplus \func{removeCtor}(CR, L')
	   \uplus \func{removeExt}(CR, L') ) = CR'|_S $, a substitution $\sigma$, 
	   such that $\sigma(CR)$ is ground, and $CT \satisfy\ \sigma(CR)$. Then $\sigma$ solves $S$.
\end{lemma}
\begin{proof}
  By the definitions of remove for different clauses, $S = S_c \cup S_e \cup S_k \cup S_f \cup S_m$.
 Let us consider constrains generated from field remove $S_f$.
Suppose there exist  $f \in \dom{CR}$ such that $(T.f : T', cond)\in CR$ and $\sigma(cond)\ hold$.
 
 Let $\sigma(T)= C$ and $\sigma(T') = C_f$, since $\sigma(CR)$ ground, $C$, $C_f$ are ground.
 
Since $CT\ satisfy\ \sigma(CR)$, by inversion $field(f,C, CT)= C_f$, i.e, exist $D > C$ such that $D.f:\sigma(T') \in CT$.  
We distinguish two cases when $f$ is declared in $C$ or in one of its superclasses $D$:
\begin{enumerate}[1)]
	\item $D = C$. By rule~\rulename{S-Field}; $C.f : C_f \in CT$. We apply remove for field clause $C.f : C_f$. By definition of removeF the correspondent requirement is $(T.f : T', cond)\in CR$ and the new constraint generated is $S_f = (T' = C_f\ if\ T= C)$. This constraint is solved, because the condition holds and $\sigma(T') = C_f$.
	\item $D > C$. Then there exist $C\ extends\ D \in CT$ and $D. f : C_f \in CT$. In this case we have to apply remove for extends and field clauses. 
	First, we apply remove of extends. By definition of removeExt the requirement under scrutiny is duplicated, i.e., $(T.f : T', cond \cup T\neq C),\ (D.f : T', cond\cup T =C)$.\\
	Second we apply remove of field $f$. By definition of removeFs the generated  constrains are $S_f= \{(T' = C_f\ if\ T = D), (T'= C_f)\ if\ D = D\}$. The first constraint is not valid because the condition does not hold ($\sigma(T) \neq D$), therefore is not considered. the second constraint is solved because the condition holds and $\sigma(T') =C_f$
\end{enumerate}
  
The same procedure we follow for extends, constructors and methods clauses. 
\end{proof}

\begin{lemma}[Class requirements empty] \label{thm:CRempty}
Given a complete class table $CT$ constructed from all possible class declarations $\seq L$, a set of requirements $CR$, 
$\biguplus_{ L' \in \seq{L}}( \func{removeMs}(CR, L')\uplus \func{removeFs}(CR, L')\uplus \func{removeCtor}(CR, L')
	   \uplus \func{removeExt}(CR, L') ) = CR'|_S $, 
and a substitution $\sigma$, such that $\sigma(CR)$ is ground, $\sigma$ solves $S$, we have that if $CT \ satisfy\ \sigma(CR)$, then and $\sigma(CR') = \emptyset$.
\end{lemma}

\begin{proof} By contradiction. \\
 By assumption $\sigma(CR') \neq \emptyset$, and $CR' = \{(CReq, cond)\WHERE \exists (T\neq T') \in cond \}$. From this, follows $\forall (CReq, cond) \in CR'$. $cond$ holds, i.e.,  all conditions of $cond$ do hold. 
This is derived after performing remove, we already know the exact types for classes and their extends, constructor, fields, method clauses. Therefore from remove we add inequalities to invalidate requirements for which we know their exact types, as result exist one their conditions that does not hold. Since by assumption the set is not empty then all conditions hold.
 For sake of brevity we consider only the conditions that are added after performing remove, because are the ones we are interested in.
 
First we consider the extend clauses in the requirement set.
%
 All conditions of the requirements corresponding extends clause do hold. Let us consider $\exists (T.extends T', cond) \in CR'.\ \forall (T \neq C) \in cond.\ \sigma(T) \neq C\ holds$. 
By definition it is given that $\sigma(CR)$ is ground, namely $\sigma(T) = C', \sigma(T') = D'$, such that $C', D'$ are ground. Since all the inequalities in $cond$ hold, this means that in the class table was not added any $extends\ clause$, such that $(C'.extends = D' ) \notin CT$. 
Therefore $CT\ \satisfy\ \sigma(CR)$ does not hold.\\
This strategy of proof is used for constructor since from the definition of $removeCtor$ only inequality conditions are added, and not considered while removing extends clause. 

Second we consider field clauses. From the definition of $removeFs$ and $removeExt$ for every field clause we have a duplicated requirement corresponding to the parents type.
All conditions of the requirements corresponding field clause do hold. By definition it is given that $\sigma(CR)$ is ground, namely $\sigma(T) = C', \sigma(T_f) = C_f$, such that $C', C_f$ are ground. Let us consider $(C'.extends = D) \in CT$, and $\exists (T.f : T_f, cond\cup T\neq C'), (D.f : T_f,  cond'\cup T= C') \in CR' \ such\ that\ \forall (T \neq C), (T = C) \in cond \cup cond' \cup  T= C'\cup  T\neq C'$ $(\sigma(T)\neq C ), (\sigma(T) = C)$ hold. Since all the conditions in $cond \cup cond'$ hold, this means that in the class table was not added any $field\ clause$, such that $f$ is declared in $C'$ or in 
 its parents, i.e., $\forall\ C''\ such\ that\ C'<: C'',\ then\ (C''.f : C_f ) \notin CT$. \\
 Therefore $CT\ \satisfy\ \sigma(CR)$ does not hold.
 

The same strategy of proving is used for methods. In contrast for the optional methods regardless all the conditions might hold they are removed in any case, because they are optional. The lack of inequality conditions that do not hold, only shows the given method is declared in a class of the class table but not in its parents. 

\end{proof}

\begin{theorem}[Equivalence for programs: $\Rightarrow$]  \label{thm:equivP}
Given $\seq L $, if \judgeP{\seq{L}\ OK}, then there exists S, $\Sigma$, $\sigma$, where $\projExt{\seq L} = \Sigma$ and $\func{solve}(\Sigma, S)=\sigma$, such that
\judgeOP{\seq{L}\ OK}{S} holds and $\sigma$ ground solution. 
\end{theorem}

\begin{proof} By induction on the typing judgment.\\
   Case \rulename{T-Program} with \judgeP{\seq{C}\ \seq{L}\ OK}.

 By inversion,  \textit{Class table construction CT} is \\
  $CT = \bigcup_{L'\in\seq{L}}  (\func{addExt}(L')\cup \func{addCtor}(L')\cup \func{addFs}(L') \cup \func{addMs}(L'))$ and \\
 \judgeOP{CT}{\seq{L}\ OK}.
 
 By Theorem~\ref{theo:CR}
,  \cojudgeOC{\seq{L} \ OK}{\seq{S}}{\seq{CR}}, $\forall i \in 1 .. n$. $solve(\projExt{CT}, S_i) = \sigma_i$, $\sigma_i(CR_i)$ is ground and the correspondence relation holds, i.e., $CT \ satisfy\  \sigma_i(CR_i)$. \\

Let $CR_{|S_{cr}} = merge_{CR}(CR_1, \ldots, CR_n) $, $\biguplus_{ L' \in \seq{L}}( \func{removeMs}(CR, L')\uplus$ \\
$ \func{removeFs}(CR, L')\uplus \func{removeCtor}(CR, L')
	   \uplus \func{removeExt}(CR, L') ) = CR_f|_{S_r} $, \\
$S = \seq{S} \cup S_{cr} \cup S_r$, and $\sigma = \{\sigma_i\}_{i\in [1..n] }$.
  From the Lemma ~\ref{thm:CRempty} we have $\sigma(CR_f) = \emptyset$.\\\\
Then \judgeOP{\seq{L}\ OK}{S} holds by rule \rulename{TC-Program}.\\
$\sigma$ solves $\seq{S}$, and $S_{cr}$ as shown below:
\begin{itemize}
	\item $\sigma$ solves $\seq{S}$ because $\sigma = \{\sigma_i\}_{i\in [1..n] }$.
	\item $\forall i \in [1..n].\ \sigma_i(CR)_i$ are ground by Theorem~\ref{theo:CR}.\\
 	$\forall i \in [1..n].\ CT\ \satisfy\ \sigma_i(CR_i)$ by Theorem~\ref{theo:CR}.\\
As a result $\sigma$ solves $S_{cr}$ by Lemma~\ref{lem:mergeCR-cons}.
    \item $\sigma(CR)$ is ground and $CT\ \satisfy\ \sigma(CR)$ by Theorem~\ref{theo:CR}, and given the class table $CT$, then $\sigma$ solves $S_r$ by Lemma~\ref{lem:remove-cons}.
\end{itemize}
$\sigma$ is ground solution because: 
\begin{enumerate}[1)]
	\item $\forall i \in [1..n]$. $\sigma(\projExt{CT}, CR_i)$ is ground because 
    $\sigma(CR_i) = (\{\sigma_i\}_{i\in [1..n] })$ \\ $(CR_i) = $
    $(\{\sigma_j\}_{j\in [1..i-1, i+1..n] }\circ \sigma_i )(CR_i) $ by Corollary~\ref{cor:assoc}.\\
  $( \{\sigma_j\}_{j\in [1..i-1, i+1..n ]})(\sigma_i(CR_i))= \sigma_i(CR_i)$ because $\sigma_i(CR_i)$ is ground by Theorem~\ref{theo:CR}.
  As a result $\sigma(CR)$ is ground by definition of $merge_{CR}$.
\end{enumerate}
\end{proof}

\begin{lemma}[Class table satisfy class requirements] \label{thm:CRsat}
Given class declarations $\seq L$, such that $CT = \bigcup_{L'\in\seq{L}}  (\func{addExt}(L')\cup \func{addCtor}(L')\cup \func{addFs}(L') \cup \func{addMs}(L'))$, a set of requirements $CR$, $\biguplus_{ L' \in \seq{L}}( \func{removeMs}(CR, L')\uplus \func{removeFs}(CR, L')\uplus \func{removeCtor}(CR, L')
	   \uplus \func{removeExt}(CR, L') ) = CR'|_S $, and a substitution $\sigma$, such that $\sigma(CR)$ is ground, $\sigma$ solves $S$,
 we have that if $\sigma(CR') = \emptyset$, then $CT \ satisfy\ \sigma(CR)$.
\end{lemma}
\begin{proof} By contradiction. \\
	 By assumption $CT\ \satisfy \ \sigma(CR)$ does not hold. From this, follows $\exists (CReq, cond) \in CR$. $cond$ holds, i.e.,  all conditions of $cond$ do hold and no compatible clause with $CReq$ exists in $CT$.\\
As property of remove we add inequalities to invalidate requirements for which we know their exact types, as result exist at least one inequality condition that does not hold, and the requirement is removed, otherwise it remains in the requirements set. 

First we consider the extend clauses in the requirements set.
 Let us consider \\ $\exists (T.extends T', cond) \in CR\ such\ that\ cond\ hold$. By definition $\sigma(CR)$ is ground, namely $\sigma(T) = C', \sigma(T') = D'$. By assumption $(C'.extends : D') \notin CT$, i.e., the clause it is not member of any of the class declarations $\seq L$ that are used to realize removing. Therefore after performing remove $\nexists$  $\sigma(T)\neq C \in \sigma(cond')$ such that $\ \sigma(T)\neq C\ does\ not\ hold$, where $(T.extends : T', cond')\in CR'$, i.e., $\sigma(cond')$ hold. \\
Therefore $\sigma(CR') \neq\emptyset$.
  
This strategy of proof is used for constructor since from the definition of $removeCtor$ only inequality conditions are added, and not considered while removing extends clause. 

Second we consider field clauses. From the definition of $removeFs$ and $removeExt$ for every field clause we have a duplicated requirement corresponding to the parents type.\\
Let us consider $\exists (T.f : T_f, cond) \in CR.\ cond hold$. By definition $\sigma(CR)$ is ground, namely $\sigma(T) = C'$, $\sigma(T_f)= D'$. By assumption  $\nexists (D.f : D') \in CT$, such that $\sigma(T) <: D$. This means that in the class table was not added any $field\ clause$, such that $f$ is declared in $C'$ or in its parents. Therefore after performing remove 
$(T.f : T_f, cond') \in CR'$ we have that $\nexists (\sigma(T) \neq C) \in \sigma(cond').\ (\sigma(T) \neq C)\ does\ not\ hold$. i.e, $\sigma(cond')$ hold. \\
Therefore $\sigma(CR) \neq \emptyset$.

The same strategy of proving is used for methods.
\end{proof}

\begin{theorem}[Equivalence for programs: $\Leftarrow$]  \label{thm:equivPl}
Given $\seq L $, if \judgeOP{\seq{L}\ OK}{S}, $\func{solve}(\Sigma, S)=\sigma $, where $\projExt{\seq L} =\Sigma$, and $\sigma$ is a ground solution, then  \judgeP{ \seq{L}\ OK} holds.
\end{theorem}
\begin{proof} By induction on the typing judgment.\\
   Case \rulename{TC-Program} with \judgeOP{ \seq{L}\ OK}{S}.
   
   Let $S = \seq{S} \cup S_{cr} \cup S_r$, $\sigma$ is ground solutions and $solve(\projExt{\bar{L}}, S )=\sigma$, i.e., $\sigma$ solves $\seq S$, $S_{cr}$, $S_r$.

 By inversion, \cojudgeOC{\seq{L} \ OK}{\seq{S}}{\seq{CR}}, $\forall i\in 1\ldots n$. $\sigma(CR_i)$ are ground.\\
  $ CR|_{\typectxcolor{S_{c}}} = merge_{CR}(CR_1, \ldots, CR_n)$. \\
   $\biguplus_{ L' \in \seq{L}}( \func{removeMs}(CR, L')\uplus \func{removeFs}(CR, L') \uplus \func{removeCtor}(CR, L')
   \uplus \func{removeExt}(CR, L') ) = CR_f|_{S_r}$, and $\sigma(CR_f)= \emptyset$

 By Theorem~\ref{theo:CL}
, \judgeOP{CT}{\seq{L}\ OK}, and the correspondence relation holds, i.e., $\forall i\in [1..n]$. $CT \ satisfy\  \sigma(CR_i)$. 
 $CT\ \satisfy \ \sigma(CR_1)\cup\ldots \cup \sigma(CR_n)$, i.e., $CT\ \satisfy\ \sigma(CR)$ by Lemma~\ref{lem:satisfy}.

\textit{Class table construction CT} is $CT = \bigcup_{L'\in\seq{L}}  (\func{addExt}(L')\cup \func{addCtor}(L')\cup \func{addFs}(L') \cup \func{addMs}(L'))$ by Lemma ~\ref{thm:CRsat}.\\\\
Then \judgeP{\seq{L}\ OK} holds by rule \rulename{T-Program}.
\end{proof}

%% file: projExt.tex
\begin{figure*}[h]
  \raggedright
  {
  \begin{gather*}
  \inference   
   {}
  {\projExt{\emptyset}=\emptyset}
  \hskip1em
  \inference
    {}
  {\projExt{C\ \keyw{extends}\ D} =(C, D)}
  \hskip1em
      \inference
        {}
        {\projExt{C.f: C'} =\emptyset}
\\[2ex]
      \inference
        {}
        {\projExt{C.m(): \seq C \rightarrow C'} =\emptyset}
   \hskip1em 
      \inference
        {}
        {\projExt{C.init(\seq C)} =\emptyset}
    \\[2ex]
      \inference
        {\projExt{CTcls_1}=\Sigma_1 & \projExt{CTcls_2} =\Sigma_2}
        {\projExt{CTcls_1 \cup CTcls_2} =\Sigma_1 \cup \Sigma_2}
     \end{gather*}
  }
  \caption{Projection of Class Table to Extends.}
  \label{fig:proj}
\end{figure*}
\begin{figure*}[h]
  \raggedright
  {
  \begin{gather*}
  \inference   
   {}
  {\projExt{\emptyset}=\emptyset}
  \hskip1em
  \inference
    {}
  {\projExt{\keyw{class}\ C\ \keyw{extends}\ D\ \{ \overline{C}\ \overline{f};\ K\ \overline{M}\}} =(C, D)}
      \\[2ex]
      \inference
        {\projExt{L_1}=\Sigma_1 & \projExt{L_2} =\Sigma_2}
        {\projExt{L_1; 	L_2} =\Sigma_1 \cup \Sigma_2}
     \end{gather*}
  }
  \caption{Projection of Class Declarations to Extends.}
  \label{fig:projL}
\end{figure*}



%% file: satisfy-judg.tex
\begin{figure}[t]
  \raggedright
  {
  \begin{gather*}
  \inference[\rulename{field-lookup}]
    {C.f_i : C_i \in fields(C, CT)}
  {field(f_i, C, CT) = C_i}
  \hskip1em
  \inference[\rulename{extends}]
  {(C.extends = D) \in CT}
  {extends(C, CT) = D}
  \\[2ex]
      \inference[\rulename{S-Extend}]
        {(C.extends = D) \in CT}
        {CT \ satisfy\ (C.extends : D, cond)}
\\[2ex]
      \inference[\rulename{S-Constructor}]
        {fields(C, CT) = \overline{C.f: C_f}} 
        {CT \ satisfy\ (C.init(\seq{C_f}), cond)}
\\[2ex]
      \inference[\rulename{S-Field}]
        {field(f, C, CT) = C'}
         {CT \ satisfy\  (C.f : C', cond)}
\\[2ex]
  \inference[\rulename{S-Method}]
        { if\ mtype(m, C, CT) = \seq{C}\rightarrow C'}
         {CT \ satisfy\  (C.m : \seq{C}\rightarrow C', cond)}
\\[2ex]
 \inference[\rulename{Satisfy}]
{ (cond \ hold \Rightarrow CT\ satisfy\ (CReq, cond) ) & \forall (CReq, cond) \in CR }
  {CT \ \textit{satisfy } CR}
  \end{gather*}
  }
  \caption{Judgment for Satisfy.}
  \label{fig:satisfy}
\end{figure}



%% file: document.bbl
\begin{thebibliography}{10}

\bibitem{AnconaDDZ05}
Davide Ancona, Ferruccio Damiani, Sophia Drossopoulou, and Elena Zucca.
\newblock Polymorphic bytecode: Compositional compilation for {J}ava-like
  languages.
\newblock In {\em Proceedings of Symposium on Principles of Programming
  Languages (POPL)}, 2005.
\newblock \href {http://dx.doi.org/10.1145/1040305.1040308}
  {\path{doi:10.1145/1040305.1040308}}.

\bibitem{AnconaZ04}
Davide Ancona and Elena Zucca.
\newblock Principal typings for {J}ava-like languages.
\newblock In {\em Proceedings of Symposium on Principles of Programming
  Languages (POPL)}, 2004.
\newblock \href {http://dx.doi.org/10.1145/964001.964027}
  {\path{doi:10.1145/964001.964027}}.

\bibitem{Chitil01}
Olaf Chitil.
\newblock Compositional explanation of types and algorithmic debugging of type
  errors.
\newblock In {\em Proceedings of International Conference on Functional
  Programming (ICFP)}, 2001.
\newblock \href {http://dx.doi.org/10.1145/507635.507659}
  {\path{doi:10.1145/507635.507659}}.

\bibitem{Christiansen13}
David~Raymond Christiansen.
\newblock Bidirectional typing rules: A tutorial, 2013.

\bibitem{Dunfield13}
Joshua Dunfield and Neelakantan~R. Krishnaswami.
\newblock Complete and easy bidirectional typechecking for higher-rank
  polymorphism.
\newblock In {\em Proceedings of International Conference on Functional
  Programming (ICFP)}, 2013.
\newblock \href {http://dx.doi.org/10.1145/2500365.2500582}
  {\path{doi:10.1145/2500365.2500582}}.

\bibitem{Erdweg15}
Sebastian Erdweg, Oliver Bra\v{c}evac, Edlira Kuci, Matthias Krebs, and Mira
  Mezini.
\newblock A co-contextual formulation of type rules and its application to
  incremental type checking.
\newblock In {\em Proceedings of Conference on Object-Oriented Programming,
  Systems, Languages, and Applications (OOPSLA)}, 2015.
\newblock \href {http://dx.doi.org/10.1145/2814270.2814277}
  {\path{doi:10.1145/2814270.2814277}}.

\bibitem{Hammer14}
Matthew~A. Hammer, Khoo~Yit Phang, Michael Hicks, and Jeffrey~S. Foster.
\newblock Adapton: Composable, demand-driven incremental computation.
\newblock In {\em Proceedings of Conference on Programming Language Design and
  Implementation (PLDI)}, 2014.
\newblock \href {http://dx.doi.org/10.1145/2666356.2594324}
  {\path{doi:10.1145/2666356.2594324}}.

\bibitem{Igarashi01}
Atsushi Igarashi, Benjamin~C. Pierce, and Philip Wadler.
\newblock Featherweight {J}ava: A minimal core calculus for {J}ava and {GJ}.
\newblock {\em Transactions on Programming Languages and Systems (TOPLAS)},
  2001.
\newblock \href {http://dx.doi.org/10.1145/503502.503505}
  {\path{doi:10.1145/503502.503505}}.

\bibitem{Jim96}
Trevor Jim.
\newblock What are principal typings and what are they good for?
\newblock In {\em Proceedings of Symposium on Principles of Programming
  Languages (POPL)}, 1996.
\newblock \href {http://dx.doi.org/10.1145/237721.237728}
  {\path{doi:10.1145/237721.237728}}.

\bibitem{pierce2002types}
Benjamin~C Pierce.
\newblock {\em Types and programming languages}.
\newblock MIT press, 2002.

\bibitem{PierceT00}
Benjamin~C. Pierce and David~N. Turner.
\newblock Local type inference.
\newblock In {\em Proceedings of Symposium on Principles of Programming
  Languages (POPL)}, 1998.
\newblock \href {http://dx.doi.org/10.1145/268946.268967}
  {\path{doi:10.1145/268946.268967}}.

\bibitem{Shao93}
Zhong Shao and Andrew~W. Appel.
\newblock Smartest recompilation.
\newblock In {\em Proceedings of Symposium on Principles of Programming
  Languages (POPL)}, 1993.
\newblock \href {http://dx.doi.org/10.1145/158511.158702}
  {\path{doi:10.1145/158511.158702}}.

\bibitem{Tichy86}
Walter~F. Tichy.
\newblock Smart recompilation.
\newblock {\em Transactions on Programming Languages and Systems (TOPLAS)},
  1986.
\newblock \href {http://dx.doi.org/10.1145/5956.5959}
  {\path{doi:10.1145/5956.5959}}.

\bibitem{Tip11}
Frank Tip, Robert~M. Fuhrer, Adam Kie\.{z}un, Michael~D. Ernst, Ittai Balaban,
  and Bjorn De~Sutter.
\newblock Refactoring using type constraints.
\newblock {\em Transactions on Programming Languages and Systems (TOPLAS)},
  2011.
\newblock \href {http://dx.doi.org/10.1145/1961204.1961205}
  {\path{doi:10.1145/1961204.1961205}}.

\bibitem{Tip03}
Frank Tip, Adam Kiezun, and Dirk B\"{a}umer.
\newblock Refactoring for generalization using type constraints.
\newblock In {\em Proceedings of Conference on Object-Oriented Programming,
  Systems, Languages, and Applications (OOPSLA)}, 2003.
\newblock \href {http://dx.doi.org/10.1145/949305.949308}
  {\path{doi:10.1145/949305.949308}}.

\bibitem{Wells02}
J.~B. Wells.
\newblock The essence of principal typings.
\newblock In {\em Proceedings of International Colloquium on Automata,
  Languages and Programming (ICALP)}, 2002.
\newblock \href {http://dx.doi.org/10.1007/3-540-45465-9_78}
  {\path{doi:10.1007/3-540-45465-9_78}}.

\end{thebibliography}
